\documentclass[aoas,preprint]{imsart}

\RequirePackage{amsthm,amsmath,amsfonts,amssymb}
\RequirePackage[authoryear]{natbib}
\RequirePackage[colorlinks,citecolor=blue,urlcolor=blue]{hyperref}
\RequirePackage{graphicx}
\makeatletter
\newcommand\footnoteref[1]{\protected@xdef\@thefnmark{\ref{#1}}\@footnotemark}
\makeatother
\usepackage{graphicx,psfrag,epsf}
\usepackage{indentfirst}
\usepackage{enumerate}
\usepackage{amsthm}
\usepackage{wrapfig}
\usepackage{hyperref}
\usepackage{makecell}
\usepackage{multicol}
\usepackage{multirow}
\usepackage{rotating}
\usepackage{float}
\usepackage{url} 
\usepackage[FIGTOPCAP]{subfigure}
\usepackage{anyfontsize}         
\usepackage{bm}
\usepackage{xcolor}
\usepackage{tabulary}
\usepackage{booktabs}

\startlocaldefs

\newtheorem{theorem}{Theorem}[section]
\newtheorem{lemma}[theorem]{Lemma}
\theoremstyle{remark}
\newtheorem{remark}{Remark}

\endlocaldefs

\begin{document}

\begin{frontmatter}
\title{Solar Radiation Ramping Events Modeling Using Spatio-temporal Point Processes}
\runtitle{Spatio-temporal Processes for Solar Radiation Ramping Events}

\begin{aug}
\author[A]{\fnms{Minghe} \snm{Zhang}\ead[label=e1,mark]{mzhang388@gatech.edu}}\footnote{\label{note1}Equal contribution},
\author[A]{\fnms{Chen} \snm{Xu}\ead[label=e2,mark]{cxu310@gatech.edu}}\footnoteref{note1},
\author[B]{\fnms{Andy} \snm{Sun}\ead[label=e3,mark]{sunx@mit.edu}},
\author[C]{\fnms{Feng} \snm{Qiu}\ead[label=e4,mark]{fqiu@anl.gov}}
\and
\author[A]{\fnms{Yao} \snm{Xie}\ead[label=e5,mark]{yao.xie@gatech.edu}}
\address[A]{H. Milton Stewart School of Industrial and Systems Engineering (ISyE),
Georgia Institute of Technology,
\printead{e1,e2,e5}}
\address[B]{Sloan School of Management and MIT Energy Initiative, Massachusetts Institute of Technology,
\printead{e3}}
\address[C]{Argonne National Lab, \printead{e4}}

\end{aug}

\begin{abstract}

Modeling and predicting solar events, particularly the solar ramping event, is critical for improving situational awareness for solar power generation systems. It has been acknowledged that weather conditions such as temperature, humidity, and cloud density can significantly impact the emergence and position of solar ramping events. As a result, modeling these events with complex spatio-temporal correlations is highly challenging. To tackle the question, we adopt a novel spatio-temporal categorical point process model, which intuitively and effectively addresses correlation and interaction among ramping events. We demonstrate the interpretability and predictive power of our model on extensive real-data experiments.

\end{abstract}

\begin{keyword}
\kwd{spatio-temporal point process}
\kwd{solar radiation ramping events}
\kwd{time series anomaly detection}
\kwd{online prediction}
\end{keyword}

\end{frontmatter}

\section{Introduction}
\label{sec:intro}


Solar power generation is essential to daily life, as its installations are increasingly common in residential and commercial areas \citep{liu2009solar}. The proliferation of PV units in distribution systems fundamentally alters the daily power flow patterns from traditional uni-directional flows to multi-directional and reverse flows \citep{ningegowda2014coupled, kim2018flexible}. As a result, it is essential to increase the awareness of the production status of individual PV units, such as through an accurate prediction of solar generation, to sustainably supply solar power and avoid power outages \citep{huang2012solar}. However, exact radiation prediction is highly challenging due to inherent stochasticity. 

Besides exact radiation prediction, modeling ramping events in solar power systems is also essential. Ramping events \citep{florita2013identifying,cui2015optimized, kamath2010understanding} are one-bit information that represents abnormal events in sequential observations. They can be interpreted as abrupt slope increases or decreases in power generations and typically occur under extreme weather (e.g., rainstorm or hurricane) \citep{rocchetta2015risk}. The power system is vulnerable to such events, which forces the affected units to shut down. Therefore, predicting these events is valuable for distribution operators to take necessary precautions and reduce restoration costs.

In general, the ramping event is also difficult to predict because of (i) spatio-temporal correlation, (ii) non-stationarity, (iii) computational efficiency in online prediction. More precisely, ramping events typically occur in an affected area rather than at an isolated location. In many cases, an event also starts at one place and then propagates to its neighbors with a time delay, leading to spatio-temporal correlation among the events. Thus, it is crucial to collect observations at multiple locations, which correspondingly form the multivariate series. These challenges require us to think (i) how to jointly and effectively capture the spatial and temporal correlations; (ii) how to predict ramping events if probabilistic models are used; (iii) how to make efficient real-time predictions, especially under non-stationary.

This paper addresses the challenges above through the following two contributions. First, we present a spatio-temporal categorical point process, which is originally proposed by \cite{juditsky2020convex}. The model can flexibly capture the spatio-temporal correlations and interactions among binary or categorical ramping events without assuming time-decaying influence. The model parameters are efficiently estimated using convex optimization with theoretical guarantees. In addition, the model can efficiently make online probabilistic ramping event predictions at any location and time. Second, we propose dynamic decision thresholds to address non-stationary when making online predictions of a future ramping event. These dynamic thresholds show improved performance over static ones.

The rest of the paper is organized as follows. We start with a motivating example using real data for solar radiation modeling. Section \ref{sec:model} presents our modeling framework, including the proposed spatial-temporal point process model, estimation procedures, computational complexity, and dynamic thresholds for prediction. Section \ref{sec:theory} presents theoretical guarantees for model estimation. Section \ref{sec:exp} presents a study of US solar radiation data to show our method's interpretability and predictive power. Section \ref{sec:conclusion} concludes the work and discusses future steps. Appendix A contains proofs and Appendix B contains additional experiments.

\subsection{Motivation with real-data example}

\begin{figure}[t]
\begin{center}
  \includegraphics[scale=0.45]{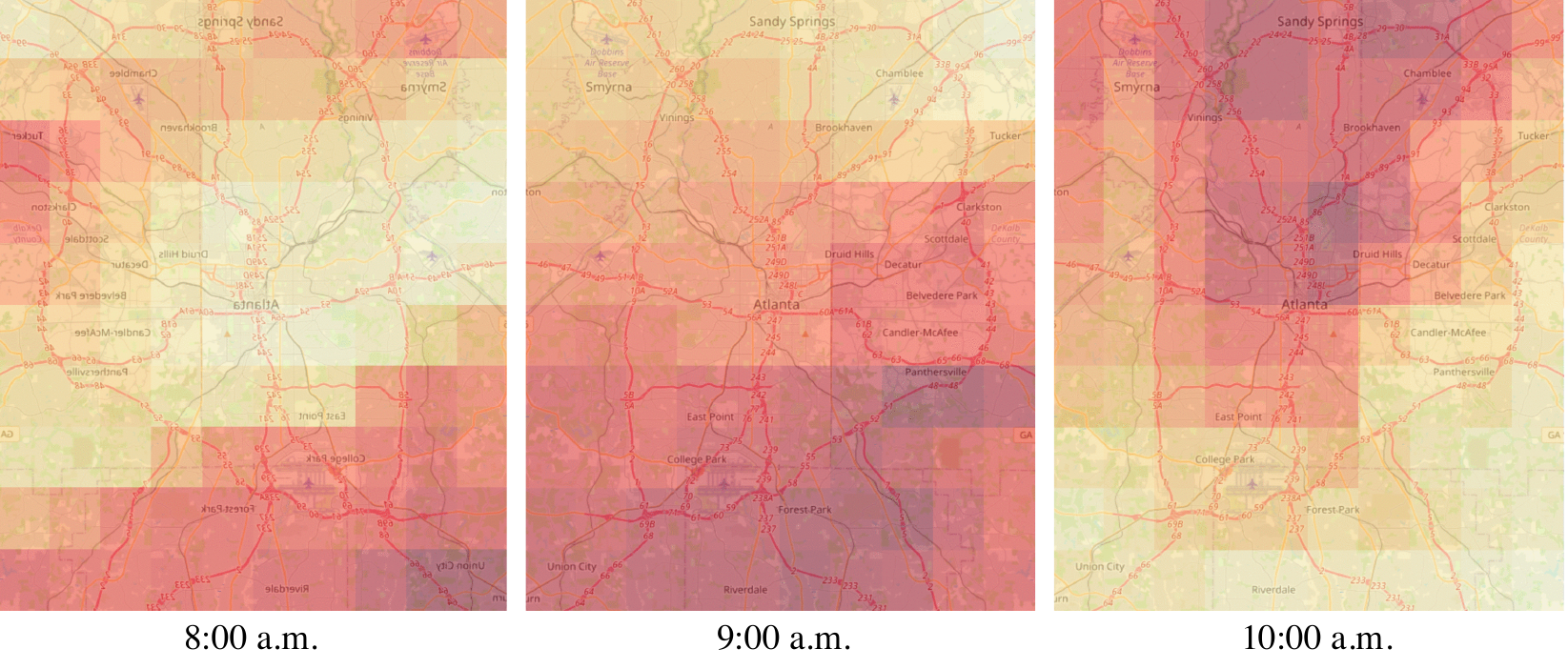}
  \caption {Illustration of the dynamic of solar radiation at different locations around Atlanta. The darker color on the map represents less solar radiation. At 8:00 a.m., South Atlanta has less solar radiation than the average, while two hours later, it is clear that the low-radiation area has moved northward. Such a phenomenon shows that an abnormal event could start from the south but move to the north as time passes.}
  \label{fig:animation}
\end{center}
\end{figure}

\begin{figure}[t]
  \centering
  \includegraphics[width=\linewidth]{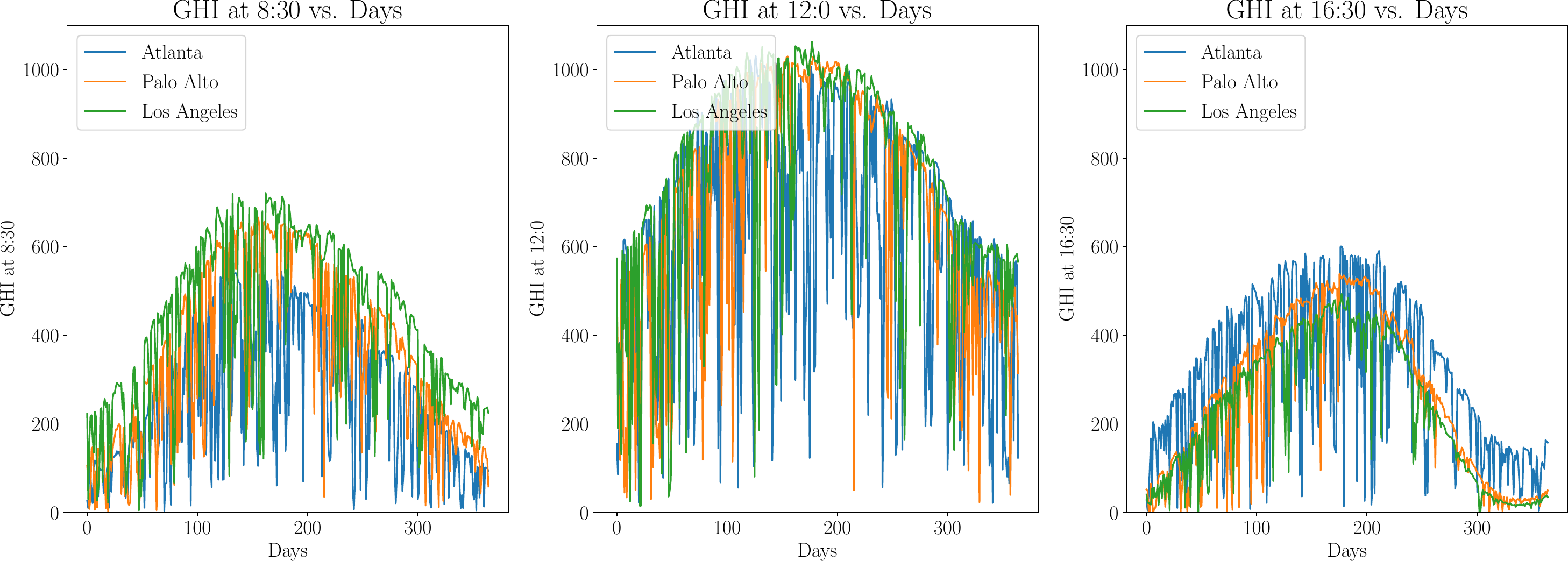}
  \vspace{-0.85cm}
  \caption{2017 Raw data at different hours during the day over 365 days. We observe clear non-stationarity in all series. Observations in Atlanta have lower radiation levels than those in California in the morning (8:30 AM) and at noon (12:00 PM), but have higher values in the late afternoon (4:30 PM).}
  \label{fig:5_1_raw_houly_data}
\end{figure}

We present a motivating example using solar radiation data collected in Atlanta, Palo Alto, and Los Angeles to motivate the correlation among and dynamics of ramping events. Figure \ref{fig:animation} shows areas with different intensities of radiation levels, where darker regions indicate lower radiation. Since ramping events are often defined locally, we can see that they tend to cluster and are spatially and temporally correlated. Thus, it is essential to consider the dynamics of the ramping events in both spatial and temporal aspects to capture and further predict the solar radiation patterns. Moreover, low-radiation areas migrate northward as time passes, indicating non-stationarity in the data. Therefore, such a pattern requires models that jointly capture the non-stationary dynamics and data correlations. Since ramping events are discrete observations in space, point process models can be a natural fit.

Meanwhile, Figure \ref{fig:5_1_raw_houly_data} shows volatile fluctuations in daily solar radiation (measured in GHI) at specific hours of the day, with similar peak values from June to August. 
For example, radiation values in Palo Alto are typically higher than those in Atlanta during the mornings. However, they are typically lower during late afternoons, indicating that the proposed model should be flexible enough to address such stochasticity. 
Lastly, because data come from raw radiation values (e.g., GHI), Figure \ref{fig:abnormal_extraction} illustrates the procedure of ramping event extraction. The top row contains raw solar radiation observations, and the bottom denotes the extracted ramping events. The detailed extraction procedure is explained in Section \ref{sec:exp_data_desribe}.
\begin{figure}[b]
  \centering
  \includegraphics[width=0.75\linewidth]{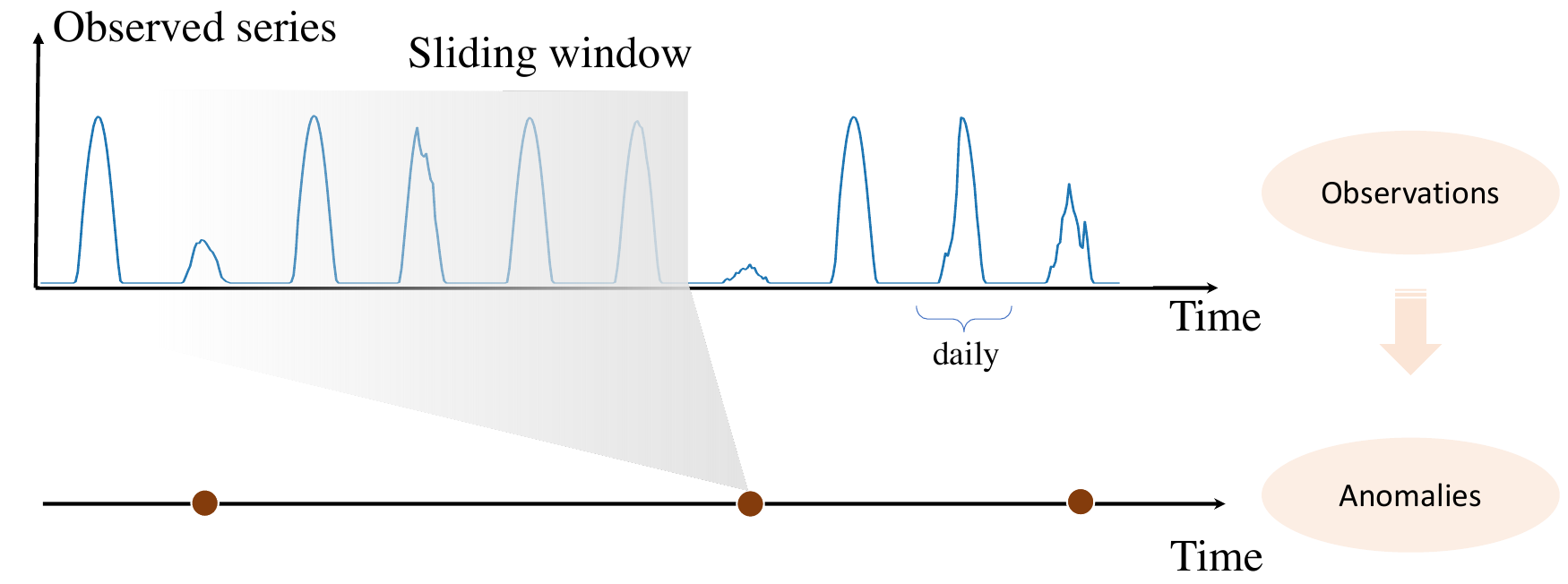}
  \vspace{-0.2cm}
  \caption{Illustration of the procedure for ramping event extraction using raw GHI radiation values: first, the raw data are fed into a sliding window sequentially. Given the set of historical events, we then calculate its $1-\delta$ and $\delta$ quantile; the current event is considered abnormal if it lies beyond or below the $1-\delta$ or $\delta$ quantile.}
  \label{fig:abnormal_extraction}
\end{figure}

\subsection{Related work}

There have been many works on solar ramping event modeling. A line of work started by \cite{sda2013} adopts the Swinging Door Algorithm (SDA) commonly used in data compression, which does not build statistical models. Later, \cite{optsda2015} and \cite{optsda2017} propose an optimized SDA using dynamic programming. However, these works do not have theoretical guarantees and consider no spatio-temporal correlations. Among recent works, \cite{postprocess} uses ensemble-based probabilistic forecasts but needs access to many features besides historical ramping events. In addition, limited data can hamper performance. Meanwhile, \cite{limithistory} proposes a forewarning method using a credal network and imprecise Dirichlet model to study power change by meteorological fluctuation. However, the authors only provide probabilistic forecasts without online categorical predictions. As a result, dynamic solar ramping event modeling remains a great challenge, especially under strong spatio-temporal correlation.

There have also been many advances in stochastic event modeling outside the energy research community, considering data correlation. These works date back to Hawkes process model \citep{hawkes1971spectra} and the self-correcting point process model \citep{isham1979self}. Several recent works provide statistical methods for modeling temporal point processes. \cite{MultivariateTP} proposes a multivariate point process regression model to analyze neural spike trains. \cite{ModelingEP} models event propagation using graph biased temporal point process to leverage structural information from the graph. An earlier application by \citep{Rumor} models rumors prevalence on social media via a log-Gaussian Cov process to learn the underlying temporal probabilistic model. On the other hand, methods using deep learning techniques have also been widely considered. For example, RNN or LSTM cells are used in \cite{du2016recurrent,mei2017neural} to memorize the influence of historical events. A recent work \citep{zhu2020deep} uses attention to model the historical information in the point process. There is also work that uses an imitation learning algorithm \citep{li2018learning} to learn a generator that captures the exact time dependency pattern. While diverse and useful, existing works consider little spatial influence, mainly due to the structural differences between time and space and the computational difficulties in higher dimensions. Recent work by \cite{zhu2020spatio} uses tail-up models for spatial modeling and attention mechanism to capture the time dependency for traffic flows. However, the lack of theoretical guarantee and differences in modeling purpose makes it hard to apply here.

We can also view ramping event modeling as a type of anomaly detection. Research on anomaly detection in the time series analysis area has also been active for decades, which could be traced back to  \cite{grubbs1969procedures}'s work in the 1960s. There are many conventional works on the topic, such as \cite{aggarwal2015outlier,akoglu2015graph,gupta2013outlier}. Also, many works use modern deep learning frameworks to deal with more complex patterns such as \cite{hawkins2002outlier} and more. However, none of these works have a flexible anomaly detection model under strong spatio-temporal correlation for non-stationary data. Lastly, \cite{xu2021conformal} uses novel conformal prediction techniques for anomaly detection in spatio-temporal data, but it is yet unclear how the method can be applied to categorical responses.

\section{Model and Prediction} \label{sec:model}
Let $T$ and $K$ be the total number of ramping observations at each location and the total number of locations. Our observed data are denoted as  $\boldsymbol\omega^T=\{\omega_{tk}, 1\leq t <T, 1\leq k \leq K\}$, where $\omega_{tk}\in \{0,1,\ldots,M\}$ denotes the state of the ramping event (0 means it is a non-ramping event). The value $M$ denotes the number of ramping states. When $M=1$, we are only concerned whether ramping events occur; when $M>1$, we model event types to examine reasons for being ramping events. Therefore, the model can capture ramping events in many situations. Meanwhile, we assume that only past $d$ time units of ramping observations impact the current state; given a memory window $d$, we thus define $\boldsymbol\omega^t_{t-d}=\{\omega_{sk}, t-d \leq s < t, 1\leq k \leq K\}$ as all ramping data in the past $d$ time units from all location. 

We now define the spatio-temporal probabilistic model for ramping events. For each location $k$, we associate an array of \textit{birthrate} parameters $\bar{\beta}_{k}=\left\{\beta_{k}(p), 1 \leq p \leq M\right\}$, and for every pair of locations $k,l$ and every $s\in \{1,\ldots,d\}$, an array of \textit{interaction} parameters $\bar{\beta}_{k \ell}^{s}=\left\{\beta_{k \ell}^{s}(p, q), 1 \leq p \leq M, 0 \leq q \leq M\right\}$. Thus, for each ramping state $p\in \{1,\ldots,M\}$, the conditional probability of $\omega_{tk}$ being in state $p$ on $\boldsymbol\omega^{t}_{t-d}$ is written as
\begin{equation} \label{eq:multi_state_Ber}
    \mathbb{P}[\omega_{tk}=p|\boldsymbol\omega^{t}_{t-d}]
    =\beta_{k}(p)+\sum_{s=1}^{d} \sum_{\ell=1}^{K} \beta_{k \ell}^{s}\left(p, \omega_{(t-s) \ell}\right),
\end{equation}
and $\mathbb{P}[\omega_{tk}=0|\boldsymbol\omega^{t}_{t-d}]=1-\sum_{p=1}^{M}\left[\beta_{k}(p)+\sum_{s=1}^{d} \sum_{\ell=1}^{K} \beta_{k \ell}^{s}\left(p, \omega_{(t-s) \ell}\right)\right]$. 

The model parameters are interpreted and understood as follows: 
\begin{itemize}
    \item $\beta_k(p)$ denotes the intrinsic probability of a ramping event at the $k$-th location under state $p$, without any exogenous influence from the past periods and from other locations. It is also called the birthrate.
    
    \item $\beta^s_{kl}(p,\omega_{(t-s)l})$ denotes the magnitude of influence on a ramping event with possible state $p$ at location $k$ and time $t$ by the event occurred on location $l$ and at time $t-s$. The sum captures the \textit{cumulative influence} up to time $t-1$ from all the other locations.
\end{itemize}
This conditional probability model explicitly captures the dependence of event at time $t$ and location $k$ on events from all locations in the past $d$ days. The set of model parameters are $\boldsymbol\beta=\left\{\bar{\beta}_{k}, \bar{\beta}_{k \ell}^{s}: 1 \leq s \leq d, 1 \leq k, \ell \leq K \right\}$. Define the number of parameters $\kappa:=KM+K^2dM(M+1)$, then $\bm \beta \in \mathbb{R}^\kappa$.
To ensure (\ref{eq:multi_state_Ber}) outputs reasonable probabilities, we require that parameters belong to the following set:
\begin{equation}\label{eq:multi_constraint}
\mathcal X:= \left\{\boldsymbol \beta : \begin{array}{l}
0 \leq \beta_{k}(p)+\sum_{s=1}^{d} \sum_{\ell=1}^{K} \min _{0 \leq q \leq M} \beta_{k \ell}^{s}(p, q), 1 \leq p \leq M, 1 \leq k \leq K, \\
1 \geq \sum_{p=1}^{M} \beta_{k}(p)+\sum_{s=1}^{d} \sum_{\ell=1}^{K} \max _{0 \leq q \leq M} \sum_{p=1}^{M} \beta_{k \ell}^{s}(p, q), 1 \leq k \leq K.
\end{array}\right\}
\end{equation}

\subsection{Model estimation} \label{sec:model_estimations} We introduce both the least-square method and the maximum likelihood estimation technique for model estimation.


\vspace{0.1in}
\noindent \textit{Least-Square (LS) method: }
We estimate the model parameters $\boldsymbol \beta$ by minimizing the sum-of-square errors between conditional probabilities and extracted abnormal events. This leads to the following optimization problem objective. 
\begin{equation} \label{eq:multi_obj}
\Psi_{\bm\omega^N}(\hat{\bm\beta}) = \frac{1}{2N}\sum_{t=1}^N\sum_{k=1}^K \left \Vert \left (\bar{\beta}_{k}+\sum_{s=1}^{d} \sum_{\ell=1}^{K} \bar{\beta}_{k \ell}^{s}\left(\omega_{(t-s) \ell}\right)\right) - \bar{\omega}_{tk}\right \Vert^2_2
\end{equation}
In (\ref{eq:multi_obj}), $\bar{\omega}_{tk}:=e_{\omega_{tk}} \in \mathbb{R}^M$ is the vector of ramping states, whose $i^{th}$ entry is 0 if $\omega_{tk}$ is in state $i > 0$, and is the zero vector if $\omega_{tk}=0$. This vector encodes the true state of the data point at time $t$ and location $k$. Also, $\bar{\beta}_{k}=\{\beta_{k}(p)\in \mathbb{R}, p\in \{1,...,M\} \} \in \mathbb{R}^M$ and $\bar{\beta}_{k \ell}^{s}\left(\omega_{(t-s) \ell}\right) =\{\beta_{k \ell}^{s}\left(p, \omega_{(t-s) \ell}\right) \in \mathbb{R},\ p\in \{1,...,M\} \} \in \mathbb{R}^M$.

Thus, the problem of recovering parameters $\boldsymbol \beta$ is formulated as solving the following constrained convex optimization problem \citep{juditsky2020convex}, subject to constraints (\ref{eq:multi_constraint}):
\begin{equation}
\hat{\boldsymbol \beta}^{\rm{LS}}(\bm\omega ^N) := \arg\min_{\bm\beta \in \mathcal X} \Psi_{\bm\omega^N}(\bm\beta)
\end{equation}
Note this optimization problem has a strongly convex objective function (e.g. the least-square cost function), so it can be solved efficiently in polynomial time.

\vspace{0.1in}
\noindent \textit{Maximum-Likelihood (ML) method: } 
Following \citep{juditsky2020convex}, we additionally assume that for every $t$, random variables $\omega_{tk}$ are conditionally independent across locations $k$ given the past history $\boldsymbol \omega^{t-1}$, so that conditional probabilities are separable. In particular, we define the following likelihood objective function.
%
\begin{equation} \label{eq:multi_state_lik}
    L(\bm\beta)=-\frac{1}{N} \sum_{t=1}^{N} \sum_{k=1}^{K} \psi_{t k}\left(\bar{\beta}_k, \omega^{N}\right),
\end{equation} where 
\begin{align*}
    \psi_{t k}\left(\bar{\beta}_k, \bm\omega^{N}\right)=
    \left\{\begin{array}{ll}
\ln \left(\beta_{k}(p)+ \sum_{s=1}^d\sum_{l=1}^K \beta^s_{kl}(p, \omega_{(t-s)l})\right), & \omega_{t k}=p>0\\
\ln \left(1-\sum_{p=1}^{M}\left(\beta_{k}(p)+ \sum_{s=1}^d\sum_{l=1}^K \beta^s_{kl}(p, \omega_{(t-s)l})\right)\right), & \omega_{t k}=0
\end{array}\right.
\end{align*}

We call a minimizer vector $\hat{\bm\beta}$ of
(\ref{eq:multi_state_lik}), subject to the same constraints 
in  
(\ref{eq:multi_constraint}) 
, the maximum likelihood (ML) estimate:
\[
\hat{\boldsymbol \beta}^{\rm{MLE}}(\bm\omega ^N) = \arg\min_{{\boldsymbol \beta} \in \mathcal X} L (\boldsymbol{\beta})
\]
The objective function is convex since it resembles the likelihood function for a generalized linear model (GLM) with Bernoulli link functions. Thus it can be solved efficiently by convex optimization algorithms as we did for LS estimates.

\begin{remark}[Computational Complexity] The total number of parameters $\kappa$ depend on $K$ (i.e., number of location) and $M$ (i.e., number of states) quadratically. Therefore, even for convex solvers with polynomial complexity, the complexity of solving the problem can be high when these numbers are large. In such a case, sparsity can be imposed in this model to speed up computation. For example, we may assume that a pair of locations do not influence each other if the distance between them exceeds a threshold without losing model convexity.
\end{remark}

\subsection{Prediction procedures}\label{sec:anomaly_detection} 

We first make sequential prediction on future conditional probabilities using the estimated parameters $\hat{\bm\beta}$, data from the past $d$ days, and equation (\ref{eq:multi_state_Ber}). Thus, implicit in our prediction mechanism is the assumption that we can observe future data sequentially after prediction and thus incorporate this new information into prediction. Then, we use the sequentially estimated probability to detect ramping events and their states. More precisely, for each location $k$, future time index $t$, and ramping state $p$, we let $\hat{\omega}_{t,k}=p$ if $\hat{p}_{tk}(p) \geq \tau_{tk}(p)$ and 0 otherwise, where $\hat{p}_{tk}(p)$ (resp. $\tau_{tk}(p)$) is the predicted state probability (resp. state threshold). When multiple possible states exist for $\hat{\omega}_{t,k}$, we pick the state with the highest predicted probability. 

The set of abnormal thresholds $\{\tau_{tk}(p):1\leq p\leq M\}$ can be defined sequentially and dynamically or be fixed in advance. We propose the following dynamic $\tau_{tk}(p)$ for sequential prediction. Denote $w$ as the window of past $w$ time units of observations. Then, the dynamic thresholds $\tau_{tk}(p), p=1,\dots, M$ have the form:
\begin{equation}\label{thres}
  \tau_{tk}(p) = \alpha_p\frac{\sum_{i=1}^{w} \hat p_{(t-i)k}(p) \omega_{(t-i)k}(p)}{\sum_{i=1}^{w} \omega_{(t-i)k}(p)} + (1-\alpha_p) \frac{\sum_{i=1}^{w} \hat p_{(t-i)k}(p) (1-\omega_{(t-i)k}(p))}{w-\sum_{i=1}^{w}\omega_{(t-i)k}(p)}, \quad t\geq w.
\end{equation}
The intuition for choosing this threshold is that we want to find a data-driven threshold that distinguishes between the $p^{\rm{th}}$ ramping state from the rest. We do so by using the history of the ramping events. Meanwhile, if $\alpha_p$ is close to 0, the threshold is likely small since the latter summand in (\ref{thres}) tends to have a larger denominator than its numerator. Lastly, if there is no ramping event with the given state in the past $w$ days, we have $0/0$ in \eqref{thres}. In this case, we will let the dynamic threshold be the static threshold (e.g. $\tau_{tk}(p)=\tau_{k}(p)$ for all $t$ and $\tau_{k}(p)$ is found via cross-validation). 

\vspace{0.1in}
\noindent \textit{Prediction interval:} We can first build bootstrap confidence intervals for parameters $\boldsymbol \beta$ as follows: Let $B$ be the number of bootstrap estimates. For each $b=1,\ldots,B$, we re-sample uniformly with replacement from $\{\omega_{tk}\}$ to create $\{\omega^b_{tk}\}$. Then, we fit $\hat{\boldsymbol \beta^b}$ using our earlier LS or ML models and data $\{\omega^b_{tk}\}$ and calculate standard errors for each component of $\boldsymbol \beta$ using the $B$ bootstrap estimates; for notation simplicity, we write the standard error vector as ${\rm SE}_{\rm boot}$, which has the same dimension as $\boldsymbol \beta$. Finally, we form the $1-\epsilon$ confidence interval for $\boldsymbol \beta$ as $\hat{\boldsymbol \beta} \pm z_{1-\epsilon/2\kappa} {\rm SE}_{\rm boot}$, where $z_{\alpha}$ is the $\alpha$-quantile of a $N(0,1)$ random variable. Note we use the conservative Bonferroni correction here in the multiplier $z_{1-\epsilon/2\kappa}$, because we want to make sure the intervals are valid for all the components of $\boldsymbol \beta$. Then, we build prediction intervals for ramping event probabilities as $[\hat{\boldsymbol \beta} \pm z_{1-\epsilon/2\kappa} {\rm SE}_{\rm boot}]$. We choose $B=100$ and $\epsilon=0.05$ in experiments.

\section{Theory}\label{sec:theory}
We first present the performance guarantee for $\hat{\boldsymbol \beta}^{\rm{LS}}=\hat{\boldsymbol \beta}^{\rm{LS}}(\bm\omega ^N)$, which is proven by \cite{juditsky2020convex}. We then state our contribution for the MLE estimates $\hat{\boldsymbol \beta}^{\rm{MLE}}=\hat{\boldsymbol \beta}^{\rm{MLE}}(\bm\omega ^N)$. 

Define function $\eta(\cdot)$ on the set of all arrays $\bm \omega_{t-d}^{t-1}\in\{0,1,\ldots,M\}^{d\times K}$ that takes values in the matrix space $\mathbb{R}^{K \times \kappa}$:
\begin{equation}
\eta^\top (\bm \omega_{t-d}^{t-1}) = [I_K, I_K\otimes \text{vec}(\bm \omega_{t-d}^{t-1})^\top ]\in \mathbb{R}^{K \times \kappa},
\end{equation}
where $I_K$ is $K$-dimensional identity matrix, $\otimes$ denotes the standard Kronecker product, vec$(\cdot)$ vectorizes a matrix by stacking all columns. 

Next, define $A[\bm\omega^n]:=\frac{1}{N}\sum_{t=1}^N\eta(\bm \omega_{t-d}^{t-1})\eta^\top(\bm \omega_{t-d}^{t-1})$ and the condition number given $A\in \mathbb{R}^{\kappa \times \kappa}, A\succ 0$:
\begin{equation} \label{def:theta}
    \theta_p[A]:=\max\{\theta\geq 0: g^\top Ag\geq \theta||g||_p^2,\forall g\in \mathbb{R}^\kappa, p\in[1,\infty]\}.
\end{equation}
By concentration inequalities for martingales, we have

\begin{theorem}[Bounding $\ell_p$ error of LS estimate \citep{juditsky2020convex}] 
\label{thm:LS-bound}
For every $\epsilon\in(0,1)$, every $\bm \omega ^N$, and any $p\in[1,\infty]$, we have with probability at least $1-\epsilon$ that 
\[
    ||\hat{\boldsymbol \beta}^{\rm{LS}}-\bm\beta||_p\leq \bigg(\sqrt{\frac{\ln(2\kappa/\epsilon)}{2N}}+\frac{\ln(2\kappa/\epsilon)}{3N}\bigg)/\sqrt{\theta_p[A[\bm\omega^N]]\theta_{1}[A[\bm\omega ^N]]}.
\]
\end{theorem}

We also have performance guarantee for $\hat{\boldsymbol \beta}^{\rm{MLE}}=\hat{\boldsymbol \beta}^{\rm{MLE}}(\bm\omega ^N)$. Under the same notation in Theorem \ref{thm:LS-bound}, we have the following guarantee at the same order as above, whose proof is in Appendix A. 

\begin{theorem}[Bounding $\ell_p$ error of MLE estimate]\label{thm:ML-bound}
Suppose we restrict $\beta$ to its $\rho$-strengthened version, where probabilities in (\ref{eq:multi_constraint}) are required to be lower bounded by $\rho$ and upper bounded by $1-\rho$. For every $\epsilon\in(0,1)$, every $\bm \omega ^N$, and any $p\in[1,\infty]$, we have with probability at least $1-\epsilon$ that
 \[||\hat{\boldsymbol \beta}^{\rm{MLE}} -\boldsymbol \beta ||_p \leq \frac{(1-\rho)^2}{\rho}  \sqrt{\frac{2\ln (2 \kappa / \epsilon)}{N}} / \sqrt{\theta_p[A[\bm\omega^N]]\theta_{1}[A[\bm\omega ^N]]}.\]
\end{theorem}

\begin{remark}[Numerical Computation]
The quantity $\theta_p[A[\bm\omega^N]]$ is readily computable when $p=2$ or $\infty$: when $p=2$, it is the minimum eigenvalue and when $p=\infty$, it is $\min _{1 \leq i \leq \kappa}\left\{x^{T} A x:\|x\|_{\infty} \leq 1, x_{i}=1\right\}$, which is the minimum of $\kappa$ efficiently computable quantities. Solving for $p=1$ is hard in general. However, per the discussion in \citep{juditsky2020convex}, one can pass to polars and solve for the semidefinite relaxation upper bound on $\max _{\|x\|_{\infty} \leq 1} x^{T} Q x$ for $Q:=A^{-1}$ by $$
\min _{\lambda}\left\{\sum_{i} \lambda_{i}: \lambda_{i} \geq 0, \forall i ; \operatorname{Diag}\left\{\lambda_{1}, \ldots, \lambda_{\kappa}\right\} \succeq Q\right\},
$$
which is tight within the factor $\pi/2$ \citep{nesterov1998}. Furthermore, we can expect that the minimum eigenvalue of $A[\boldsymbol \omega^n]$ will be of order 1 with high probability, so that the bound using Theorem \ref{thm:LS-bound} or \ref{thm:ML-bound} goes to 0 as $N\rightarrow\infty$ at the rate $O(1/\sqrt{N})$. As a result, $\hat{\boldsymbol \beta} \overset{p}{\to} \boldsymbol \beta$. Lastly, the bound on the RHS in either Theorem \ref{thm:LS-bound} or \ref{thm:ML-bound} is fully data-driven and computable given historical ramping events. We will show numerical values of these bounds under different norms in Section \ref{sec:exp}.
\end{remark}

\section{Real-data study} \label{sec:exp}
We focus on modeling either single-state ramping events (e.g. $\omega_{tk}\in \{0,1\}$) or two-state multi-state ramping events (e.g. $\omega_{tk}\in \{-1,0,1\}$). Two-state ramping events occur when the current radiation is too high or too low compared to historical values. We first introduce the raw data before converting them to ramping events. We only show detailed modeling results for single-state ramping events due to paper length but summarize findings for multi-state event modeling, which are detailed in Appendix B.1.

\subsection{Dataset} 
\label{sec:exp_data_desribe}
We retrieve from the NSRDB\footnote{dataset available at \protect\url{https://nsrdb.nrel.gov/}} website 2017 and 2018 bi-hourly solar radiation data collected every day from sensors across different states and cities in the US. Sensors can be located using longitudes and latitudes. The 2017 data are used for estimating parameters in our model, and the 2018 data are used for sequential prediction and model performance assessment. Specifically, we collect three sets of such data: two uniform 3-by-3 grids of data, which center at downtown Atlanta or at downtown Los Angeles and each pair of grids points (i.e., sensors) with the same longitude and latitude are separated by 0.076-degree (5 miles), and a non-uniform set of 10 downtown data across different cities\footnote{Complete list of cities: Fremont, Milpitas, Mountain View, North San Jose, Palo Alto, Redwood City, San Mateo, Santa Clara,
South San Jose, Sunnyvale} in California. We choose these places because of the weather differences to demonstrate our model's flexibility. We visualize the two sets of sensor locations on the map in Figure \ref{fig:5_1_mapview}.

\begin{figure}[t]
  \centering
  \subfigure[]{\includegraphics[scale=0.37]{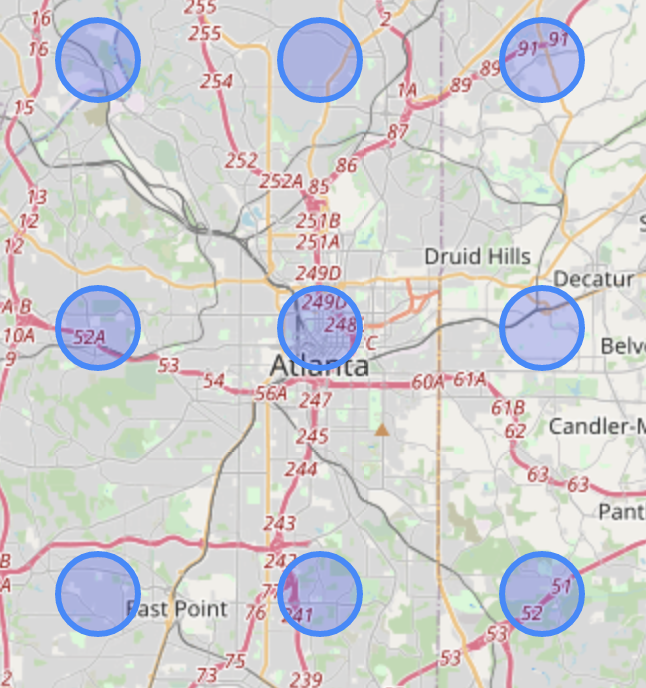}}
  \subfigure[]{\includegraphics[scale=0.37]{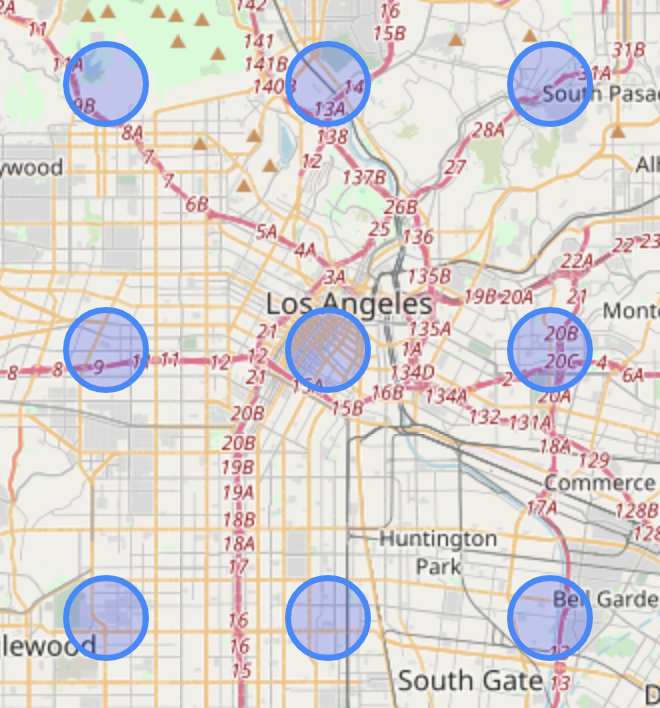}}
  \subfigure[]{\includegraphics[scale=0.30]{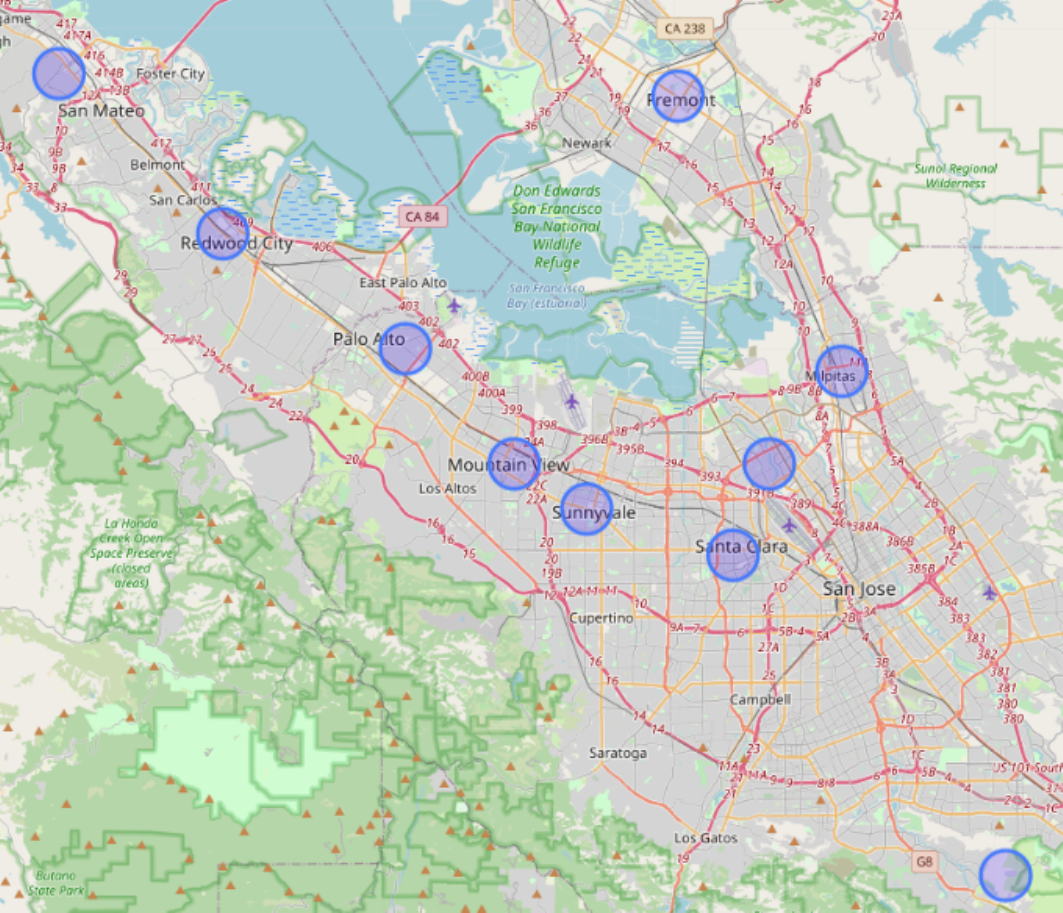}}
  \caption{Locations where the 2017-18 raw radiation data are collected: (a) 9 locations near downtown Atlanta, (b) 9 locations around downtown Los Angeles, and (c) 10 different downtown in North California. Each pair of locations along the exact longitude or magnitude in (a) and (b) are 0.076 degrees (5 miles) apart.}
  \label{fig:5_1_mapview}
\end{figure}

We also briefly describe the structure of these data: there are 48 \textit{raw radiation} recordings at a location on any given day, so two consecutive recordings are 30 minutes apart. Each bi-hourly recording contains a row of radiation values, including Global Horizontal Irradiance (GHI), Direct Normal Irradiance (DNI), and Diffuse Horizontal Irradiance (DHI) all of which reflect radiation levels. Other features such as wind speed, temperature, and solar zenith angle are available. Since the GHI is correlated with the other two metrics as $\text{GHI} =\text{DNI} \cdot \cos(\theta) + \text{DHI}$, where $\theta$ is the solar zenith angle, we decide to only model ramping events based on GHI values.

Recall that our point process model in Section \ref{sec:model} only takes categorical data as input. Therefore, we preprocess the GHI data to extract the ramping events as follows; Figure \ref{fig:abnormal_extraction} earlier illustrates the process. Let $\{x_{tk}\in\mathbb{R}: t=1,\dots, nT; k=1,\dots K\}$ be a matrix of GHI observations, where $T$ denotes the time horizon (number of days), $n$ denotes the number of observations each day per location and $K$ denotes the total number of locations. The value of $K$ is 9 for Atlanta and Los Angeles and 10 for California. For any fixed $k$, we also denote the past $w_1$-days' observations as ${\bm x}^{t-1}_{t-w_1,k}:=\{{x}_{t'k}: t'=(t-w_1) n,(t-w_1)n+1,\dots, nt-1\}$. Then, ramping events $\omega_{tk}$ are defined locally. For single-state events, we define $\omega_{tk}=1, t\in \{1,\ldots,T\}$ if $\alpha \%$ of $\{{x}_{t’k}, t’ \in \{nt, \ldots,(n+1)T\}\}$ in day $t$  is either above the $1-\delta$ empirical quantile or below the $\delta$ empirical quantile of ${\bm x}^{t-1}_{t-w_1,k}$. Intuitively, we think there is a ramping event in day $t$ if a reasonable amount of GHI values in day $t$ are too high or too low comparing to historical values. Similarly, for two-state multi-state events, we let $\omega_{tk}$ be -1 (resp. 1) if $\alpha \%$ of GHI in day $t$ are below (resp. above) the $\delta$ empirical quantile of ${\bm x}^{t-1}_{t-w_1,k}$. Based on experience with solar data, $w_1=30$, $\delta=0.0005 (0.05\%)$, and $\alpha=2/48$ were chosen as sensible choices. We did not purposefully tune these parameters to make the ramping event detection easy.

\subsection{Evaluation metric and procedure}\label{exp:metric}
We adopt standard performance metrics for classification, including precision, recall, and F1 score. We do so because ramping event detection can be viewed as binary or categorical classification, where we want to detect as many actual ramping events sequentially and timely as possible. In the binary-state case, define the set of all actual ramping events as $A$ and the set of predicted ones as $B$. Then precision $P$ and recall $R$ are defined as:
\[P=|A \cap B| /|B|, R=|A \cap B| /|A|,\]
where $|\cdot|$ denotes set cardinality. The $F_1$ score combines $P$ and $R$: $F_{1}=2 P R /(P+R)$ and a higher $F_1$ score indicates better precision and recall, as well as a balance between them. Because ramping events are extreme cases in data (e.g., positive cases in data are rare), we do not use the ROC curve (true positive rate versus false-positive rate) in our setting. For multi-state ramping events, we would compute these scores separately but use the same formula for each state.

The evaluation procedure is as follows. We first use 2017 data in each city (e.g., Atlanta, Los Angeles, or California) to train their respective estimated parameters. We apply the point process model \textit{separately} for each city because ramping events have a heterogenous occurring mechanism that is region-dependent. We note that no randomization is involved in the whole process, so estimation only needs to be performed once. Then, we make a sequential prediction of the ramping probability at each $t$ and $k$ using our model, estimated parameters, and historical data. The probabilities are converted to actual states via the dynamic threshold introduced in Section \ref{sec:anomaly_detection}. If the predicted state on the day $t$ and location $k$ align with the actual state, we consider it to be a true positive (e.g., belong to $|A\cap B|$ in Section \ref{exp:metric}). 

\subsection{Choice of hyperparameters}
We show how we tune the best static threshold in Figure \ref{fig:5.2:accu_metric}, which will be used to detect ramping events at all locations in a given city. We set the first 30$\%$ training data for tuning the parameter and the rest $70\%$ for testing; cross-validation is not used because of sequential data dependency. To find the best static threshold $\tau$ for prediction described in Section \ref{sec:anomaly_detection}, we search over a grid of 25 uniformly spaced $\tau$ in $[0,1]$ and use the one that yields the highest $F1$ score. From Figure \ref{fig:5.2:accu_metric}, we can see that the $F_1$ score curves generally have an inverse-$U$ shape and that the best static thresholds for all three cities are between 0.45-0.6. Such results indicate we need a high amount of evidence (i.e., non-zero ramping probability) to quantify any day as a ramping event. Hence, the model tends to have a low false-positive rate.

Regarding other parameters, we realize that large choices of the memory depth $d$ used in model (\ref{eq:multi_state_Ber}) dramatically decrease the magnitude of spatial-temporal parameters (e.g. $\beta_{kl}^s(p,q)$ for large $s$); we chose $d=10$ since the value declines to around 0 when $s>5$. As for the dynamic threshold parameters in (\ref{thres}), we chose the best $w$ among 25 uniformly spaced $w$ from 10 to 110 and $\alpha_p$ from 0.1 to 0.9 for each state $p$. The ones reaching the highest $F_1$ scores are $\alpha_p\equiv\alpha=0.75$ and $w=50$, as more emphasis are placed on true predictions and window size should be large enough to avoid encountering too few ramping events with non-zero states.

\begin{figure}[t]
    \centering
    \vspace{-0.5cm}
    \subfigure[Downtown Atlanta]{\includegraphics[scale=0.33]{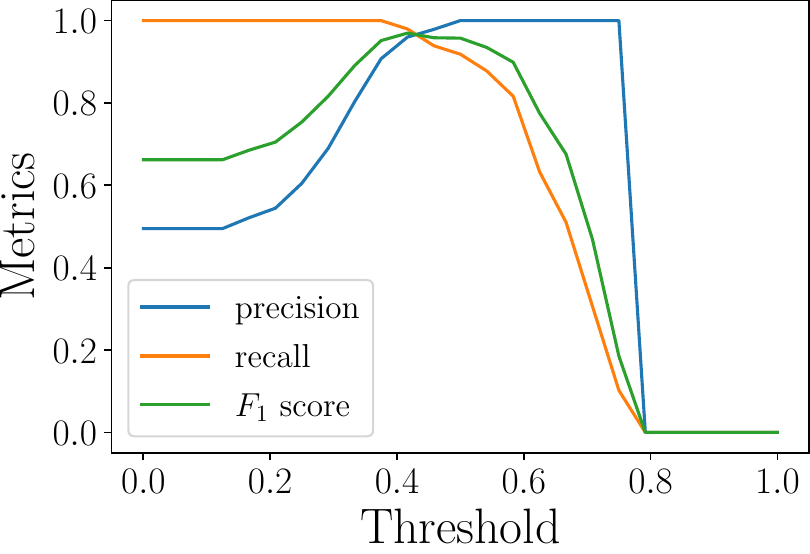}}
    \subfigure[Downtown Los Angeles]{\includegraphics[scale=0.33]{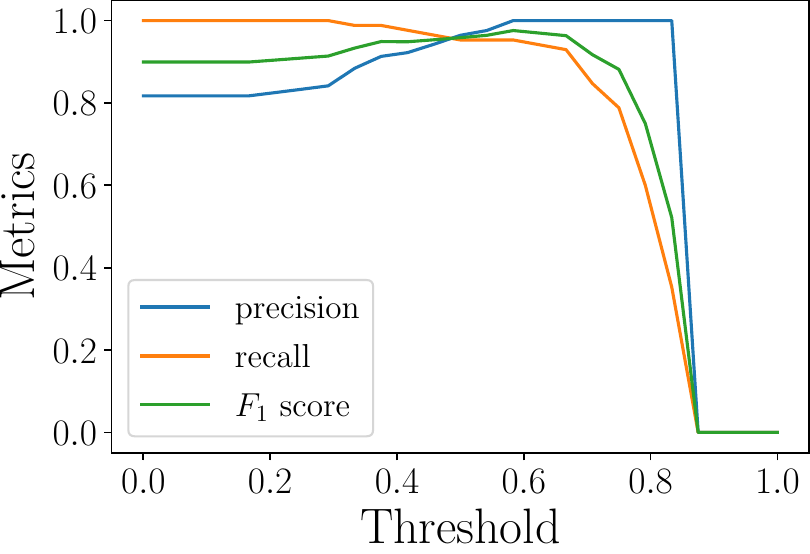}}
    \subfigure[Palo Alto]{\includegraphics[scale=0.33]{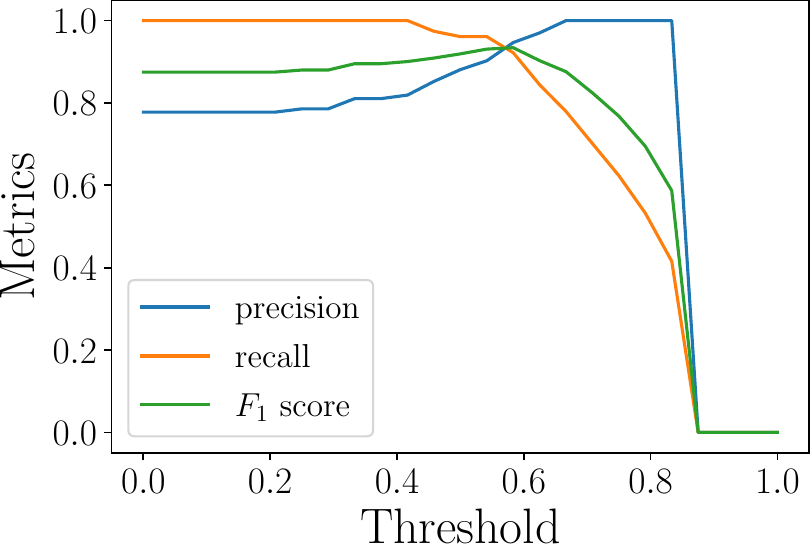}}
    \vspace{-0.3cm}
    \caption{Accuracy metrics vs. static threshold at different locations for single-state detection. Static thresholds that lead to the highest $F_1$ scores on the first $30\%$ data are used during prediction.}
    \label{fig:5.2:accu_metric}
\end{figure}

\subsection{Parameter estimation results}
\label{sec:single-state}
We first visualize the estimated birthrate and interaction parameters in Figure \ref{fig:5_2_beta_single}, which provide interpretable insights into correlation among locations. We then overlay these parameters on terrain maps in Figure \ref{fig:5_2_grid_single} to illustrate certain spatio-temporary patterns uncovered by our point-process model. We lastly provide in Table \ref{tab:performance_guarantee} data-driven bounds on parameter estimates based on Theorem \ref{thm:LS-bound}.

\vspace{0.1in}
\noindent \textit{(a) Recovered spatio-temporal influence: } Figure \ref{fig:5_2_beta_single} shows the recovered single-state birthrate parameters $\beta_k$ over different locations by MLE and LS, and shows a selected number of interaction parameters $\beta^s_{kl}$ over time. Bootstrap confidence intervals are plotted around the estimates. The figures show that the magnitude of birthrate estimates is similar between LA and North California (CA), whereas birthrates in Atlanta (ATL) are the smallest among all three. We suspect the difference occurs because California weather is generally shinier than Atlanta's, leading to higher radiation levels in the former region. Estimates of the interaction parameters between different sensors decay fast regardless of location, indicating that location-to-location influences do not persist over time. In practice, this indicates that ramping events that happened several times ago at any location cannot noticeably influence the current ramping probabilities. Lastly, estimates recovered by MLE and LS generally have very similar magnitude.
\begin{figure}[b]
  \centering
  \subfigure[Atlanta]{\includegraphics[scale=0.22]{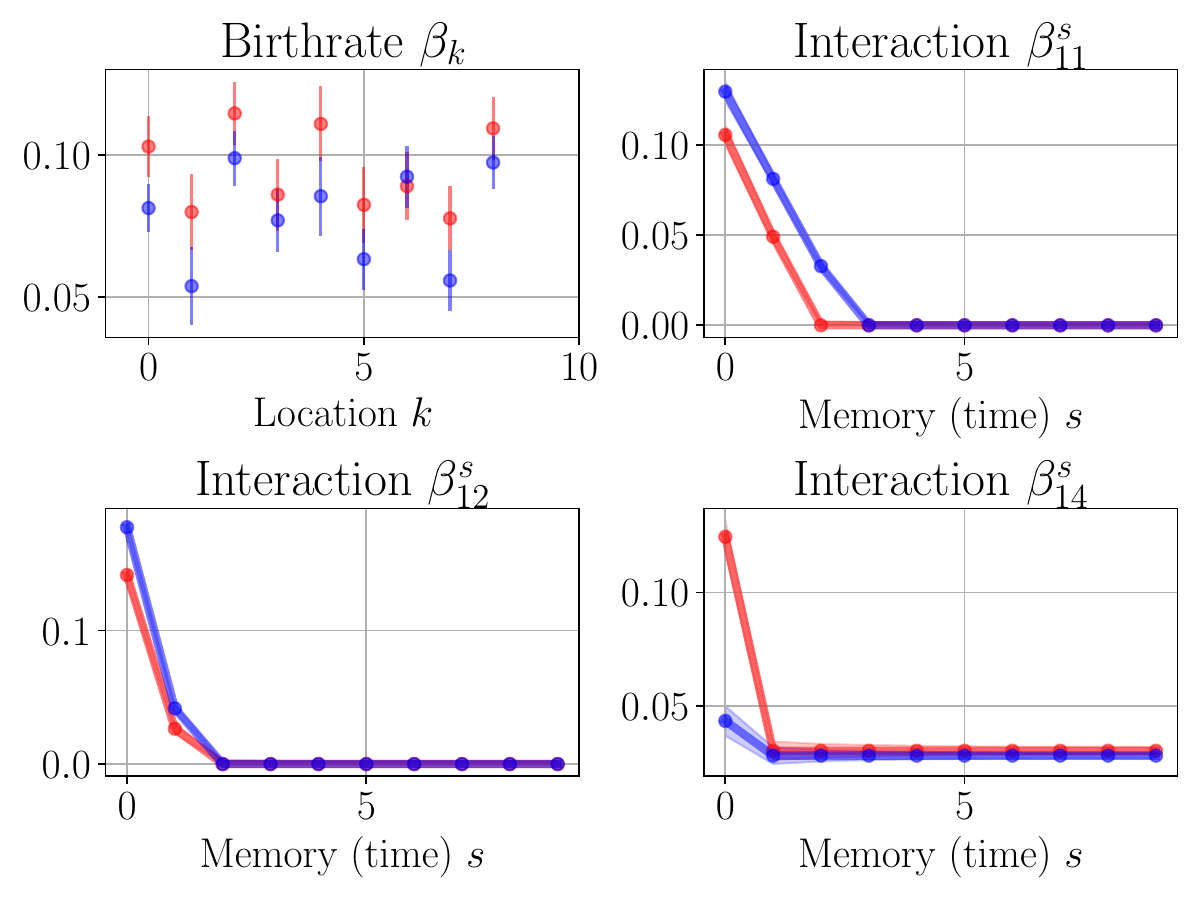}}
   \subfigure[Los Angeles]{\includegraphics[scale=0.22]{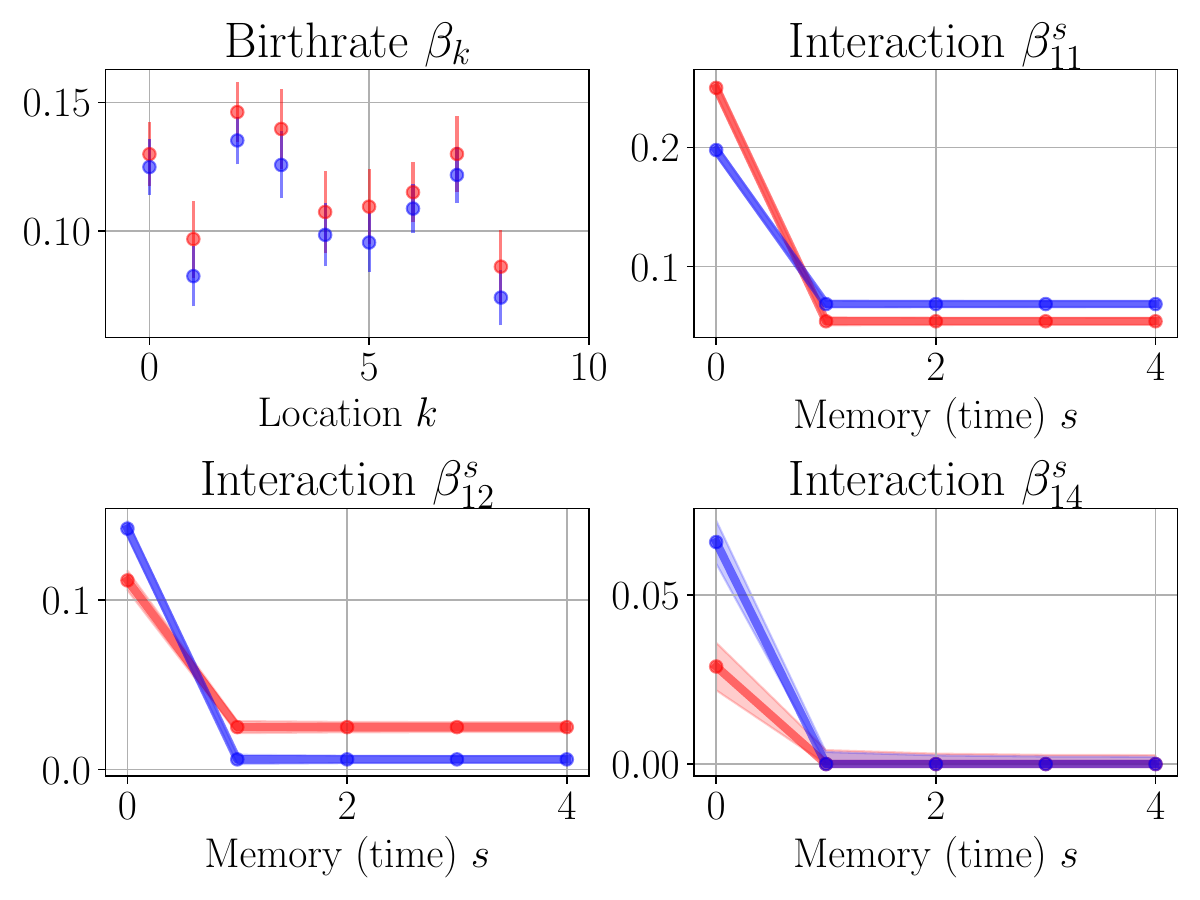}}
  \subfigure[North California]{\includegraphics[scale=0.22]{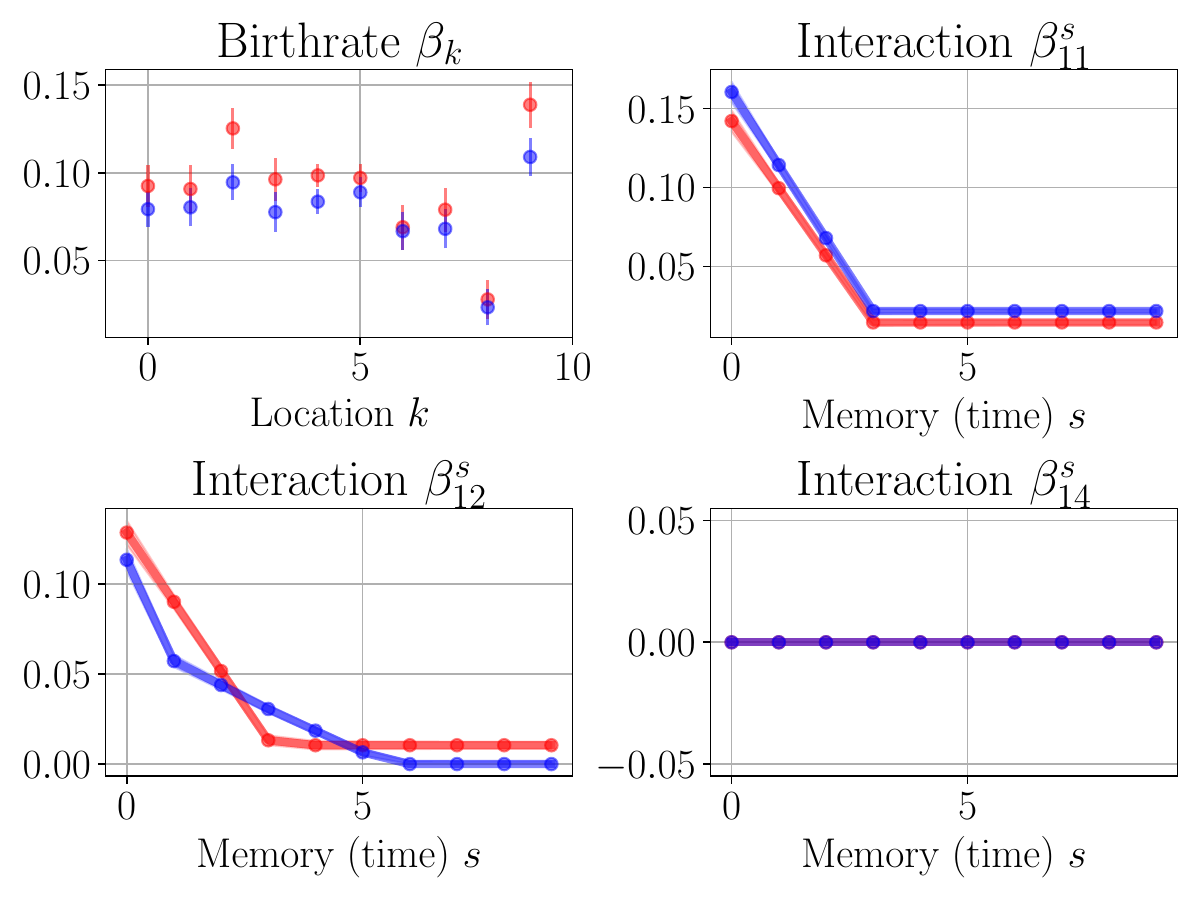}}
  \caption{Recovered birthrate parameters at different locations and interaction parameters over time. The blue curve shows ML recoveries, and the red curve shows LS recoveries. The shaded areas are 95$\%$ confidence intervals under bootstrap. Birthrates are similar in magnitude across different locations, and interactions all decay very fast over time.}
  \label{fig:5_2_beta_single}
\end{figure}

\vspace{0.1in}
\noindent \textit{(b) Visualize interaction parameters on terrain map:} We plot these parameters as graphs overlaid on a terrain map to visualize the influences. We only show the MLE estimates; the LS estimates are shown in Appendix B.3. The vertices' size and edges' width are proportional to the magnitude of corresponding recovered parameters. 

Below we summarize some findings from the results in Figure \ref{fig:5_2_grid_single}, which are interpretable and unveil previously unknown connections among the locations. (1) In Atlanta, birthrates are stronger in the corners. Regarding interaction parameters, at $s=1$, interactions have a relatively constant magnitude, so every pair of neighboring sensors interacts with each other. There are more latitude-wise (vertical) interactions than longitude-wise (horizontal) ones. At $s=5$ and 10, influences decay fast in magnitude, but those in the west and east of downtown Atlanta maintain reasonably large magnitude in comparison to influences along with other directions. (2) In Los Angeles, birthrates are stronger on the west and north sides of downtown. Regarding interaction parameters, at $s=1$, interactions are more obvious along latitude than along longitude. Interaction parameters in the south side of downtown have compatible magnitudes. At $s=3$ and 5, strong influences at the beginning tend to persist, with more persistence along latitude than longitude. (3) In North California, birthrates are very similar in size and concentrated around the northwest side of the map. Regarding interaction parameters, at $s=1$, we find strong influences from a city onto itself. In addition, influences mainly flow towards the southeast dimension, with large magnitudes even if cities are far away (e.g., San Mateo to Santa Clara and Sunnyvale to South San Jose). At $s=5$ and 10, influences follow almost the exact directions with decreasing magnitude.

In addition, based on comparisons between the two uniform grids of sensors at ATL and LA shown in Figure \ref{fig:5_1_mapview}, we find that they have a similar magnitude of interactions, but birthrates in ATL are smaller than those in LA. On the other hand, it is harder to compare either Atlanta or Los Angeles results with North California ones due to differences in sensor distributions. Nevertheless, we can see that magnitudes of birthrate and interactions are similar between cities in North California and grids in Los Angeles. Although we generally expect this situation to occur due to the closer geographic proximity between North California and Los Angeles, it is interesting to quantitatively and intuitively analyze this situation based on the results from our point-process model.
\begin{figure}[h!]
  \centering
  \subfigure[$s=1$]{\includegraphics[scale=0.41]{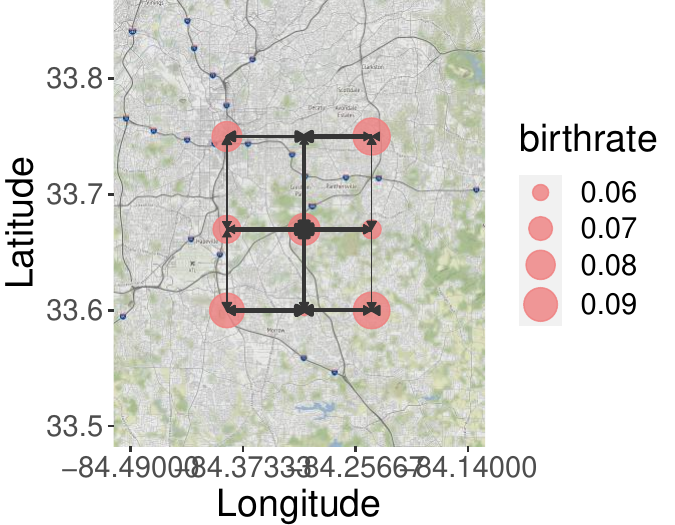}}
  \subfigure[$s=5$]{\includegraphics[scale=0.41]{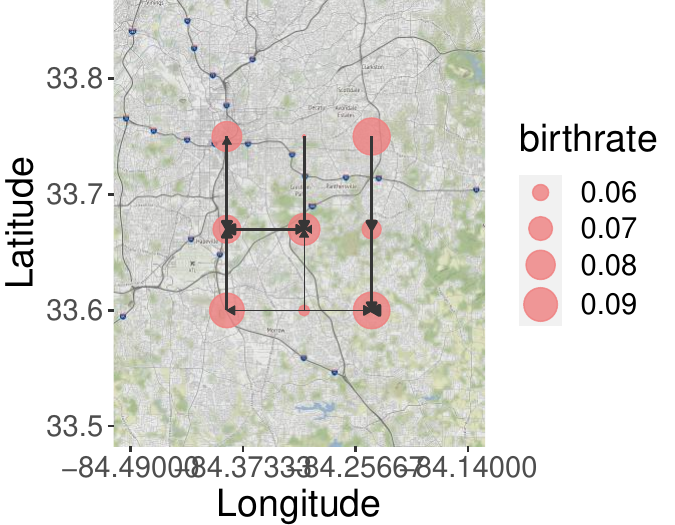}}
  \subfigure[$s=10$]{\includegraphics[scale=0.41]{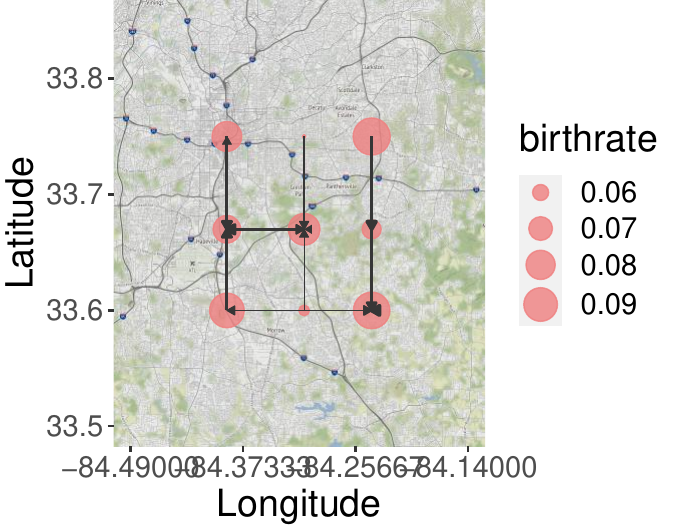}}
  \subfigure[$s=1$]{\includegraphics[scale=0.41]{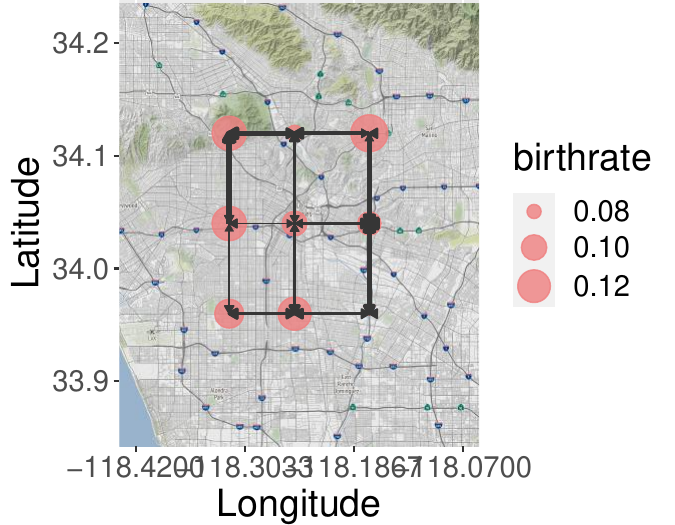}}
  \subfigure[$s=3$]{\includegraphics[scale=0.41]{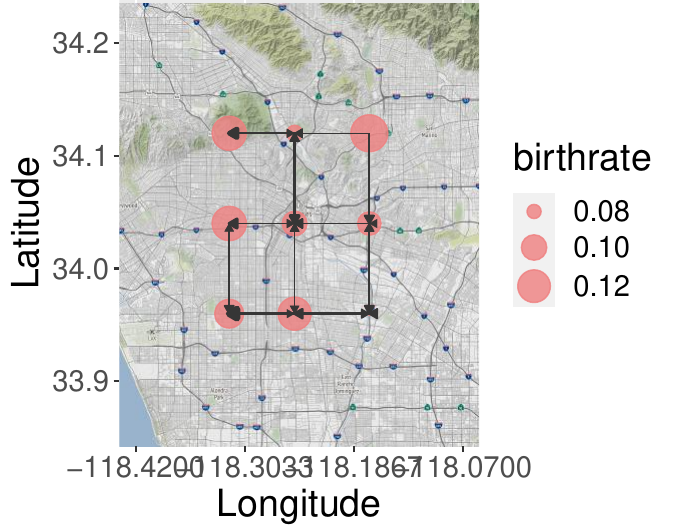}}
  \subfigure[$s=5$]{\includegraphics[scale=0.41]{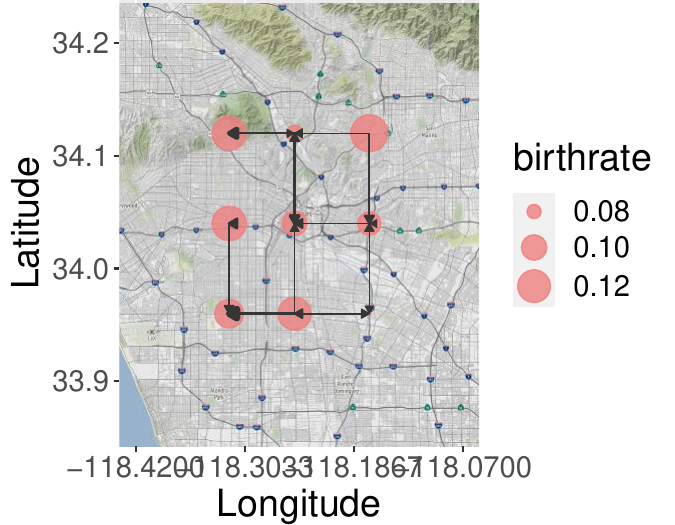}}
  \subfigure[$s=1$]{\includegraphics[scale=0.37]{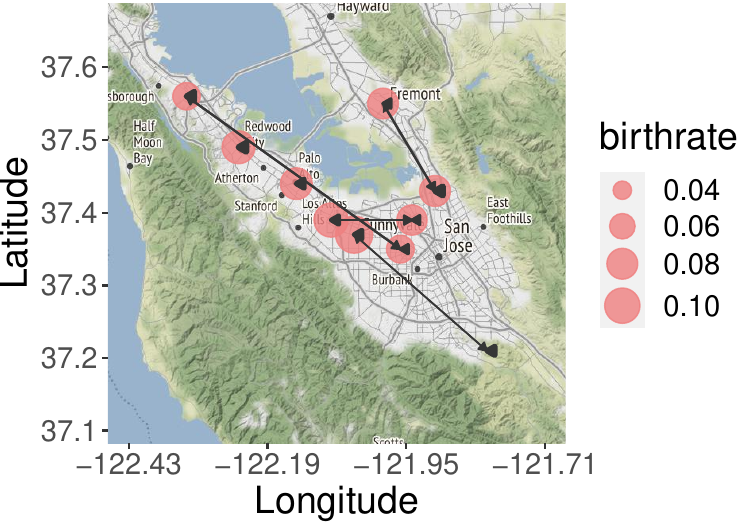}}
  \subfigure[$s=5$]{\includegraphics[scale=0.37]{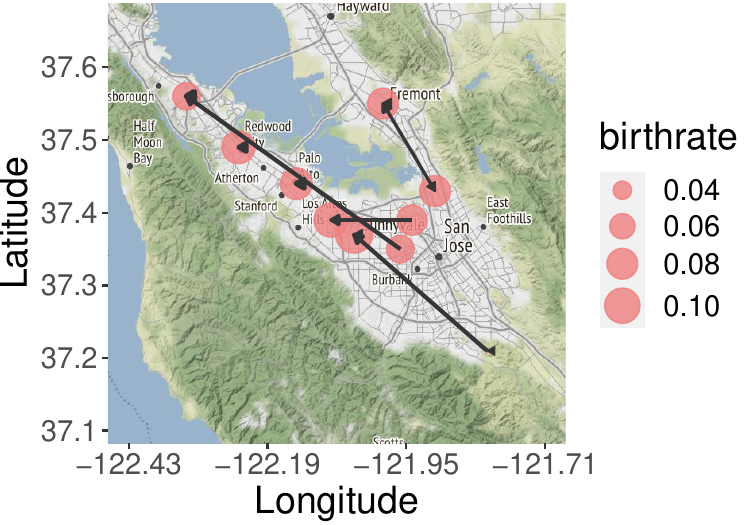}}
  \subfigure[$s=10$]{\includegraphics[scale=0.37]{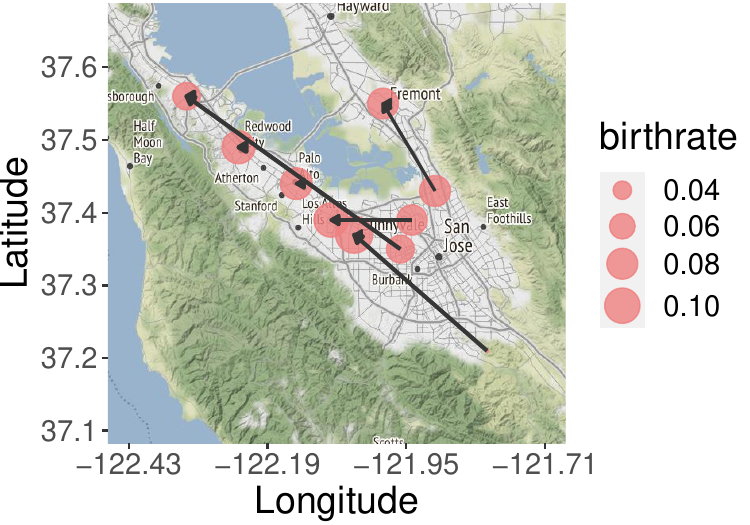}}
  \caption{Visualization of ML birthrate (red circles) and interactions (black lines) recovery (shown in Figure \ref{fig:5_2_beta_single} over time and locations) on terrain map. To make sure the edges are visible when $s>1$, we magnify the edge weight of Atlanta (a-c) and North California (g-i) estimates five times and of Los Angeles (d-f) estimates three times.}
  \label{fig:5_2_grid_single}
\end{figure}

\vspace{0.1in}
\noindent \textit{(c) Empirical examination of performance guarantee of estimates: }Using data from Atlanta, Los Angeles, and North California, we compute the norm bound under $p=1,2,\infty$ efficiently and provide the numerical values for the bound as mentioned in Theorem \ref{thm:LS-bound}. Bounds from using Theorem \ref{thm:ML-bound} are not shown since they contain the unknown parameter $\rho$. However, they are in the same order as bounds from Theorem \ref{thm:LS-bound}.

Table \ref{tab:performance_guarantee} presents the results under these three choices of the vector norm. Here, $\kappa=K+dK^2$, $N=365$, and we choose $\epsilon=0.1$. Results show that increasing $p$ leads to a larger condition number and, as a result, a tighter bound around the norm of differences. We notice the condition number in North California when $p=2$ is 0, possibly since the norm of the minimal eigenvalue of $A[\omega^N]$ under that case is 0 under numerical rounding (to five decimal places). 


\begin{table}[b]
    \centering
    \caption{Condition number $\theta_p(A[\boldsymbol\omega^n])$ and the resulting bounds on $\Vert \hat{\boldsymbol \beta} - \boldsymbol \beta \Vert_p$ from $p=1,2,\infty,$ using Theorem \ref{thm:LS-bound}.}
    \label{tab:performance_guarantee}
    \resizebox{\linewidth}{!}{
    
\begin{tabular}{p{3cm}p{2cm}p{2cm}p{2cm}p{2cm}p{2cm}p{2cm}p{2cm}}
     \hline
     & \multicolumn{2}{c}{$p=1$} & \multicolumn{2}{c}{$p=2$} & 
     \multicolumn{2}{c}{$p=\infty$}\\
     \hline
     \hline
Location & $\theta_1(A[\boldsymbol\omega^n])$ & \makecell[l]{Bound on\\$\Vert \hat{\boldsymbol \beta} - \boldsymbol \beta \Vert_1$} & $\theta_2(A[\boldsymbol\omega^n])$ & \makecell[l]{Bound on\\$\Vert \hat{\boldsymbol \beta} - \boldsymbol \beta \Vert_2$}& $\theta_{\infty}(A[\boldsymbol\omega^n])$ & \makecell[l]{Bound on\\$\Vert \hat{\boldsymbol \beta} - \boldsymbol \beta \Vert_{\infty}$} \\
\hline
Atlanta & 9e-05 & 24.53408 & 0.00668 & 2.84776 & 0.29315 & 0.42988 \\
Los Angeles & 1.6e-04 & 15.23908 & 0.01175 & 1.77828 & 0.4274 & 0.29485 \\
North California & 3e-05 & 35.89018 & 0.0 & $\infty$ & 0.41096 & 0.30665 \\
     \hline
    \end{tabular}
    }
\end{table}

\subsection{Prediction performance: } We use techniques in Section \ref{sec:anomaly_detection} to make a sequential prediction in Atlanta, Los Angeles, and Palo Alto (one representative example in North California) and evaluate the metrics described in Section \ref{exp:metric}. From Table \ref{tab:5_2_mod_accuracy}, our model yields very high $F_1$ scores in all three places, with generally better performance in Atlanta and Palo Alto than in Los Angeles. In particular, the performance of the dynamic threshold measured in the $F_1$ score in Palo Alto (i.e. North California) is much higher than that of the static threshold. The performance between static vs. dynamic thresholds is comparable in the other two cities. In addition, Table \ref{tab:baseline} compares the performance of our model on Atlanta data with two additional baselines: the logistic regression and the linear regression. Our method yields significantly higher $F_1$ scores.

To better visualize the estimated probabilities' trajectories, we plot the probability estimates, dynamic thresholds, and the prediction intervals in Figure \ref{fig:5_2_point_pred}. We also use the bootstrap confidence interval for $\boldsymbol \beta$, which was shown in Figure \ref{fig:5_2_beta_single}, to compute the confidence interval for $p_{tk}$. The prediction intervals at 95$\%$ confidence level concentrate closely around the estimates, even if the Bonferroni correction was used. Based on the figure, dynamic thresholds and probability estimates (red/blue dots) have similar rise-and-fall patterns. At the same time, the trajectory of probabilities also highly correlates with the actual ramping events (black dots). Such high correlation enables accurate prediction. Together with results in Table \ref{tab:5_2_mod_accuracy}, it is thus clear that dynamic thresholds yield decision boundaries that better distinguish the ramping events. Moreover, dynamically computing thresholds are more computationally efficient than computing statics ones, as we need not search over a grid of values as in Figure \ref{fig:5.2:accu_metric}.

\begin{table}[t]
    \centering
    \caption{Sequential prediction performance for single-state model: precision, recall, and $F_1$ score in three downtown under static vs. dynamic threshold. The highest value among the four methods (LS or MLE combined with static or dynamic threshold) is in bold.}
    \label{tab:5_2_mod_accuracy}
    \resizebox{\linewidth}{!}{\begin{tabular}{p{2cm}p{2cm}p{2cm}p{2cm}p{2cm}p{2cm}p{2cm}p{2cm}}
     \hline
     &&\multicolumn{3}{c}{Least Square} &\multicolumn{3}{c}{Maximum Likelihood} \\
     Location & $\tau$ & Precision & Recall & $F_1$ & Precision & Recall & $F_1$\\
     \hline
\multirow{2}{*}{Atlanta} & Static & 0.79 & 0.95 & 0.86 & \textbf{0.97} & \textbf{0.98} & \textbf{0.97} \\
 & Dynamic & 0.82 & 0.95 & 0.88 & 0.96 & 0.96 & 0.96 \\
\multirow{2}{*}{Los Angeles} & Static & 0.91 & 0.78 & 0.84 & 0.91 & 0.81 & 0.86 \\
 & Dynamic & \textbf{0.91} & 0.83 & 0.87 & 0.85 & \textbf{0.91} & \textbf{0.88} \\
\multirow{2}{*}{Palo Alto} & Static & \textbf{0.95} & 0.60 & 0.73 & 0.94 & 0.52 & 0.67 \\
 & Dynamic & 0.78 & \textbf{0.89} & \textbf{0.83} & 0.76 & 0.88 & 0.82 \\
     \hline
    \end{tabular}}
\end{table}


\begin{table}[htbp]
    \centering
    \caption{Baseline comparison for single-state model: we compare the $F_1$ scores of our method with two baselines methods, where ours yield significantly higher scores.}
    \label{tab:baseline}
    \resizebox{\linewidth}{!}{\begin{tabular}{p{2cm}p{2cm}p{2cm}p{2cm}p{2cm}p{2cm}}
     \hline
     &&\multicolumn{1}{c}{Least Square} &\multicolumn{1}{c}{Maximum Likelihood} &\multicolumn{1}{c}{Linear Regression}&\multicolumn{1}{c}{Logistic Regression}\\
     Location & $\tau$ & $F_1$ score  & $F_1$ score &$F_1$ score  & $F_1$ score\\
     \hline
\multirow{2}{*}{Atlanta} & Static& 0.96 & 0.97 & 0.64  &  0.67\\
 & Dynamic & 0.88 & 0.96 & 0.65 & 0.67  \\
     \hline
    \end{tabular}}
\end{table}

\begin{figure}[b]
    \centering
    \subfigure[Palo Alto (Static)]{\includegraphics[scale=0.21]{{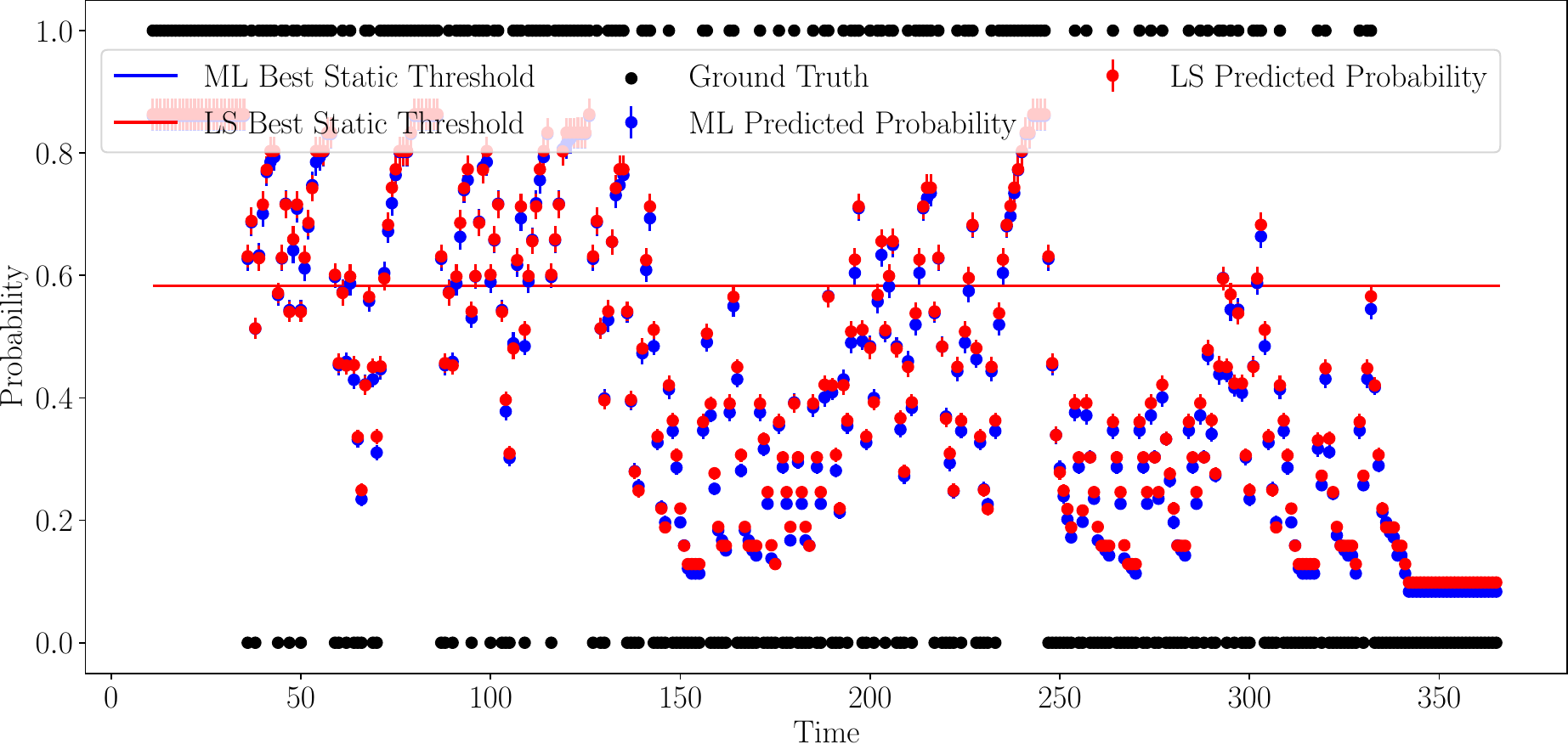}}}
    \subfigure[Palo Alto (Dynamic)]{\includegraphics[scale=0.21]{{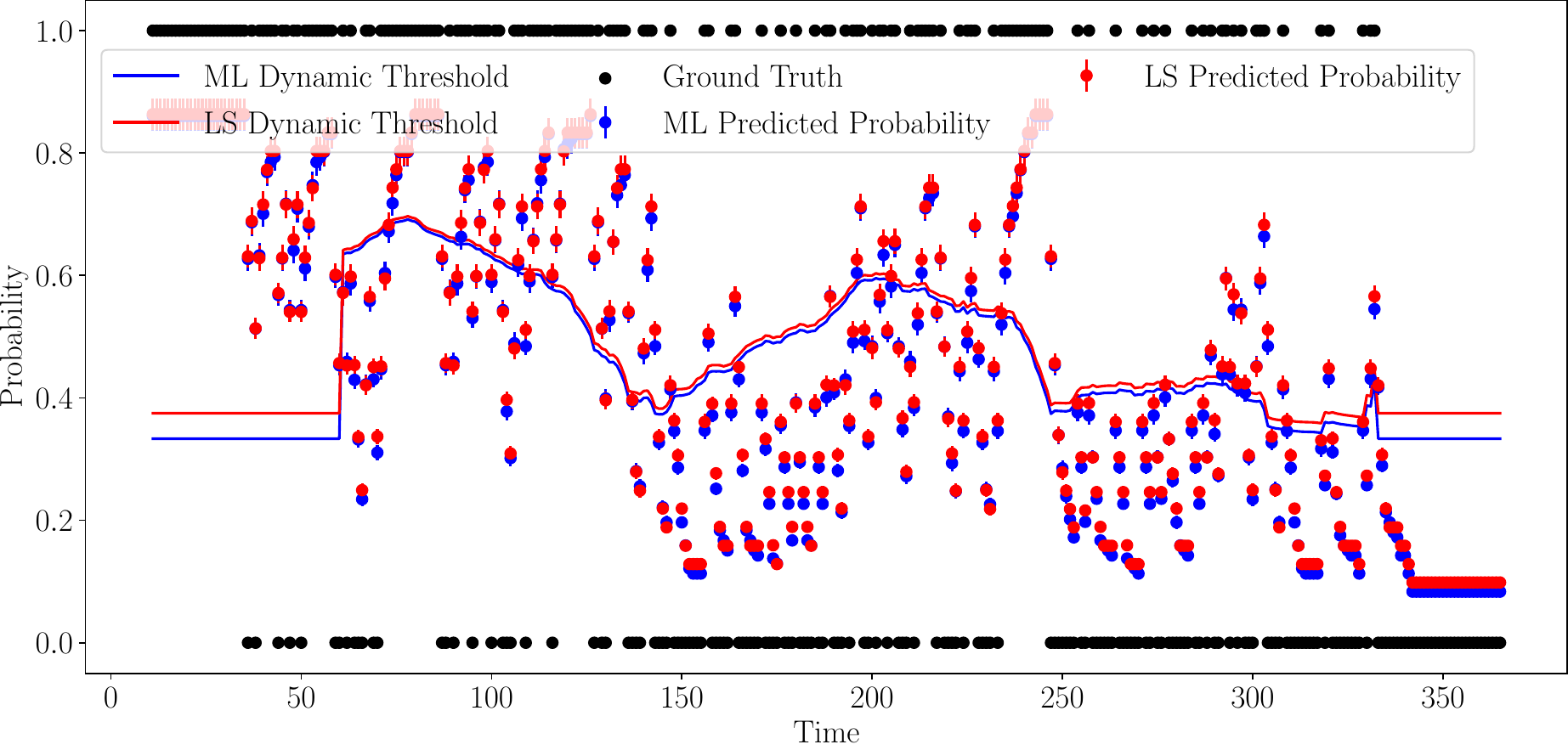}}}
    \vspace{-0.3cm}
    \caption{Prediction intervals for online point prediction of probabilities for single-state ramping events using LS (red dots) and ML (blue dots), compared with true ramping events (black dots). The left column is generated using a static threshold, and the right column uses dynamic thresholds. Dynamic thresholds separate the set of predicted ramping vs. normal probabilities better than static thresholds do.}
    \label{fig:5_2_point_pred}
\end{figure}

\subsection{Extension: Multi-state Modeling} We perform similar experiments as we did in the earlier section. Figures and Tables appear in Appendix B.1, and we summarize the findings.

\vspace{0.1in}
\noindent \textit{(a) Visualize these influences on terrain map:} We only discuss MLE results similar to what we did in the earlier section. In multi-state modeling, there are two ramping states (e.g., positive and negative) besides the non-ramping state. Hence, interactions among positive or negative ramping events must be separately considered. In addition, we note that there is barely any interaction between negative and positive ramping events, so we omit the discussion of such influences. We denote patterns among positive (resp. negative) ramping events as positive (resp. negative) patterns and denote interaction among positive (resp. negative) ramping events as positive-to-positive (resp. negative-to-negative) interactions. Upon analyzing these recovered parameters on the map, we find the following patterns:

\vspace{0.1in}
\textit{(1) Positive patterns, Figure 9}: (1) In Atlanta, only the lower-left sensor has a strong birthrate. Regarding interaction parameters, at $s=1$, horizontal estimates are strong at all latitudes. At $s=5$, some strong estimates at the beginning disappeared, but overall directions and relative magnitude stay the same. At $s=10$, the pattern persists. Some influences from west to east disappear. (2) In Los Angeles, sensors closer to downtown have stronger birthrates. Regarding interaction parameters, at $s=1$, it influences flow toward the middle from all directions. At $s=5$ and 10, interactions decay but exist almost everywhere, especially west of downtown. (3) In North California, we observe high birthrates on the northwest and east sides of the valley, but other places have small birthrates. Regarding interaction parameters, at $s=1$, influences flow towards the direction of silicon valley (i.e., northwest). There are clear influences from different cities themselves. At $s=5$ and 10, we see a similar pattern and direction of influence but the magnitude of influences decreases dramatically.

\vspace{0.1in}
\textit{(2) Negative patterns, Figure 10}: (1) In Atlanta, birthrates are stronger at the corners and downtown than in other places. Regarding interaction parameters, at $s=1$, latitude-wise influences are weaker than longitude-wise ones. We also see some flow from west to east. At $s=5$ and 10, influences almost disappear, even under magnification. (2) In Los Angeles, birthrates have similar magnitudes except those in the southeast corner. Regarding interaction parameters, at $s=1$, influences exist more along latitudes than along longitude. At $s=5$ and 10, we see fast decay in magnitude, even under strong magnification. Only influences south of downtown and northeast of downtown exist. (3) In North California, birthrates are similar in magnitude except in South San Jose. Regarding interaction parameters, at $s=1$, results are mostly similar to ones in the positive ramping events modeling, except for the strong influence from Sunnyvale to San Jose (reversed in the positive ramping event case). At $s=5$, we can only see influence from Sunnyvale to Mountain View and Sunnyvale to South San Jose, besides some past influences from a certain location onto itself. Some of these also decrease a lot in magnitude. At $s=10$, only Sunnyvale to Mountain View exists, besides influences at certain locations onto itself.

We also identified some additional findings: (1) Upon comparing the positive and negative patterns, we found that interactions among positive ramping events persist more prolonged than those in negative ones. In addition, the directions and the flow of the interaction also differ among these two sets of events. (2) Within positive patterns, interactions among sensors in Atlanta can be sparser than those in Los Angeles; in single-state recovery, Atlanta interactions are denser than Los Angeles results. Birthrates are comparable in these two places. On the other hand, interaction among sensors in North California also becomes weaker in the multi-state recovery than in the single-state recovery. The sparsity of interactions in North California is similar to that in Los Angeles. Birthrate results in North California are still closer to those in Los Angeles. In contrast, interactions at all locations are relatively weak within negative patterns after $s>1$. However, birthrates at all these locations are more comparable to each other.

\vspace{0.1in}
\noindent \textit{(b) Prediction Performance:} Table 4 and Table 5 show prediction performance for positive and negative ramping event detection under the performance metrics described in Section \ref{exp:metric}. Predictions using MLE under dynamic threshold lead to the highest $F_1$ scores among 3 out of 6 situations. Each situation denotes results in a city under different prediction methods for either positive or negative ramping event detection. We expect such behavior to happen because MLE estimates typically enjoy better theoretical performance under large sample sizes than LS ones, and dynamic thresholds are expected to outperform static ones. In addition, we note that the best prediction performances under multi-state modeling are better than those under single-state modeling, indicating that the point process model performs better when the type of ramping events is analyzed separately 

Figure 11 shows the point predictions, static or dynamic thresholds, and actual ramping events, providing further insights behind model performances. We use the same bootstrap techniques to provide prediction intervals on predicted point probabilities. Similar to the single-state results in Figure \ref{fig:5_2_point_pred}, dynamic thresholds closely follow the pattern in point predictions, which highly correlate with the actual ramping events. In addition, positive or negative ramping events often occur in sequence, reflecting the temporal dependencies among such events.

\subsection{Extension: Seasonal Modeling} 

We treat seasonality as an essential factor affecting ramping events' distribution. As a result, we next fit separate models based on data in different seasons to better capture the underlying season-dependent data distribution and yield more accurate ramping event prediction. 
In particular, we fit our point-process models using the LS or MLE method separately on data between January (01/01) till March (03/31) and on data from July (07/01) till September (09/30) data in 2017, resulting in four separate models. The test data for each model are from the same months in 2018. We notice that fitting different models under these seasons can sometimes further improve prediction performances. Meanwhile, the interaction parameters do not have significant differences across seasons in single-state modeling; in multi-state modeling, interaction parameters behave differently because of the vast differences in the number of positive and negative abnormal events in each season.

We create figures and tables similar to before in Appendix B.2, by comparing estimation and prediction results from different seasonal models side by side. We show single-state results before multi-state ones and no longer show bootstrap prediction intervals because we mainly want to compare the estimated parameters and prediction results between seasonal models.

\vspace{0.1in}
\noindent \textit{(a) Single-state results:} We first compare and contrast the estimated parameters. Figure 12 shows that the birthrates are similar between the two models, but the interaction parameters may behave very differently. Figure 13 visualizes these parameters on terrain maps. The directions of interactions are similar in both models, but there are obvious differences between the magnitude of interaction and birthrate parameters. Then, Table 6 shows prediction performances, where the $F_1$ scores in any season are better than earlier single-state results. Finally, Figure 14 shows the estimated probabilities, ground-truth events, and static/dynamic thresholds. The interpretations are similar to before. Moreover, we remark that because our dynamic thresholds defined in \eqref{thres} depend on the past anomalies (yet there are very few in each season), the dynamic thresholds are much flatter in seasonal models.

\vspace{0.1in}
\noindent \textit{(b) Multi-state results:} Figures 15 and 16 visualize on terrain maps the estimated parameters from positive and negative ramping events. We note that the interaction parameters are very different across seasons due to huge differences in the number of positive and negative ramping events in each season. In particular, during winter (i.e., January---March), there are many more positive ramping events than negative, and the reverse situation occurs during summer (i.e., July---September). We suspect such situations occur because radiation is typically very low during winter. Hence, a negative ramping event rarely occurs, which is defined as a certain radiation level with a much lower magnitude than other levels according to our definition. The reverse reasoning applies to ramping events in summer.

We then show the sequential prediction performance in Table 7. Some cells are marked as 0 because there is no ramping event of a given type during that season. The highest values of these accuracy measures are higher than that of the earlier multi-state model, revealing the benefit of fitting seasonal models. Finally, Figure 17 shows the estimated probabilities, ground-truth events, and static/dynamic thresholds, and has very similar patterns as Figure \ref{fig:5_2_point_pred}, so we omit the discussions.

\section{Conclusion and future works}
\label{sec:conclusion}

In this paper, we introduce a new framework for modeling spatio-temporal abnormal events in solar radiation, which consists of three steps: ramping event extraction, parameter estimation, and ramping event prediction. We apply the method to various regions in the US under different settings (i.e., multi-state, seasonal data) to show that the model is flexible and robust, yielding physically meaningful results for interpretation. We believe the proposed model is a general framework for other spatio-temporal data modeling tasks.

There are a few possible extensions. In terms of data, we can consider other definitions of ramping events and more categories per event, which may include the severity of each ramping event. We can also include link functions and other regularization techniques for improved performances during modeling. For prediction, we can improve dynamic threshold designs via techniques in sequential decision making \citep{6208875}.

\section*{Acknowledgements}

The authors would like to thank the anonymous referees, an Associate
Editor and the Editor for their constructive comments that improved the
quality of this paper.

The first three authors are supported by NSF CCF-1650913, NSF DMS-1938106, NSF DMS-1830210, and CMMI-2015787.



\bibliographystyle{imsart-nameyear} 
\bibliography{ref.bib}       

\begin{thebibliography}{32}

\bibitem[\protect\citeauthoryear{Abuella and Chowdhury}{2019}]{postprocess}
\begin{barticle}[author]
\bauthor{\bsnm{Abuella},~\bfnm{Mohamed}\binits{M.}} \AND
  \bauthor{\bsnm{Chowdhury},~\bfnm{Badrul}\binits{B.}}
(\byear{2019}).
\btitle{Forecasting of solar power ramp events: A post-processing approach}.
\bjournal{Renewable Energy}
\bvolume{133}
\bpages{1380 - 1392}.
\bdoi{https://doi.org/10.1016/j.renene.2018.09.005}
\end{barticle}
\endbibitem

\bibitem[\protect\citeauthoryear{Aggarwal}{2015}]{aggarwal2015outlier}
\begin{binproceedings}[author]
\bauthor{\bsnm{Aggarwal},~\bfnm{Charu~C}\binits{C.~C.}}
(\byear{2015}).
\btitle{Outlier analysis}.
In \bbooktitle{Data mining}
\bpages{237--263}.
\bpublisher{Springer}.
\end{binproceedings}
\endbibitem

\bibitem[\protect\citeauthoryear{Akoglu, Tong and
  Koutra}{2015}]{akoglu2015graph}
\begin{barticle}[author]
\bauthor{\bsnm{Akoglu},~\bfnm{Leman}\binits{L.}},
  \bauthor{\bsnm{Tong},~\bfnm{Hanghang}\binits{H.}} \AND
  \bauthor{\bsnm{Koutra},~\bfnm{Danai}\binits{D.}}
(\byear{2015}).
\btitle{Graph based anomaly detection and description: a survey}.
\bjournal{Data mining and knowledge discovery}
\bvolume{29}
\bpages{626--688}.
\end{barticle}
\endbibitem

\bibitem[\protect\citeauthoryear{Cui et~al.}{2015a}]{cui2015optimized}
\begin{barticle}[author]
\bauthor{\bsnm{Cui},~\bfnm{Mingjian}\binits{M.}},
  \bauthor{\bsnm{Zhang},~\bfnm{Jie}\binits{J.}},
  \bauthor{\bsnm{Florita},~\bfnm{Anthony~R}\binits{A.~R.}},
  \bauthor{\bsnm{Hodge},~\bfnm{Bri-Mathias}\binits{B.-M.}},
  \bauthor{\bsnm{Ke},~\bfnm{Deping}\binits{D.}} \AND
  \bauthor{\bsnm{Sun},~\bfnm{Yuanzhang}\binits{Y.}}
(\byear{2015}a).
\btitle{An optimized swinging door algorithm for identifying wind ramping
  events}.
\bjournal{IEEE Transactions on Sustainable Energy}
\bvolume{7}
\bpages{150--162}.
\end{barticle}
\endbibitem

\bibitem[\protect\citeauthoryear{{Cui} et~al.}{2015b}]{optsda2015}
\begin{binproceedings}[author]
\bauthor{\bsnm{{Cui}},~\bfnm{M.}\binits{M.}},
  \bauthor{\bsnm{{Zhang}},~\bfnm{J.}\binits{J.}},
  \bauthor{\bsnm{{Florita}},~\bfnm{A.~R.}\binits{A.~R.}},
  \bauthor{\bsnm{{Hodge}},~\bfnm{B.}\binits{B.}},
  \bauthor{\bsnm{{Ke}},~\bfnm{D.}\binits{D.}} \AND
  \bauthor{\bsnm{{Sun}},~\bfnm{Y.}\binits{Y.}}
(\byear{2015}b).
\btitle{An optimized swinging door algorithm for wind power ramp event
  detection}.
In \bbooktitle{2015 IEEE Power Energy Society General Meeting}
\bpages{1-5}.
\end{binproceedings}
\endbibitem

\bibitem[\protect\citeauthoryear{Cui et~al.}{2017}]{optsda2017}
\begin{barticle}[author]
\bauthor{\bsnm{Cui},~\bfnm{Mingjian}\binits{M.}},
  \bauthor{\bsnm{Zhang},~\bfnm{Jie}\binits{J.}},
  \bauthor{\bsnm{Feng},~\bfnm{Cong}\binits{C.}},
  \bauthor{\bsnm{Florita},~\bfnm{Anthony~R.}\binits{A.~R.}},
  \bauthor{\bsnm{Sun},~\bfnm{Yuanzhang}\binits{Y.}} \AND
  \bauthor{\bsnm{Hodge},~\bfnm{Bri~Mathias}\binits{B.~M.}}
(\byear{2017}).
\btitle{Characterizing and analyzing ramping events in wind power, solar power,
  load, and netload}.
\bjournal{Renewable Energy}
\bvolume{111}.
\bdoi{10.1016/j.renene.2017.04.005}
\end{barticle}
\endbibitem

\bibitem[\protect\citeauthoryear{Du et~al.}{2016}]{du2016recurrent}
\begin{binproceedings}[author]
\bauthor{\bsnm{Du},~\bfnm{Nan}\binits{N.}},
  \bauthor{\bsnm{Dai},~\bfnm{Hanjun}\binits{H.}},
  \bauthor{\bsnm{Trivedi},~\bfnm{Rakshit}\binits{R.}},
  \bauthor{\bsnm{Upadhyay},~\bfnm{Utkarsh}\binits{U.}},
  \bauthor{\bsnm{Gomez-Rodriguez},~\bfnm{Manuel}\binits{M.}} \AND
  \bauthor{\bsnm{Song},~\bfnm{Le}\binits{L.}}
(\byear{2016}).
\btitle{Recurrent marked temporal point processes: Embedding event history to
  vector}.
In \bbooktitle{Proceedings of the 22nd ACM SIGKDD International Conference on
  Knowledge Discovery and Data Mining}
\bpages{1555--1564}.
\end{binproceedings}
\endbibitem

\bibitem[\protect\citeauthoryear{Florita, Hodge and
  Orwig}{2013a}]{florita2013identifying}
\begin{binproceedings}[author]
\bauthor{\bsnm{Florita},~\bfnm{Anthony}\binits{A.}},
  \bauthor{\bsnm{Hodge},~\bfnm{Bri-Mathias}\binits{B.-M.}} \AND
  \bauthor{\bsnm{Orwig},~\bfnm{Kirsten}\binits{K.}}
(\byear{2013}a).
\btitle{Identifying wind and solar ramping events}.
In \bbooktitle{2013 IEEE Green Technologies Conference (GreenTech)}
\bpages{147--152}.
\bpublisher{IEEE}.
\end{binproceedings}
\endbibitem

\bibitem[\protect\citeauthoryear{{Florita}, {Hodge} and
  {Orwig}}{2013b}]{sda2013}
\begin{binproceedings}[author]
\bauthor{\bsnm{{Florita}},~\bfnm{A.}\binits{A.}},
  \bauthor{\bsnm{{Hodge}},~\bfnm{B.}\binits{B.}} \AND
  \bauthor{\bsnm{{Orwig}},~\bfnm{K.}\binits{K.}}
(\byear{2013}b).
\btitle{Identifying Wind and Solar Ramping Events}.
In \bbooktitle{2013 IEEE Green Technologies Conference (GreenTech)}
\bpages{147-152}.
\end{binproceedings}
\endbibitem

\bibitem[\protect\citeauthoryear{Grubbs}{1969}]{grubbs1969procedures}
\begin{barticle}[author]
\bauthor{\bsnm{Grubbs},~\bfnm{Frank~E}\binits{F.~E.}}
(\byear{1969}).
\btitle{Procedures for detecting outlying observations in samples}.
\bjournal{Technometrics}
\bvolume{11}
\bpages{1--21}.
\end{barticle}
\endbibitem

\bibitem[\protect\citeauthoryear{Gupta et~al.}{2013}]{gupta2013outlier}
\begin{barticle}[author]
\bauthor{\bsnm{Gupta},~\bfnm{Manish}\binits{M.}},
  \bauthor{\bsnm{Gao},~\bfnm{Jing}\binits{J.}},
  \bauthor{\bsnm{Aggarwal},~\bfnm{Charu~C}\binits{C.~C.}} \AND
  \bauthor{\bsnm{Han},~\bfnm{Jiawei}\binits{J.}}
(\byear{2013}).
\btitle{Outlier detection for temporal data: A survey}.
\bjournal{IEEE Transactions on Knowledge and data Engineering}
\bvolume{26}
\bpages{2250--2267}.
\end{barticle}
\endbibitem

\bibitem[\protect\citeauthoryear{Hawkes}{1971}]{hawkes1971spectra}
\begin{barticle}[author]
\bauthor{\bsnm{Hawkes},~\bfnm{Alan~G}\binits{A.~G.}}
(\byear{1971}).
\btitle{Spectra of some self-exciting and mutually exciting point processes}.
\bjournal{Biometrika}
\bvolume{58}
\bpages{83--90}.
\end{barticle}
\endbibitem

\bibitem[\protect\citeauthoryear{Hawkins et~al.}{2002}]{hawkins2002outlier}
\begin{binproceedings}[author]
\bauthor{\bsnm{Hawkins},~\bfnm{Simon}\binits{S.}},
  \bauthor{\bsnm{He},~\bfnm{Hongxing}\binits{H.}},
  \bauthor{\bsnm{Williams},~\bfnm{Graham}\binits{G.}} \AND
  \bauthor{\bsnm{Baxter},~\bfnm{Rohan}\binits{R.}}
(\byear{2002}).
\btitle{Outlier detection using replicator neural networks}.
In \bbooktitle{International Conference on Data Warehousing and Knowledge
  Discovery}
\bpages{170--180}.
\bpublisher{Springer}.
\end{binproceedings}
\endbibitem

\bibitem[\protect\citeauthoryear{Huang et~al.}{2012}]{huang2012solar}
\begin{binproceedings}[author]
\bauthor{\bsnm{Huang},~\bfnm{Rui}\binits{R.}},
  \bauthor{\bsnm{Huang},~\bfnm{Tiana}\binits{T.}},
  \bauthor{\bsnm{Gadh},~\bfnm{Rajit}\binits{R.}} \AND
  \bauthor{\bsnm{Li},~\bfnm{Na}\binits{N.}}
(\byear{2012}).
\btitle{Solar generation prediction using the ARMA model in a laboratory-level
  micro-grid}.
In \bbooktitle{2012 IEEE third international conference on smart grid
  communications (SmartGridComm)}
\bpages{528--533}.
\bpublisher{IEEE}.
\end{binproceedings}
\endbibitem

\bibitem[\protect\citeauthoryear{Isham and Westcott}{1979}]{isham1979self}
\begin{barticle}[author]
\bauthor{\bsnm{Isham},~\bfnm{Valerie}\binits{V.}} \AND
  \bauthor{\bsnm{Westcott},~\bfnm{Mark}\binits{M.}}
(\byear{1979}).
\btitle{A self-correcting point process}.
\bjournal{Stochastic processes and their applications}
\bvolume{8}
\bpages{335--347}.
\end{barticle}
\endbibitem

\bibitem[\protect\citeauthoryear{Juditsky et~al.}{2020}]{juditsky2020convex}
\begin{barticle}[author]
\bauthor{\bsnm{Juditsky},~\bfnm{Anatoli}\binits{A.}},
  \bauthor{\bsnm{Nemirovski},~\bfnm{Arkadi}\binits{A.}},
  \bauthor{\bsnm{Xie},~\bfnm{Liyan}\binits{L.}} \AND
  \bauthor{\bsnm{Xie},~\bfnm{Yao}\binits{Y.}}
(\byear{2020}).
\btitle{Convex Recovery of Marked Spatio-Temporal Point Processes}.
\bjournal{arXiv preprint arXiv:2003.12935}.
\end{barticle}
\endbibitem

\bibitem[\protect\citeauthoryear{Kamath}{2010}]{kamath2010understanding}
\begin{binproceedings}[author]
\bauthor{\bsnm{Kamath},~\bfnm{Chandrika}\binits{C.}}
(\byear{2010}).
\btitle{Understanding wind ramp events through analysis of historical data}.
In \bbooktitle{IEEE PES T\&D 2010}
\bpages{1--6}.
\bpublisher{IEEE}.
\end{binproceedings}
\endbibitem

\bibitem[\protect\citeauthoryear{Kim et~al.}{2018}]{kim2018flexible}
\begin{barticle}[author]
\bauthor{\bsnm{Kim},~\bfnm{Min-Ook}\binits{M.-O.}},
  \bauthor{\bsnm{Pyo},~\bfnm{Soonjae}\binits{S.}},
  \bauthor{\bsnm{Oh},~\bfnm{Yongkeun}\binits{Y.}},
  \bauthor{\bsnm{Kang},~\bfnm{Yunsung}\binits{Y.}},
  \bauthor{\bsnm{Cho},~\bfnm{Kyung-Ho}\binits{K.-H.}},
  \bauthor{\bsnm{Choi},~\bfnm{Jungwook}\binits{J.}} \AND
  \bauthor{\bsnm{Kim},~\bfnm{Jongbaeg}\binits{J.}}
(\byear{2018}).
\btitle{Flexible and multi-directional piezoelectric energy harvester for
  self-powered human motion sensor}.
\bjournal{Smart Materials and Structures}
\bvolume{27}
\bpages{035001}.
\end{barticle}
\endbibitem

\bibitem[\protect\citeauthoryear{Li et~al.}{2018}]{li2018learning}
\begin{binproceedings}[author]
\bauthor{\bsnm{Li},~\bfnm{Shuang}\binits{S.}},
  \bauthor{\bsnm{Xiao},~\bfnm{Shuai}\binits{S.}},
  \bauthor{\bsnm{Zhu},~\bfnm{Shixiang}\binits{S.}},
  \bauthor{\bsnm{Du},~\bfnm{Nan}\binits{N.}},
  \bauthor{\bsnm{Xie},~\bfnm{Yao}\binits{Y.}} \AND
  \bauthor{\bsnm{Song},~\bfnm{Le}\binits{L.}}
(\byear{2018}).
\btitle{Learning temporal point processes via reinforcement learning}.
In \bbooktitle{Advances in neural information processing systems}
\bpages{10781--10791}.
\end{binproceedings}
\endbibitem

\bibitem[\protect\citeauthoryear{Liu, Yu and Liu}{2009}]{liu2009solar}
\begin{barticle}[author]
\bauthor{\bsnm{Liu},~\bfnm{Quanhua}\binits{Q.}},
  \bauthor{\bsnm{Yu},~\bfnm{Gengfa}\binits{G.}} \AND
  \bauthor{\bsnm{Liu},~\bfnm{Jue~J}\binits{J.~J.}}
(\byear{2009}).
\btitle{Solar radiation as large-scale resource for energy-short world}.
\bjournal{Energy \& environment}
\bvolume{20}
\bpages{319--329}.
\end{barticle}
\endbibitem

\bibitem[\protect\citeauthoryear{Lukasik, Cohn and Bontcheva}{2015}]{Rumor}
\begin{binproceedings}[author]
\bauthor{\bsnm{Lukasik},~\bfnm{M.}\binits{M.}},
  \bauthor{\bsnm{Cohn},~\bfnm{Trevor}\binits{T.}} \AND
  \bauthor{\bsnm{Bontcheva},~\bfnm{Kalina}\binits{K.}}
(\byear{2015}).
\btitle{Point Process Modelling of Rumour Dynamics in Social Media}.
In \bbooktitle{ACL}.
\end{binproceedings}
\endbibitem

\bibitem[\protect\citeauthoryear{Mei and Eisner}{2017}]{mei2017neural}
\begin{binproceedings}[author]
\bauthor{\bsnm{Mei},~\bfnm{Hongyuan}\binits{H.}} \AND
  \bauthor{\bsnm{Eisner},~\bfnm{Jason~M}\binits{J.~M.}}
(\byear{2017}).
\btitle{The neural hawkes process: A neurally self-modulating multivariate
  point process}.
In \bbooktitle{Advances in Neural Information Processing Systems}
\bpages{6754--6764}.
\end{binproceedings}
\endbibitem

\bibitem[\protect\citeauthoryear{Nesterov}{1998}]{nesterov1998}
\begin{barticle}[author]
\bauthor{\bsnm{Nesterov},~\bfnm{Yu}\binits{Y.}}
(\byear{1998}).
\btitle{Semidefinite relaxation and nonconvex quadratic optimization}.
\bjournal{Optimization Methods and Software}
\bvolume{9}
\bpages{141-160}.
\bdoi{10.1080/10556789808805690}
\end{barticle}
\endbibitem

\bibitem[\protect\citeauthoryear{Ningegowda and
  Premachandran}{2014}]{ningegowda2014coupled}
\begin{barticle}[author]
\bauthor{\bsnm{Ningegowda},~\bfnm{BM}\binits{B.}} \AND
  \bauthor{\bsnm{Premachandran},~\bfnm{B}\binits{B.}}
(\byear{2014}).
\btitle{A coupled level set and volume of fluid method with multi-directional
  advection algorithms for two-phase flows with and without phase change}.
\bjournal{International Journal of Heat and Mass Transfer}
\bvolume{79}
\bpages{532--550}.
\end{barticle}
\endbibitem

\bibitem[\protect\citeauthoryear{{Raginsky} et~al.}{2012}]{6208875}
\begin{barticle}[author]
\bauthor{\bsnm{{Raginsky}},~\bfnm{M.}\binits{M.}},
  \bauthor{\bsnm{{Willett}},~\bfnm{R.~M.}\binits{R.~M.}},
  \bauthor{\bsnm{{Horn}},~\bfnm{C.}\binits{C.}},
  \bauthor{\bsnm{{Silva}},~\bfnm{J.}\binits{J.}} \AND
  \bauthor{\bsnm{{Marcia}},~\bfnm{R.~F.}\binits{R.~F.}}
(\byear{2012}).
\btitle{Sequential Anomaly Detection in the Presence of Noise and Limited
  Feedback}.
\bjournal{IEEE Transactions on Information Theory}
\bvolume{58}
\bpages{5544-5562}.
\end{barticle}
\endbibitem

\bibitem[\protect\citeauthoryear{Rocchetta, Li and
  Zio}{2015}]{rocchetta2015risk}
\begin{barticle}[author]
\bauthor{\bsnm{Rocchetta},~\bfnm{Roberto}\binits{R.}},
  \bauthor{\bsnm{Li},~\bfnm{Yanfu}\binits{Y.}} \AND
  \bauthor{\bsnm{Zio},~\bfnm{Enrico}\binits{E.}}
(\byear{2015}).
\btitle{Risk assessment and risk-cost optimization of distributed power
  generation systems considering extreme weather conditions}.
\bjournal{Reliability Engineering \& System Safety}
\bvolume{136}
\bpages{47--61}.
\end{barticle}
\endbibitem

\bibitem[\protect\citeauthoryear{Tang and Li}{2020}]{MultivariateTP}
\begin{barticle}[author]
\bauthor{\bsnm{Tang},~\bfnm{Xi-Wei}\binits{X.-W.}} \AND
  \bauthor{\bsnm{Li},~\bfnm{L.}\binits{L.}}
(\byear{2020}).
\btitle{Multivariate Temporal Point Process Regression}.
\bjournal{arXiv: Methodology}.
\end{barticle}
\endbibitem

\bibitem[\protect\citeauthoryear{Wu et~al.}{2020}]{ModelingEP}
\begin{barticle}[author]
\bauthor{\bsnm{Wu},~\bfnm{Weichang}\binits{W.}},
  \bauthor{\bsnm{Liu},~\bfnm{Huanxi}\binits{H.}},
  \bauthor{\bsnm{Zhang},~\bfnm{X.}\binits{X.}},
  \bauthor{\bsnm{Liu},~\bfnm{Y.}\binits{Y.}} \AND
  \bauthor{\bsnm{Zha},~\bfnm{H.}\binits{H.}}
(\byear{2020}).
\btitle{Modeling Event Propagation via Graph Biased Temporal Point Process}.
\bjournal{IEEE transactions on neural networks and learning systems}
\bvolume{PP}.
\end{barticle}
\endbibitem

\bibitem[\protect\citeauthoryear{Xu and Xie}{2021}]{xu2021conformal}
\begin{bmisc}[author]
\bauthor{\bsnm{Xu},~\bfnm{Chen}\binits{C.}} \AND
  \bauthor{\bsnm{Xie},~\bfnm{Yao}\binits{Y.}}
(\byear{2021}).
\btitle{Conformal Anomaly Detection on Spatio-Temporal Observations with
  Missing Data}.
\end{bmisc}
\endbibitem

\bibitem[\protect\citeauthoryear{{Zhu} et~al.}{2019}]{limithistory}
\begin{binproceedings}[author]
\bauthor{\bsnm{{Zhu}},~\bfnm{W.}\binits{W.}},
  \bauthor{\bsnm{{Zhang}},~\bfnm{L.}\binits{L.}},
  \bauthor{\bsnm{{Yang}},~\bfnm{M.}\binits{M.}} \AND
  \bauthor{\bsnm{{Wang}},~\bfnm{B.}\binits{B.}}
(\byear{2019}).
\btitle{Solar Power Ramp Event Forewarning with Limited Historical
  Observations}.
In \bbooktitle{2019 IEEE/IAS 55th Industrial and Commercial Power Systems
  Technical Conference (I CPS)}
\bpages{1-8}.
\end{binproceedings}
\endbibitem

\bibitem[\protect\citeauthoryear{Zhu et~al.}{2020a}]{zhu2020deep}
\begin{barticle}[author]
\bauthor{\bsnm{Zhu},~\bfnm{Shixiang}\binits{S.}},
  \bauthor{\bsnm{Zhang},~\bfnm{Minghe}\binits{M.}},
  \bauthor{\bsnm{Ding},~\bfnm{Ruyi}\binits{R.}} \AND
  \bauthor{\bsnm{Xie},~\bfnm{Yao}\binits{Y.}}
(\byear{2020}a).
\btitle{Deep Attention Spatio-Temporal Point Processes}.
\bjournal{arXiv preprint arXiv:2002.07281}.
\end{barticle}
\endbibitem

\bibitem[\protect\citeauthoryear{Zhu et~al.}{2020b}]{zhu2020spatio}
\begin{barticle}[author]
\bauthor{\bsnm{Zhu},~\bfnm{Shixiang}\binits{S.}},
  \bauthor{\bsnm{Ding},~\bfnm{Ruyi}\binits{R.}},
  \bauthor{\bsnm{Zhang},~\bfnm{Minghe}\binits{M.}},
  \bauthor{\bsnm{Van~Hentenryck},~\bfnm{Pascal}\binits{P.}} \AND
  \bauthor{\bsnm{Xie},~\bfnm{Yao}\binits{Y.}}
(\byear{2020}b).
\btitle{Spatio-Temporal Point Processes with Attention for Traffic Congestion
  Event Modeling}.
\bjournal{arXiv preprint arXiv:2005.08665}.
\end{barticle}
\endbibitem

\end{thebibliography}

%
%
%

\clearpage

\begin{appendix}
\section{Proof}\label{appendix:MLE_proof}

Note that $F_{\omega^N}(x)$, the gradient of MLE objective, can be derived as
\[F_{\omega^N}(\boldsymbol\beta)=\nabla L(\boldsymbol \beta)= \frac{1}{N} \sum_{t=1}^{N} \eta\left(\omega_{t-d}^{t-1}\right) \theta\left(\eta^{T}\left(\omega_{t-d}^{t-1}\right) x, \omega_t\right),\]
with
\[\theta(z, \omega)=\nabla_{z} \mathcal{L}_{w}(z)=-\sum_{k=1}^{K}\left[\sum_{p=1}^{M} \frac{[w]_{k p}}{[z]_{k p}} e^{k p}-\frac{1-\sum_{p=1}^{M}[w]_{k p}}{1-\sum_{p=1}^{M}[z]_{k p}} \sum_{p=1}^{M} e^{k p}\right]
\]
Thus, defining $\xi_t:=\eta\left(\omega_{t-d}^{t-1}\right)\theta\left(\eta^{T}\left(\omega_{t-d}^{t-1}\right) \beta, \omega_{t}\right)$, we have $\mathbb{E}[\xi_t]=0$. Therefore, $\{\xi_t\}$ forms a martingale-difference sequence and $F_{\omega^N}(x)=\frac{1}{N} \sum^N_{t=1} \xi_t.$ $F_{\omega^N}(x)$ will be abbreviated as $F(x)$ from now on.

To make sure $F(\boldsymbol\beta)$ is continuous, we also require $\boldsymbol\beta \in B_{\rho}$, for $B_{\rho}$ being a $\rho$-strengthened version of constraints (\ref{eq:multi_constraint}) as follows:

\begin{equation} \label{eq:ml_constraint}
    B_{\rho}:=\left\{
    \begin{array}{l}
         \varrho \leq \beta_{k}(p)+\sum_{s=1}^{d} \sum_{\ell=1}^{K} \min _{0 \leq q \leq M} \beta_{k \ell}^{s}(p, q), 1 \leq p \leq M, 1 \leq k \leq K \\
         1-\varrho \geq \sum_{p=1}^{M-1} \beta_{k}(p)+\sum_{s=1}^{d} \sum_{\ell=1}^{K} \max _{0 \leq q \leq M} \sum_{p=1}^{M} \beta_{k \ell}^{s}(p, q), 1 \leq k \leq K.
    \end{array}\right\}
\end{equation}

Below is the bound on the $l_{\infty}$ norm of $F(\boldsymbol\beta)$ that allow us to analyze the performance guarantee of our estimate $\hat{\boldsymbol \beta}^{\rm{MLE}}$ as in Theorem \ref{thm:ML-bound}. 

\begin{lemma}\label{lem:main_lemma}
For all $\epsilon \in(0,1)$ vector $F(\boldsymbol\beta)$ satisfies

\[\operatorname{Prob}_{\omega^{N}}\left\{\left\|F(\boldsymbol\beta)\right\|_{\infty} \geq \Theta \sqrt{\frac{2\ln (2 \kappa / \epsilon)}{N}}\right\} \leq \epsilon. \rho \in (0,1)\]

where $\Theta$ is a bound on $||\xi_t||_{\infty}$

\end{lemma}

\begin{proof}

We first provide a simple bound on $||\xi_t||_{\infty}$ and then link $\Theta$ to $||\xi_t||_{\infty}$ under MLE with identity-link.

First, 
\begin{align*}
    ||\xi_t||_{\infty}
    &=||\eta\left(\omega_{t-d}^{t-1}\right)\theta\left(\eta^{T}\left(\omega_{t-d}^{t-1}\right) \beta, \omega_{t}\right) ||_{\infty}\\
    &\leq||\theta\left(\eta^{T}\left(\omega_{t-d}^{t-1}\right) \beta, \omega_{t}\right)||_{\infty} && \text{By the definition of $\eta$}\\ 
    & = ||\sum_{k=1}^{K}\left[\sum_{p=1}^{M} \frac{[w]_{k p}}{[z]_{k p}} e^{k p}-\frac{1-\sum_{p=1}^{M}[w]_{k p}}{1-\sum_{p=1}^{M}[z]_{k p}} \sum_{p=1}^{M} e^{k p}\right]||_{\infty} &&\text{\makecell[l]{Let $z:=\eta^T\beta, \omega:=\omega_t$ \\ for simpler notation}}\\
    & \leq ||\sum_{k=1}^{K}\sum_{p=1}^{M} \frac{[w]_{k p}}{[z]_{k p}} e^{k p}||_{\infty} &&\text{$1 \geq \sum_{p=1}^{M}[w]_{k p} \& \sum_{p=1}^{M}[z]_{k p}$ } \\
    & \leq ||\sum_{k=1}^{K}\sum_{p=1}^{M} \frac{1}{\rho} e^{k p}||_{\infty} && \text{$0 \leq [\omega]_{kp}\leq 1$ and $[z]_{kp} \geq \rho$ by definition}\\
    &=1/\rho &&\text{$\sum_{k=1}^{K}\sum_{p=1}^{M} e^{k p}$ is the vector of all 1's in $\mathbf{R}^{K M}$}
\end{align*}

Also, denoting by $\mathbf{E}_{\mid \omega^{t}}$ the conditional expectation given $\omega^t$ and $i=1,\ldots,\kappa$, we have

$$
\begin{aligned}
\mathbf{E}_{\omega^{t+1}}\left\{\exp \left\{\sum_{s=1}^{t+1} \gamma\left[\xi_{s}\right]_{i}\right\}\right\} &=\mathbf{E}_{\omega^{t}}\left\{\exp \left\{\sum_{s=1}^{t} \gamma\left[\xi_{s}\right]_{i}\right\} \mathbf{E}_{\omega_{t+1} \mid \omega^{t}}\left\{\exp \left\{\gamma\left[\xi_{t+1}\right]_{i}\right\}\right\}\right\} \\
& \leq \mathbf{E}_{\omega^{t}}\left\{\exp \left\{\gamma \sum_{s=1}^{t}\left[\xi_{s}\right]_{i}\right\} \exp \left\{\gamma^{2} / 2\rho^2\right\}\right\},
\end{aligned}
$$

where the last inequality is due to the Hoeffding's inequality and, from the above bound that the conditional, $\omega^t$ given, distribution of $\left[\xi_{t+1}\right]_{i}$ is zero mean and is supported on $[-1/\rho, 1/\rho]$. By induction, we have

$$
\mathbf{E}_{\omega^{N}}\left\{\exp \left\{\gamma\left[\sum_{t=1}^{N} \xi_{t}\right]_{i}\right\}\right\} \leq \exp \left\{N \gamma^{2} / 2\rho^2\right\}.
$$

Now using Chernoff bound, we have

$$
\begin{aligned}
\operatorname{Prob}_{\omega^{N}}\left\{[F(\boldsymbol\beta)]_i>\theta\right\}
&=\operatorname{Prob}_{\omega^{N}}\left\{\frac{1}{N}\left[\sum_{t=1}^{N} \xi_{t}\right]_{i}>\theta\right\} \\
& \leq \exp \{-\mu \theta\} \mathbf{E}_{\omega^{N}}\left\{\exp \left\{\mu \frac{1}{N}\left[\sum_{t=1}^{N} \xi_{t}\right]_{i}\right\}\right\} \\
& \leq \exp \left\{-\mu \theta+\frac{\mu^{2}}{2 \Tilde{N}}\right\}, 
\end{aligned}
$$

where $\Tilde{N}=N\rho^2$. Now, if we let $\theta=\gamma / \sqrt{\tilde{N}}$ and $\mu=\tilde{N} \theta,$ we simplify the above bound and use the union bound to get
$$
\operatorname{Prob}_{\omega^{N}}\left\{\|F(\boldsymbol\beta)\|_{\infty}>\gamma / \sqrt{N} \rho\right\} \leq 2 \kappa \exp \left\{-\gamma^{2} / 2\right\}, \quad \forall \gamma \geq 0
$$
Finally, we let $\gamma=\sqrt{2 \ln (2 \kappa / \epsilon)}$ and $\Theta=1 / \rho,$ where $1 / \rho$ is the bound on $\left\|\xi_{t}\right\|_{\infty} .$ Further
simplification then yields the following result that is easier to use:
\[
\operatorname{Prob}_{\omega^{N}}\left\{\left\|F(\boldsymbol\beta)\right\|_{\infty} \geq \frac{1}{\rho} \sqrt{\frac{2\ln (2 \kappa / \epsilon)}{N}}\right\} \leq \epsilon
\]
\end{proof}

Now, we can use Lemma \ref{lem:main_lemma} and $\theta_p(A)=\theta_p(\frac{1}{N}\sum_{t=1}^N \eta(\omega^{t-1}_{t-d})\eta^T(\omega^{t-1}_{t-d}))$, which is defined in (\ref{def:theta}), to prove Theorem \ref{thm:ML-bound}. 

\begin{proof}[Proof for Theorem \ref{thm:ML-bound}]
For simplicity, let $\hat{\boldsymbol \beta}:=\hat{\boldsymbol \beta}^{\rm{MLE}}$. Since $F(x)$ is the gradient of the convex and continuously differentiable log-likelihood $L_{\omega^N}(x)$, $F(x)$ is continuous and monotone. Thus, the estimator $\hat{\boldsymbol \beta}$ is a weak and strong solution to the variational inequality $VI[F,B_{\rho}]$ for $B_{\rho}$ defined in (\ref{eq:ml_constraint}). As a result, $\left< F(\boldsymbol \beta) -F (\hat{\boldsymbol \beta}),\boldsymbol \beta -\hat{\boldsymbol \beta}  \right> \geq 0$ and $F(\hat{\boldsymbol \beta})=0$. In fact, under our $\rho$-strengthened constraint on $\boldsymbol \beta$, we can prove that $\Tilde{F}(\eta^T x)$, which is defined so that $\eta \Tilde{F}(\eta^T x) = F(x)$ with $\eta\eta^T=\frac{1}{N}\sum_{t=1}^N \eta(\omega^{t-1}_{t-d})\eta^T(\omega^{t-1}_{t-d})$, satisfies $$\left< \Tilde{F}(\eta^T \boldsymbol \beta) -\Tilde{F}(\eta^T \hat{\boldsymbol \beta}),\eta^T(\boldsymbol \beta -\hat{\boldsymbol \beta)}  \right> \geq (1-\rho)^{-2} ||\eta^T(\hat{\boldsymbol \beta} -\boldsymbol \beta) ||^2_2. $$

Now, by the earlier definition of $\theta_p(A)=\theta_p(\eta\eta^T)$, we can see that $||\eta^T(\hat{\boldsymbol \beta} -\boldsymbol \beta) ||^2_2=(\hat{\boldsymbol \beta}^T -\boldsymbol \beta^T)A(\hat{\boldsymbol \beta} -\boldsymbol \beta) \geq \sqrt{\theta_1(A)\theta_p(A)} ||\hat{\boldsymbol \beta} -\boldsymbol \beta ||_1||\hat{\boldsymbol \beta} -\boldsymbol \beta ||_p$.

As a result, 
\begin{align*}
    ||F(\boldsymbol \beta)||_{\infty} ||\hat{\boldsymbol \beta} -\boldsymbol \beta ||_1
    &\geq \left< F(\boldsymbol \beta) -F (\hat{\boldsymbol \beta}),\boldsymbol \beta -\hat{\boldsymbol \beta}  \right>\\
    &=\left< \Tilde{F}(\eta^T \boldsymbol \beta) -\Tilde{F}(\eta^T\hat{ \boldsymbol \beta}),\eta^T(\boldsymbol \beta -\hat{\boldsymbol \beta)}  \right>\\  
    &\geq (1-\rho)^{-2} ||\eta^T(\hat{\boldsymbol \beta} -\boldsymbol \beta) ||^2_2\\
    & \geq \sqrt{\theta_1(A)\theta_p(A)} ||\hat{\boldsymbol \beta} -\boldsymbol \beta ||_1||\hat{\boldsymbol \beta} -\boldsymbol \beta ||_p
\end{align*}
We finish the proof by cancelling $||\hat{\boldsymbol \beta} -\boldsymbol \beta ||_1$ from the the final inequality and using Lemma \ref{lem:main_lemma} to bounding $||F(\boldsymbol \beta)||_{\infty}$.
\end{proof}

\clearpage
\section{Additional experiments}\label{appendix:exper}

\subsection{Multi-state ramping event}\label{sec:multi_ramp}

\textbf{} 

\noindent \textit{1. Visualize MLE influences on terrain map}

\begin{figure*}[htbp]
  \centering
  \subfigure[$s=1$]{\includegraphics[scale=0.41]{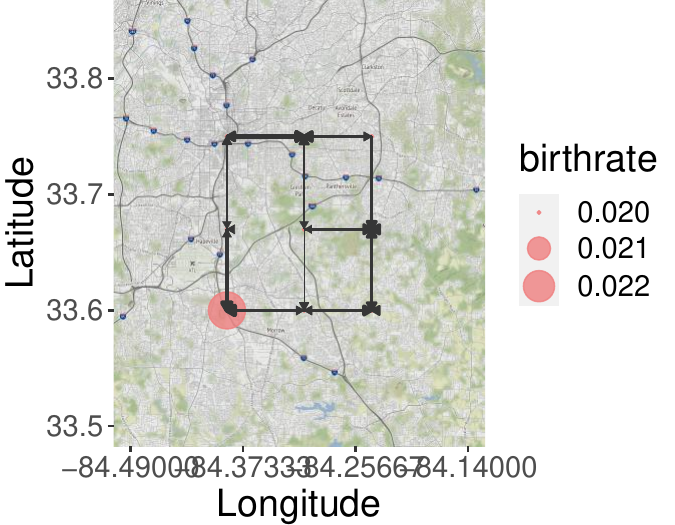}}
  \subfigure[$s=5$]{\includegraphics[scale=0.41]{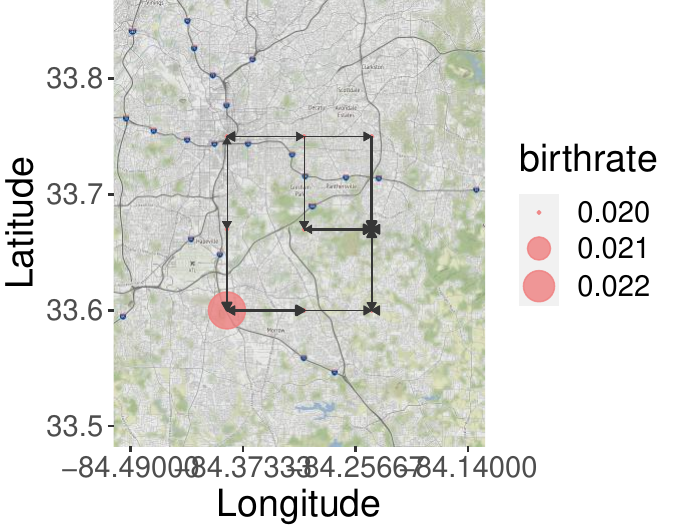}}
  \subfigure[$s=10$]{\includegraphics[scale=0.41]{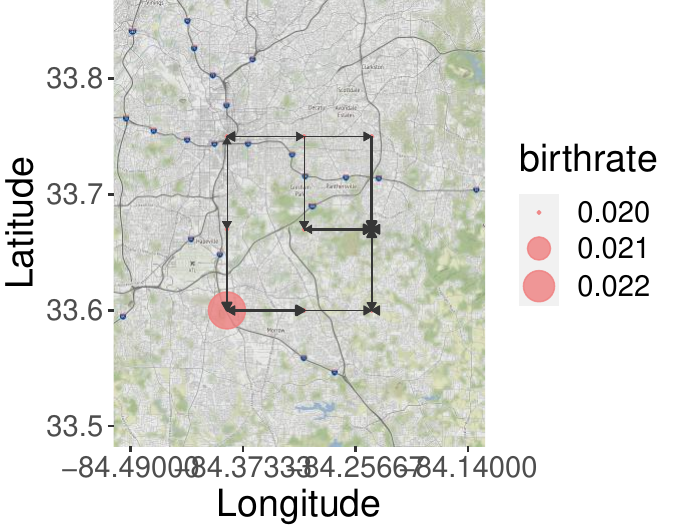}}
  \subfigure[$s=1$]{\includegraphics[scale=0.41]{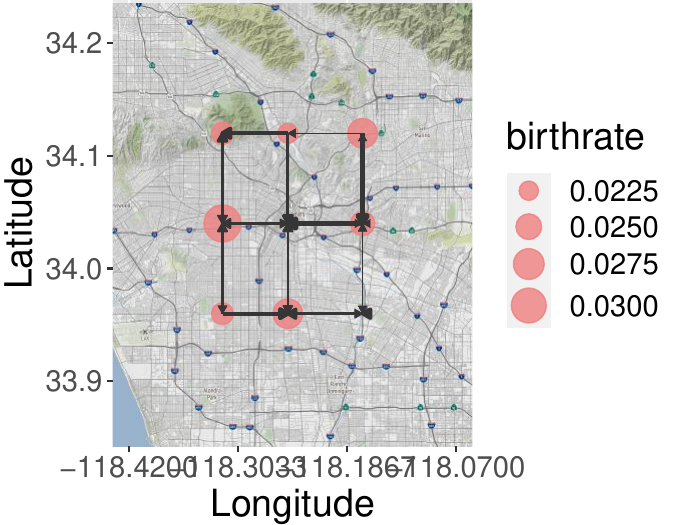}}
  \subfigure[$s=3$]{\includegraphics[scale=0.41]{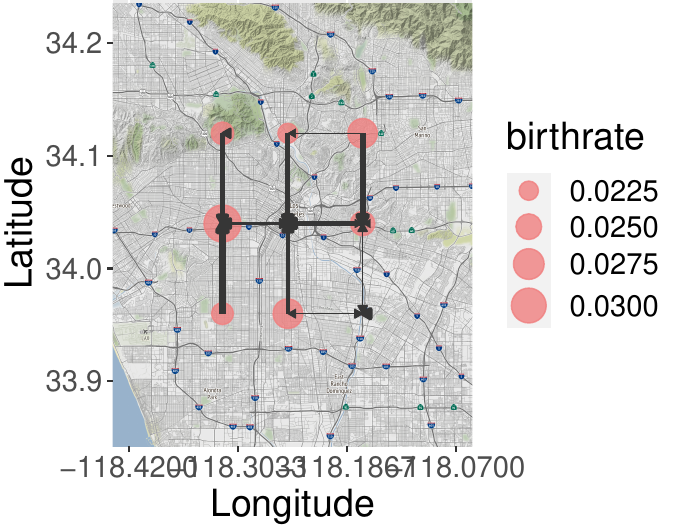}}
  \subfigure[$s=5$]{\includegraphics[scale=0.41]{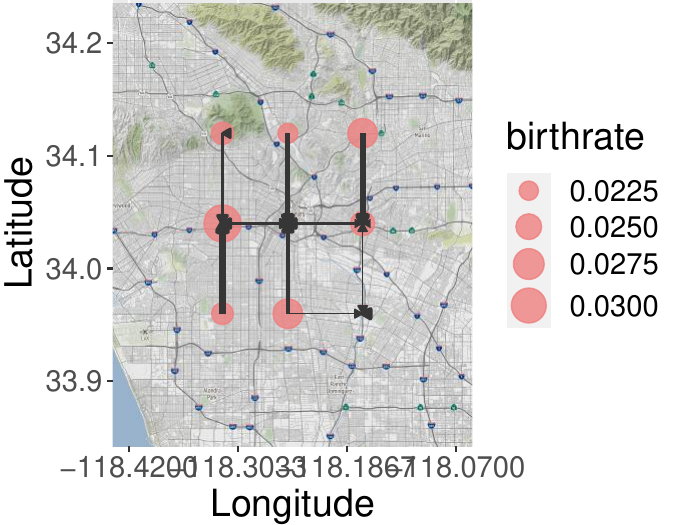}}
  \subfigure[$s=1$]{\includegraphics[scale=0.37]{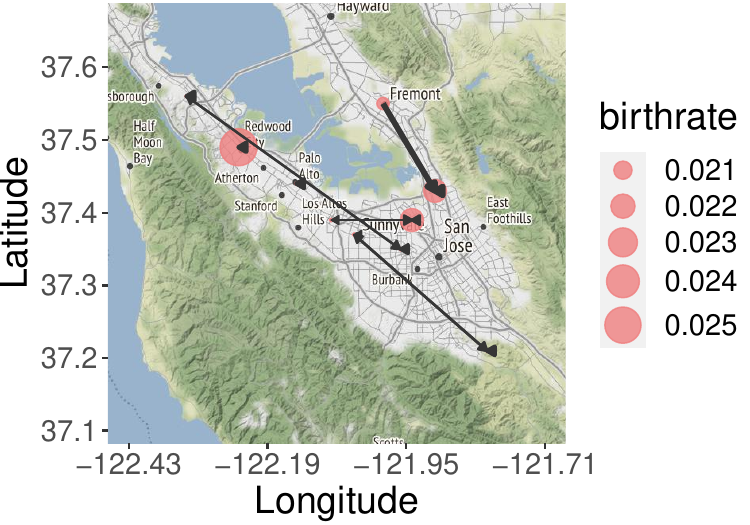}}
  \subfigure[$s=5$]{\includegraphics[scale=0.37]{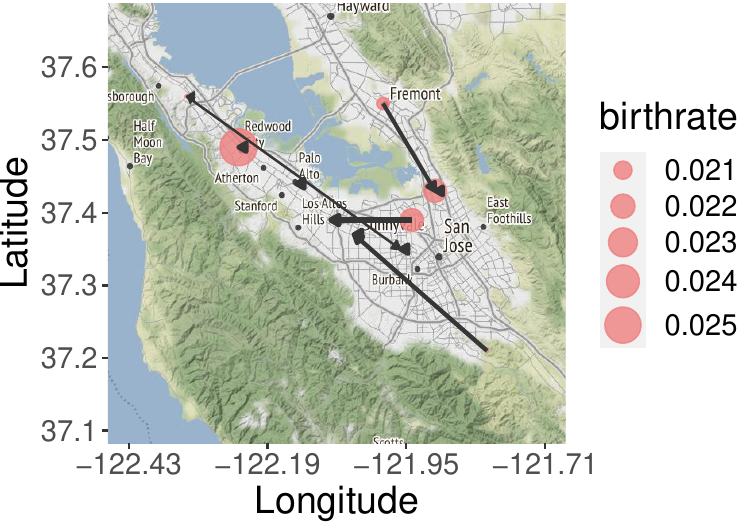}}
  \subfigure[$s=10$]{\includegraphics[scale=0.37]{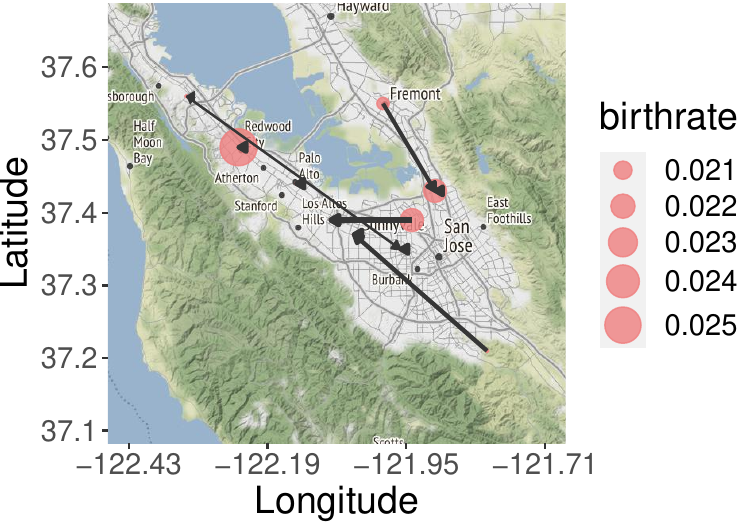}}
    \caption{Multi-state positive-to-positive ($p=q=1$) recovery results. To make sure the edges are visible when $s>1$, we magnify edge weight of Atlanta (a-c) and Los Angeles (d-f) estimates 3 times and North California (g-i) estimates 5 times.}
    \label{fig:5_2_grid_multi_pos}
\end{figure*}

\begin{figure*}[htbp]
  \centering
  \subfigure[$s=1$]{\includegraphics[scale=0.41]{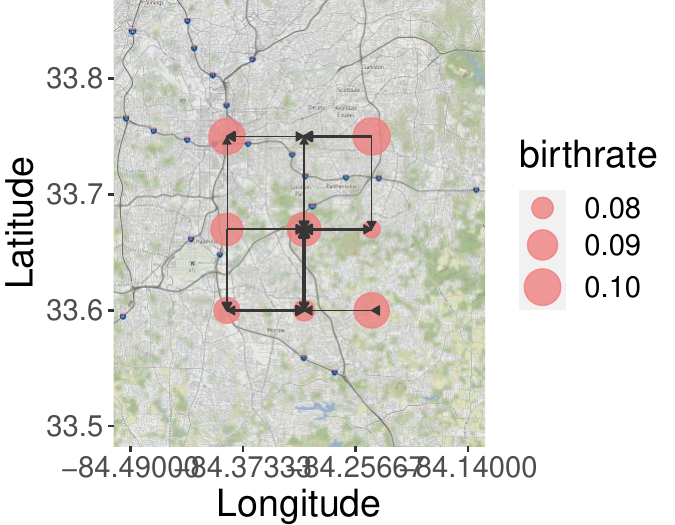}}
  \subfigure[$s=5$]{\includegraphics[scale=0.41]{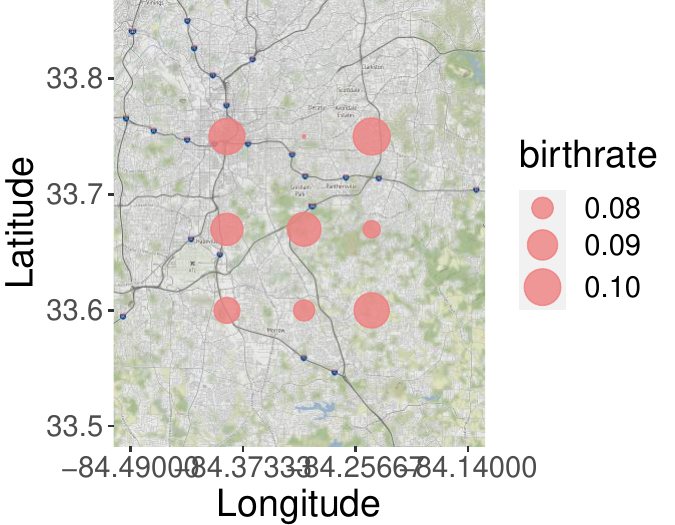}}
  \subfigure[$s=10$]{\includegraphics[scale=0.41]{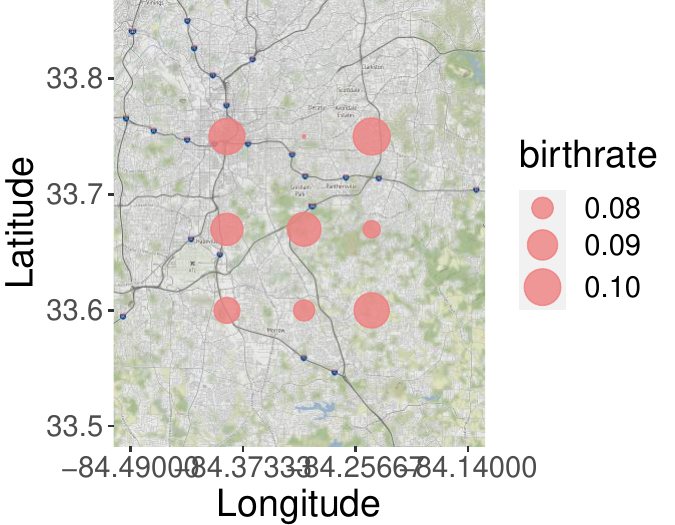}}
  \subfigure[$s=1$]{\includegraphics[scale=0.41]{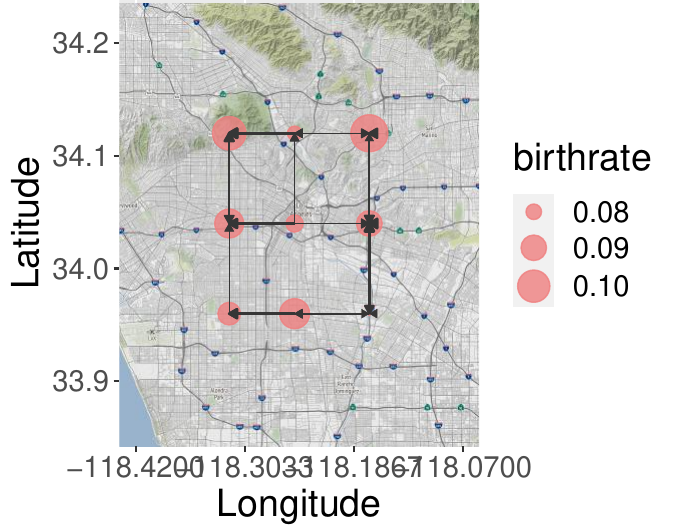}}
  \subfigure[$s=3$]{\includegraphics[scale=0.41]{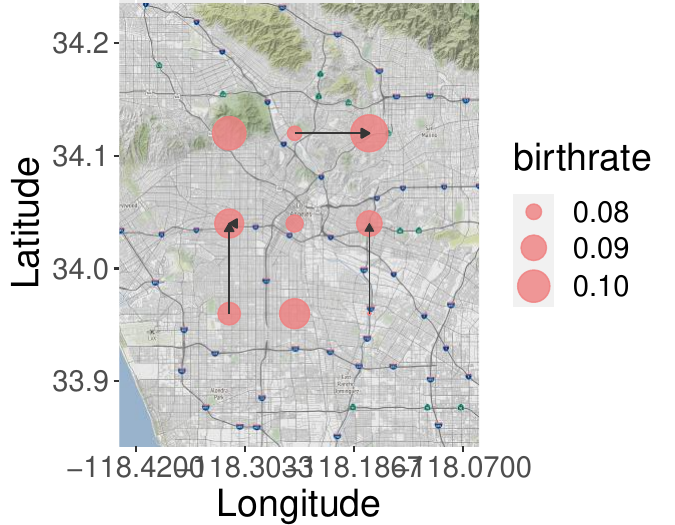}}
  \subfigure[$s=5$]{\includegraphics[scale=0.41]{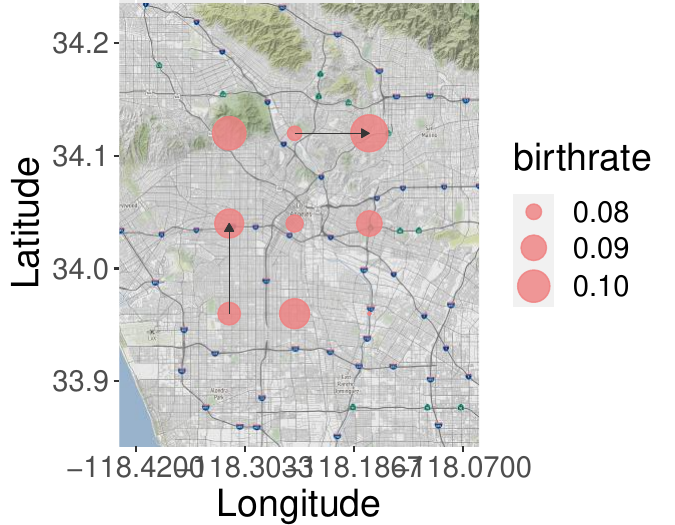}}

  \subfigure[$s=1$]{\includegraphics[scale=0.37]{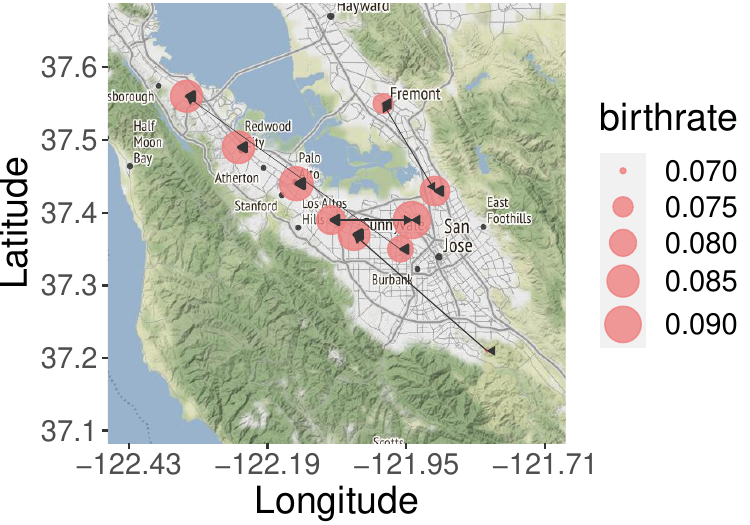}}
  \subfigure[$s=5$]{\includegraphics[scale=0.37]{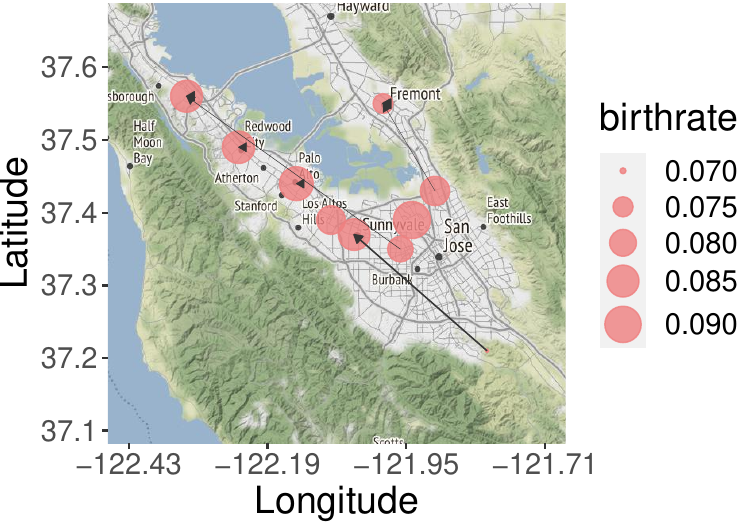}}
  \subfigure[$s=10$]{\includegraphics[scale=0.37]{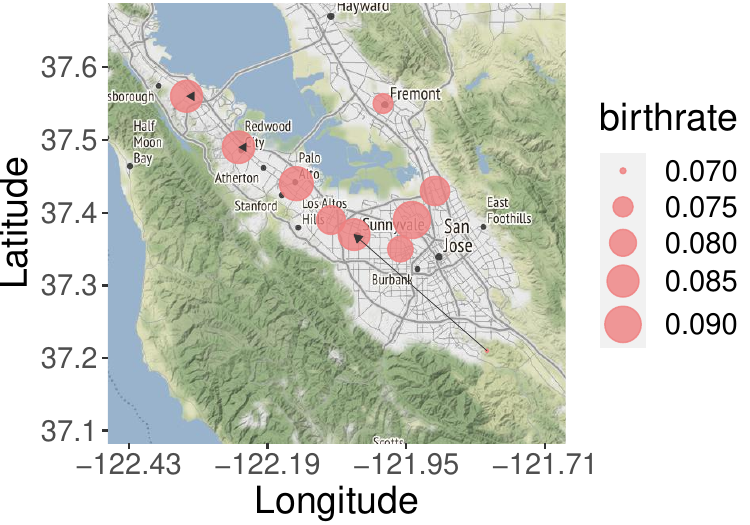}}
    \caption{Multi-state negative-to-negative ($p=q=2$) recovery result. To make sure the edges are visible when $s>1$, we magnify edge weight of Atlanta (a-c) estimates 25 times, North California (g-i) estimates 5 times, and Los Angeles (d-f) estimates 6 times.}
    \label{fig:5_2_grid_multi_neg}
\end{figure*}
\clearpage
\noindent \textit{2. Prediction Performances}

\begin{table}[htbp]
    \centering
    \caption{The positive case for multi-state model: precision, recall, and $F_1$ score in three downtown under static vs. dynamic threshold after tuning. The highest value among the four methods (LS or MLE combined with static or dynamic threshold) is in bold.}
    \label{tab:5_2_multi_p2p_mod_accuracy}
    \resizebox{\linewidth}{!}{\begin{tabular}{p{2cm}p{2cm}p{2cm}p{2cm}p{2cm}p{2cm}p{2cm}p{2cm}}
     \hline
     &&\multicolumn{3}{c}{Least Square} &\multicolumn{3}{c}{Maximum Likelihood} \\
     Location & $\tau$ & Precision & Recall & $F_1$ & Precision & Recall & $F_1$\\
     \hline
\multirow{2}{*}{Atlanta} & Static & 0.91 & 1.00 & 0.95 & 0.87 & 0.70 & 0.78 \\
 & Dynamic & \textbf{0.91} & \textbf{1.00} & \textbf{0.95} & 0.87 & 0.70 & 0.78 \\
\multirow{2}{*}{Los Angeles} & Static & 1.00 & 0.79 & 0.88 & 0.96 & 0.79 & 0.86 \\
 & Dynamic & 0.82 & \textbf{1.00} & 0.90 & 0.84 & \textbf{0.96} & \textbf{0.90} \\
\multirow{2}{*}{Palo Alto} & Static & 0.90 & 0.76 & 0.83 & 0.90 & 0.76 & 0.83 \\
 & Dynamic & 0.81 & 1.00 & 0.89 & \textbf{0.93} & \textbf{1.00} & \textbf{0.96} \\
     \hline
    \end{tabular}}
\end{table}

\begin{table}[htbp]
    \centering
    \caption{The negative case for multi-state model: precision, recall, and $F_1$ score in three downtown under static vs. dynamic threshold after tuning. The highest value among the four methods (LS or MLE combined with static or dynamic threshold) is in bold.}
    \label{tab:5_2_multi_n2n_mod_accuracy}
    \resizebox{\linewidth}{!}{\begin{tabular}{p{2cm}p{2cm}p{2cm}p{2cm}p{2cm}p{2cm}p{2cm}p{2cm}}
     \hline
     &&\multicolumn{3}{c}{Least Square} &\multicolumn{3}{c}{Maximum Likelihood} \\
     Location & $\tau$ & Precision & Recall & $F_1$ & Precision & Recall & $F_1$\\
     \hline
\multirow{2}{*}{Atlanta} & Static & \textbf{1.00} & 0.98 & \textbf{0.99} & 0.95 & 0.95 & 0.95 \\
 & Dynamic & 0.95 & \textbf{0.98} & 0.96 & 0.91 & 0.95 & 0.93 \\
\multirow{2}{*}{Los Angeles} & Static & 0.18 & \textbf{1.00} & 0.31 & 0.80 & 0.87 & \textbf{0.84} \\
 & Dynamic & 0.14 & 0.70 & 0.23 & \textbf{0.90} & 0.74 & 0.81 \\
\multirow{2}{*}{Palo Alto} & Static & 1.00 & 1.00 & \textbf{1.00} & 0.85 & 1.00 & 0.92 \\
 & Dynamic & 0.86 & \textbf{1.00} & 0.93 & \textbf{1.00} & 0.98 & 0.99 \\
     \hline
    \end{tabular}}
\end{table}

\begin{figure*}[htbp]
    \centering
    \subfigure[Downtown Atlanta (Static)]{\includegraphics[scale=0.2]{{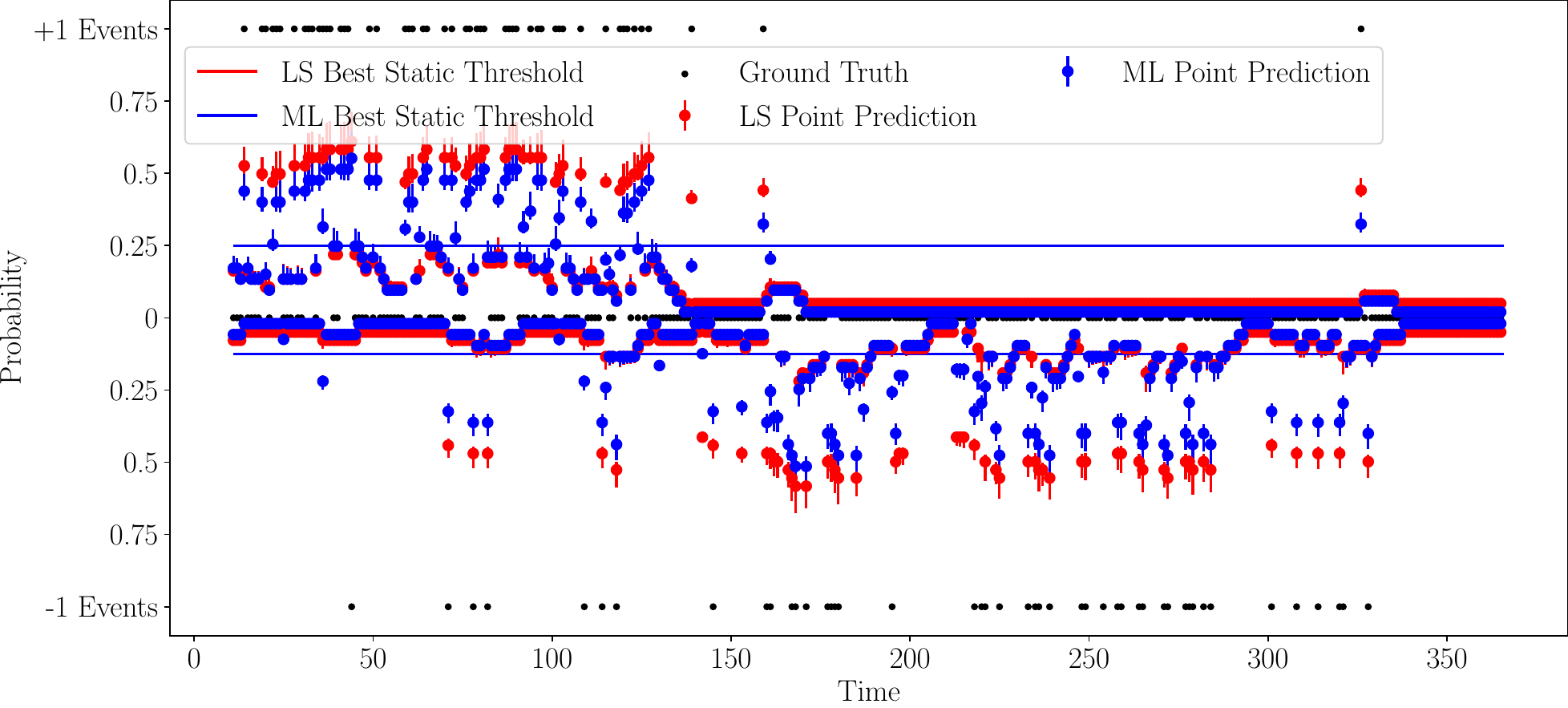}}}
    \subfigure[Downtown Atlanta (Dynamic)]{\includegraphics[scale=0.2]{{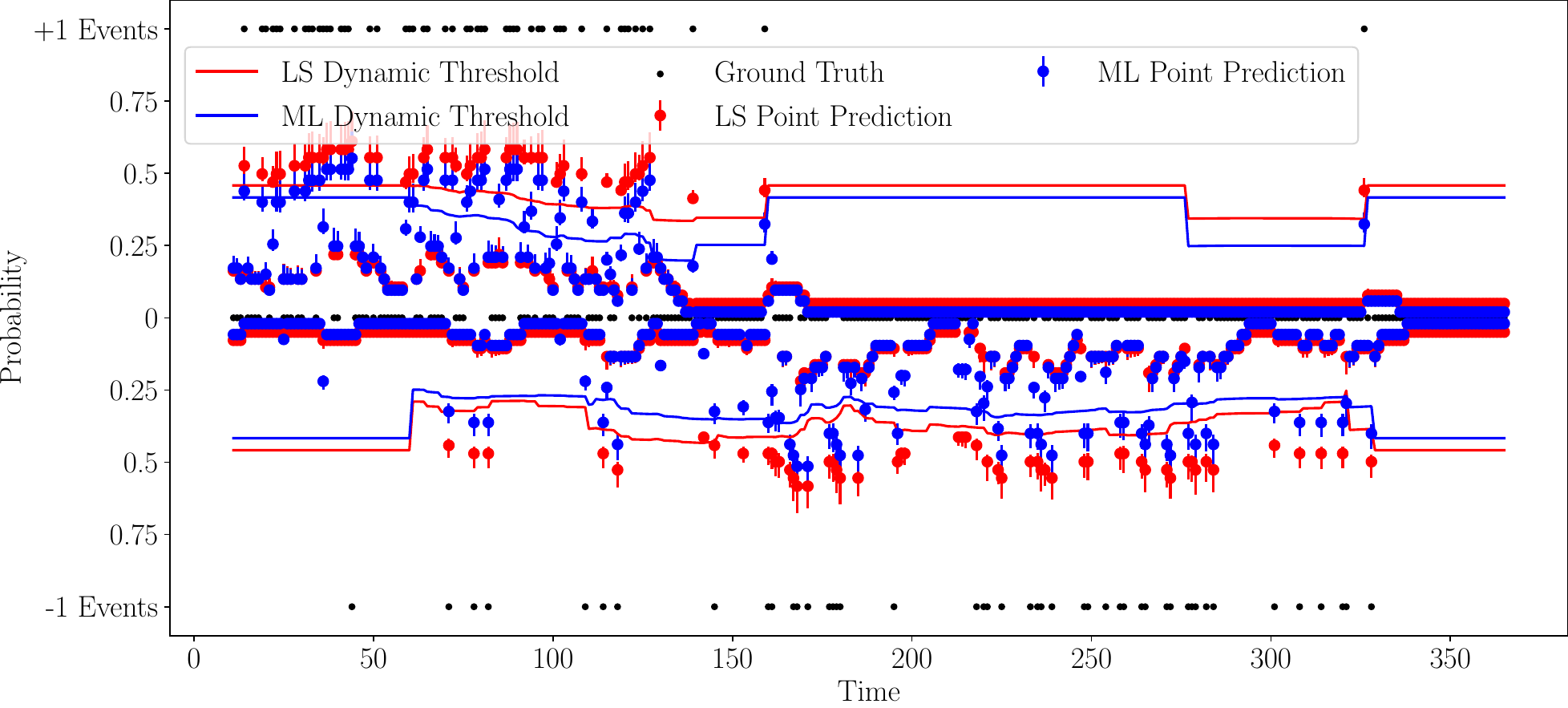}}}
    \subfigure[Downtown Los Angeles (Static)]{\includegraphics[scale=0.2]{{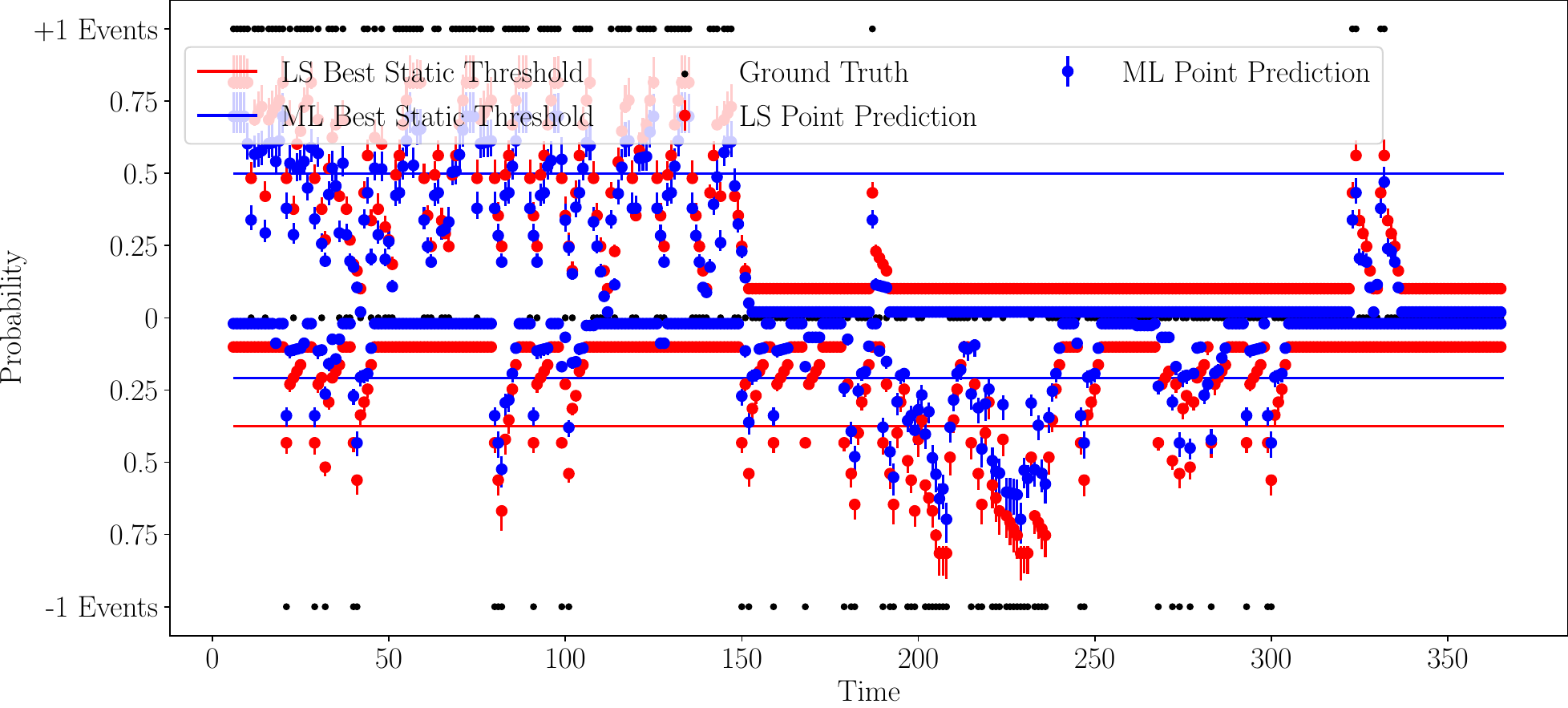}}}
    \subfigure[Downtown Los Angeles (Dynamic)]{\includegraphics[scale=0.2]{{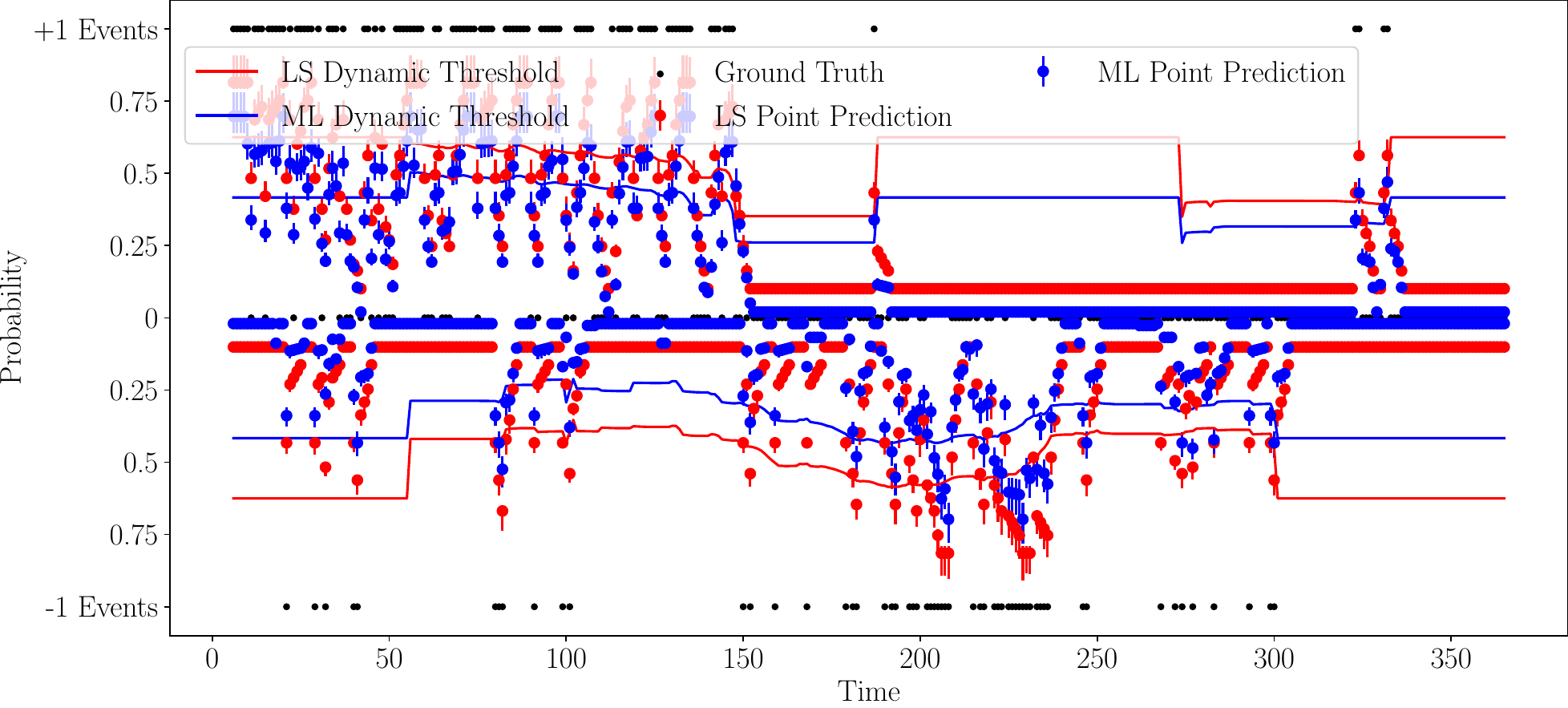}}}
    \subfigure[Palo Alto (Static)]{\includegraphics[scale=0.2]{{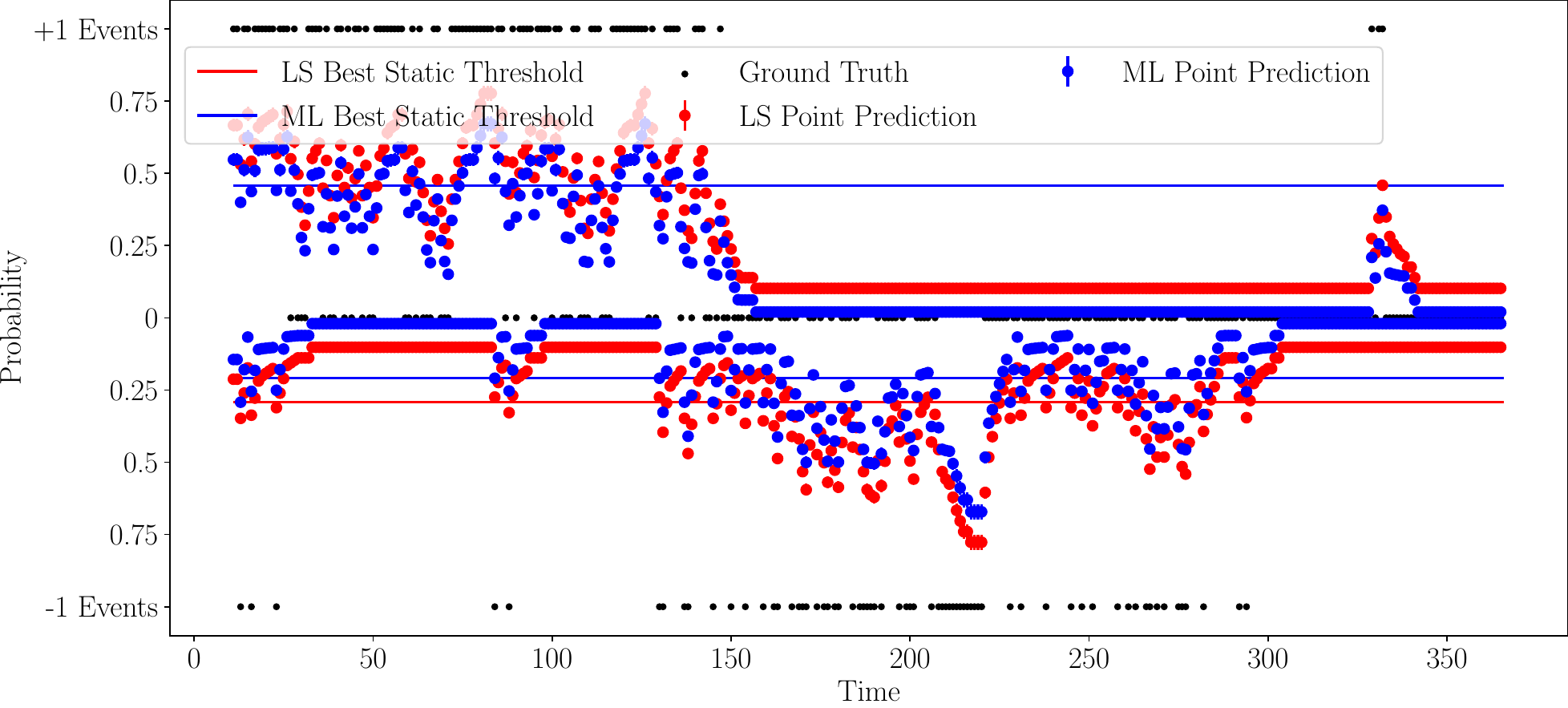}}}
    \subfigure[Palo Alto (Dynamic)]{\includegraphics[scale=0.2]{{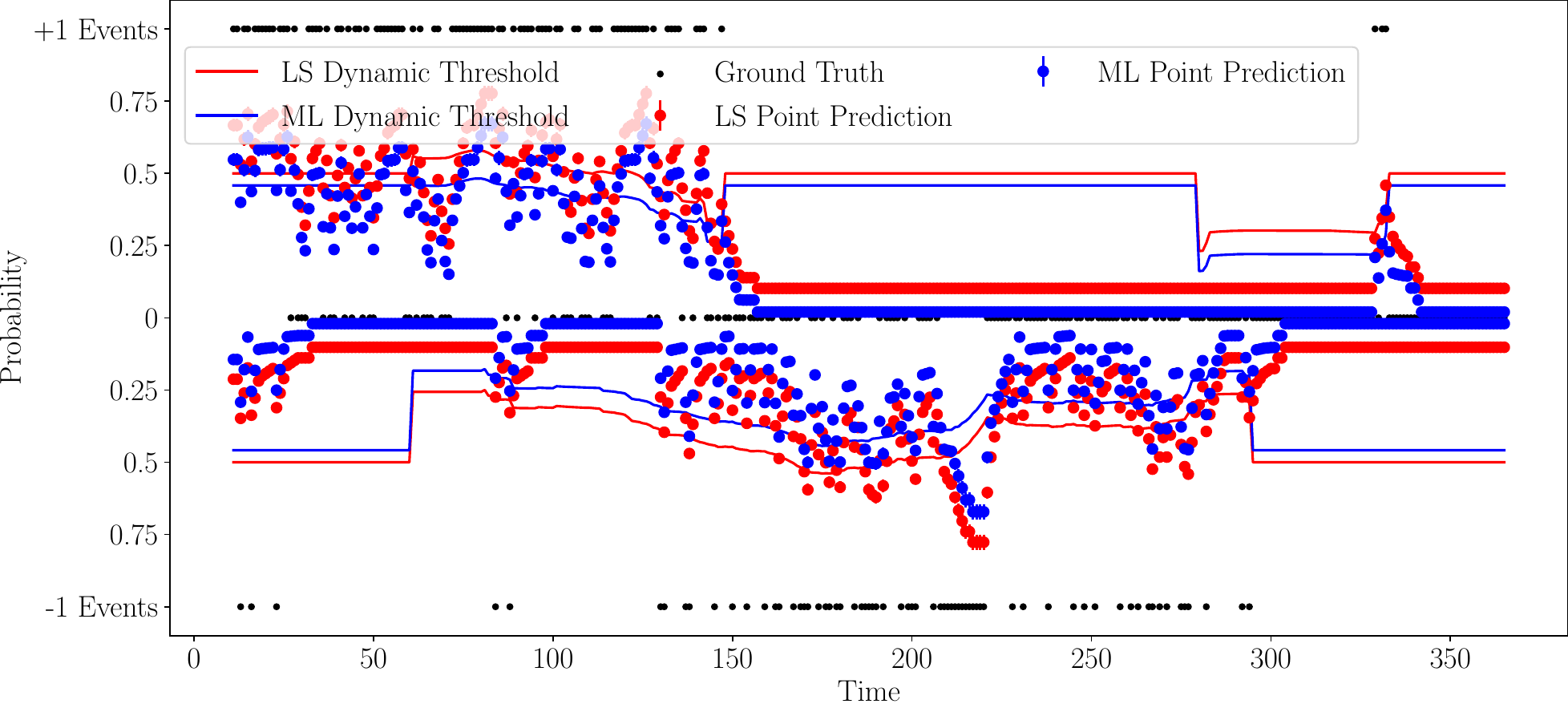}}}
    \vspace{-0.5cm}
    \caption{Prediction intervals for online point prediction of probabilities for multiple-state ramping events using LS (red dots) and ML (blue dots), compared with true ramping events (black dots). The left column is generated using a static threshold, and the right column uses dynamic thresholds.}
    \label{fig:5_2_point_pred_multi}
\end{figure*}

\subsection{Fitting seasonal models}\label{sec:season}
\textbf{}

\noindent \textit{1. Single-state results}

\begin{figure*}[htbp]
  \centering
  \includegraphics[scale=0.35]{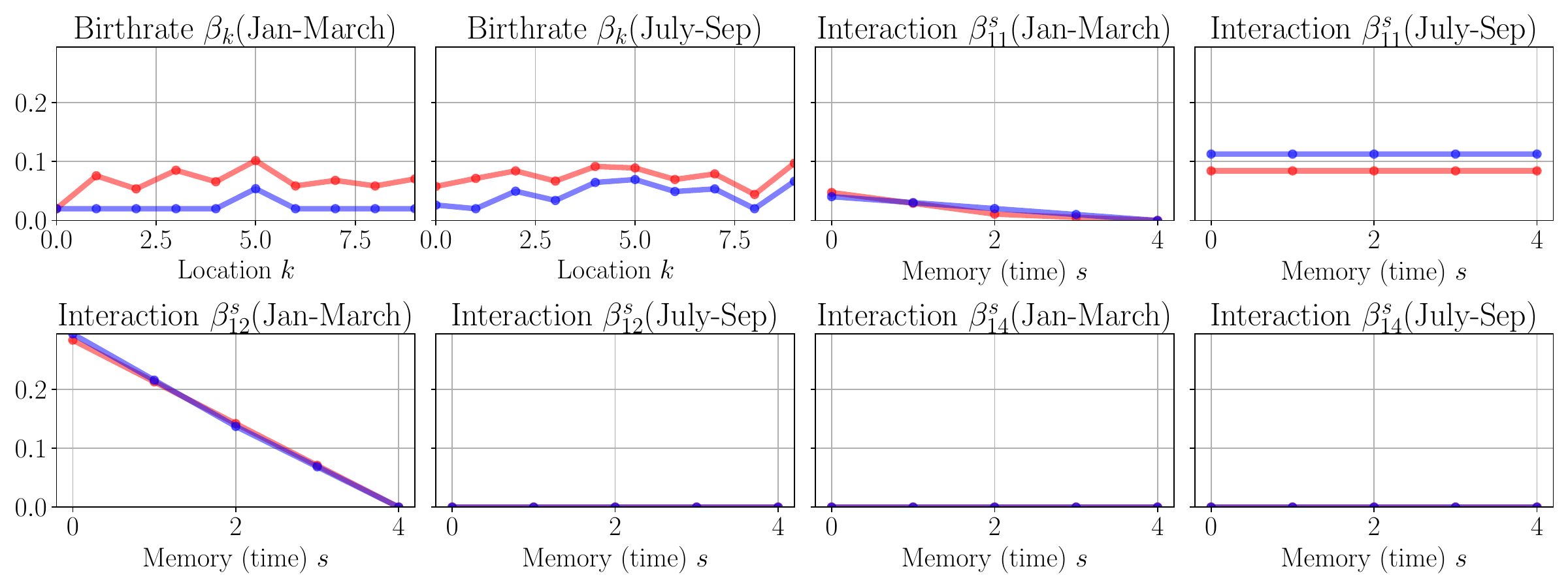}
  \caption{Recovered birthrate parameters at different locations and interaction parameters over time. The blue curve shows ML recoveries, and the red curve shows LS recoveries.}
  \label{fig:5_2_beta_single_rainfall}
\end{figure*}

\begin{figure*}[htbp]
  \centering
  \subfigure[$s=1$]{\includegraphics[scale=0.37]{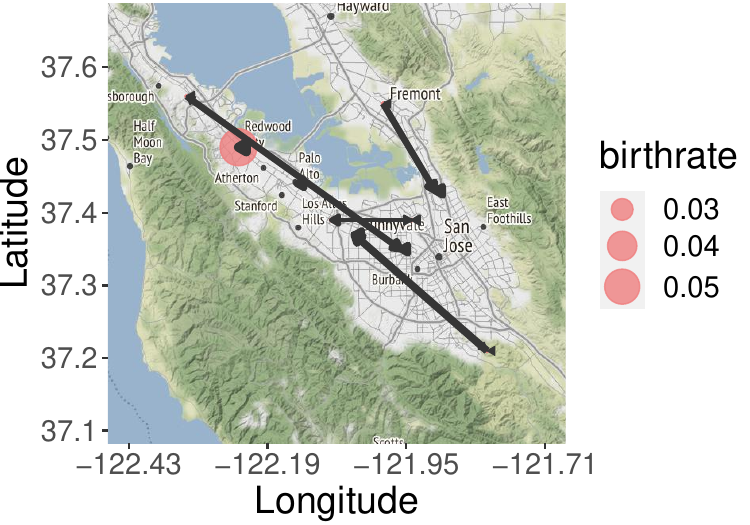}}
  \subfigure[$s=3$]{\includegraphics[scale=0.37]{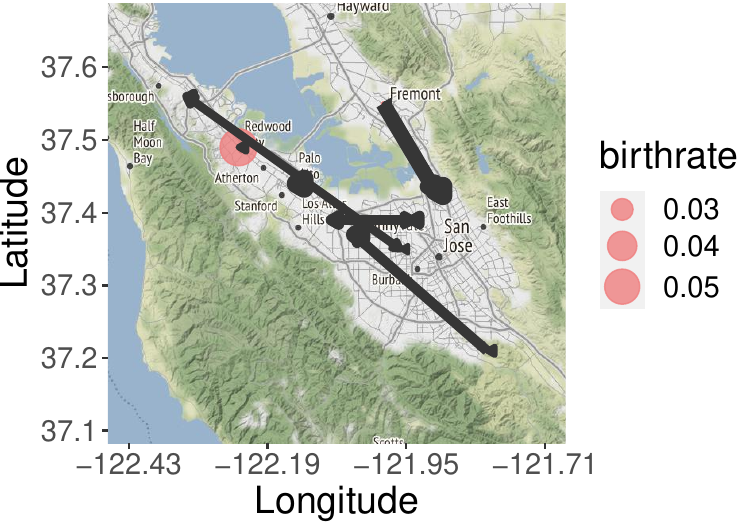}}
  \subfigure[$s=5$]{\includegraphics[scale=0.37]{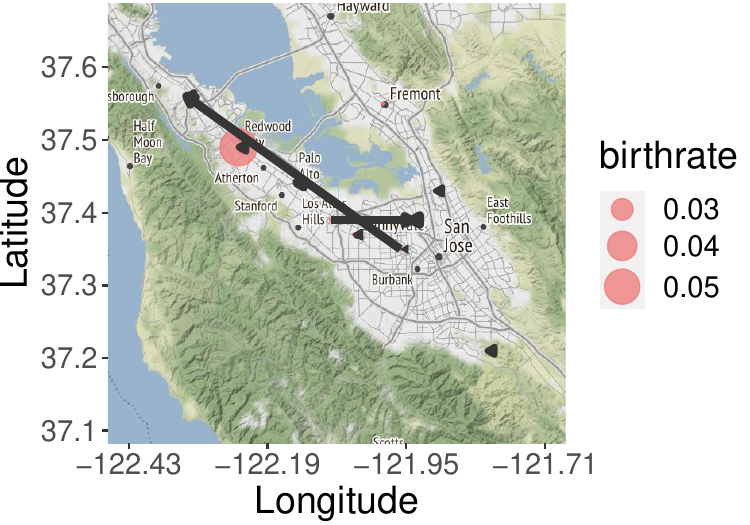}}
  \subfigure[$s=1$]{\includegraphics[scale=0.37]{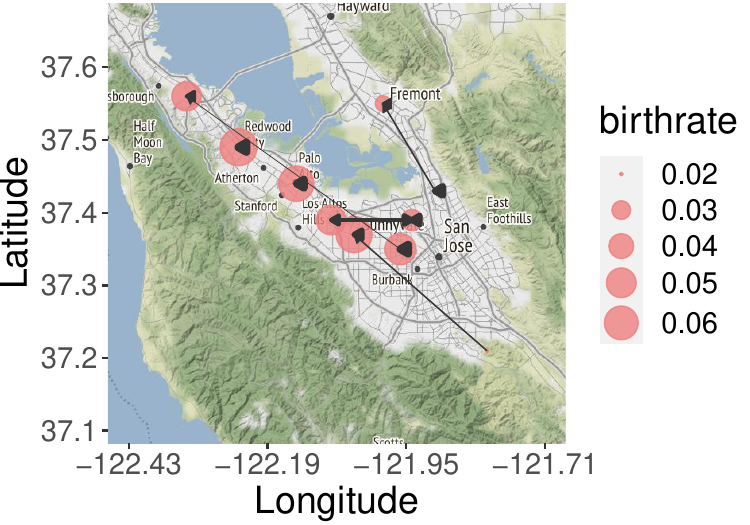}}
  \subfigure[$s=3$]{\includegraphics[scale=0.37]{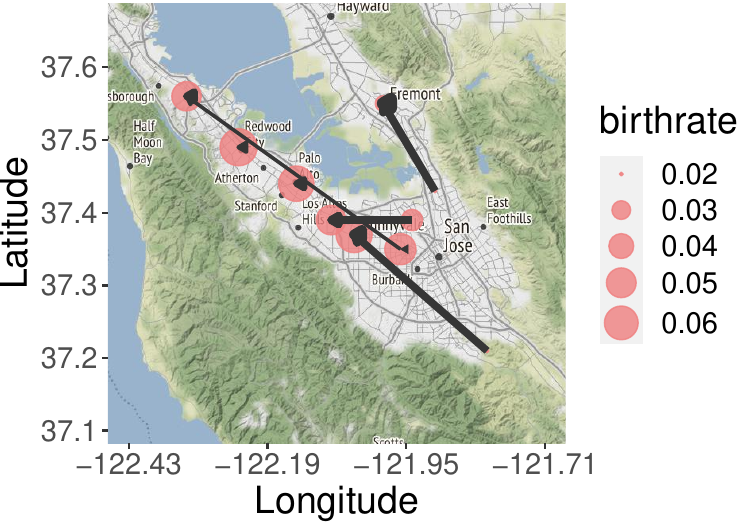}}
  \subfigure[$s=5$]{\includegraphics[scale=0.37]{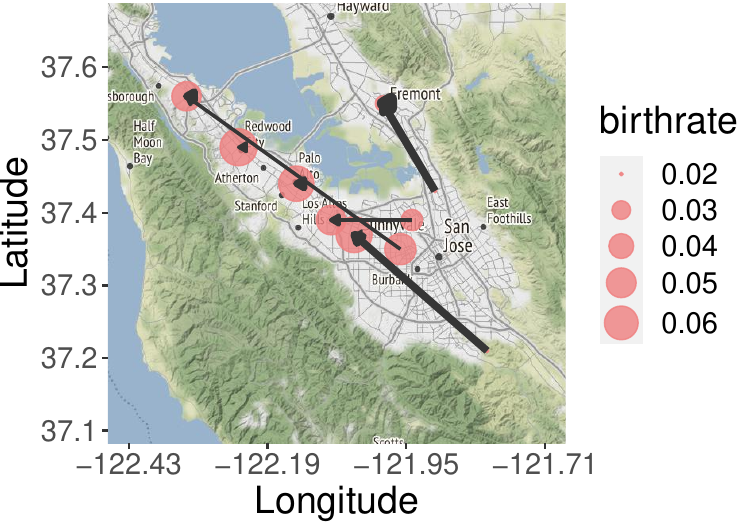}}
  \caption{Visualize single-state spatio-temporal influences. The top row shows estimates from data from 01/01/2017-03/31-2017, and the bottom row uses data during 07/01/2017-09/30-2017. To ensure the edges are visible when $s>1$, we magnify the edge weight of North California (g-i) estimates five times.}
  \label{fig:5_2_grid_single_CA}
\end{figure*}

\begin{table}[htbp]
    \centering
    \caption{Sequential prediction performance for single-state model by season: precision, recall, and $F_1$ score in Palo Alto under static vs. dynamic threshold after tuning. The higher value between static vs. dynamic is in bold.}
    \label{tab:5_2_mod_accuracy_rainfall}
    \resizebox{\linewidth}{!}{\begin{tabular}{p{2cm}p{2cm}p{2cm}p{2cm}p{2cm}p{2cm}p{2cm}p{2cm}}
     \hline
     &&\multicolumn{3}{c}{Least Square} &\multicolumn{3}{c}{Maximum Likelihood} \\
     Months & $\tau$ & Precision & Recall & $F_1$ & Precision & Recall & $F_1$\\
     \hline
\multirow{2}{*}{Jan-Mar} & Static & 0.40 & 1.00 & 0.57 & 0.40 & \textbf{1.00} & 0.57 \\
 & Dynamic & \textbf{1.00} & 0.77 & \textbf{0.87} & 0.91 & 0.81 & 0.86 \\
\multirow{2}{*}{July-Sept} & Static & 0.91 & 1.00 & 0.95 & 0.91 & \textbf{1.00} & 0.95 \\
 & Dynamic & \textbf{1.00} & 0.90 & \textbf{0.95} & 1.00 & 0.85 & 0.92 \\
     \hline
    \end{tabular}}
\end{table}

\begin{figure*}[htbp]
    \centering
    \subfigure[Jan-March (Static)]{\includegraphics[scale=0.21]{{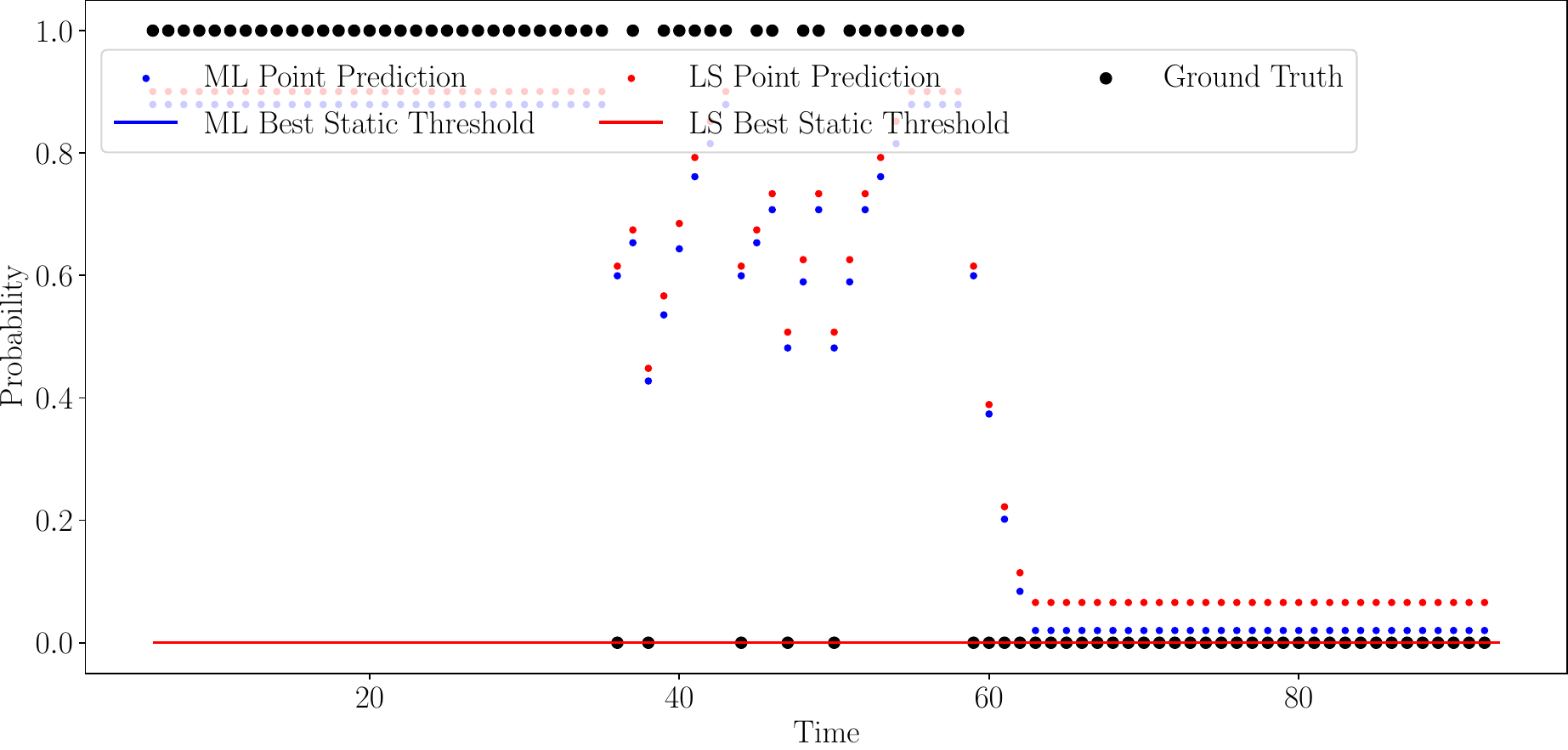}}}
    \subfigure[Jan-March (Dynamic)]{\includegraphics[scale=0.21]{{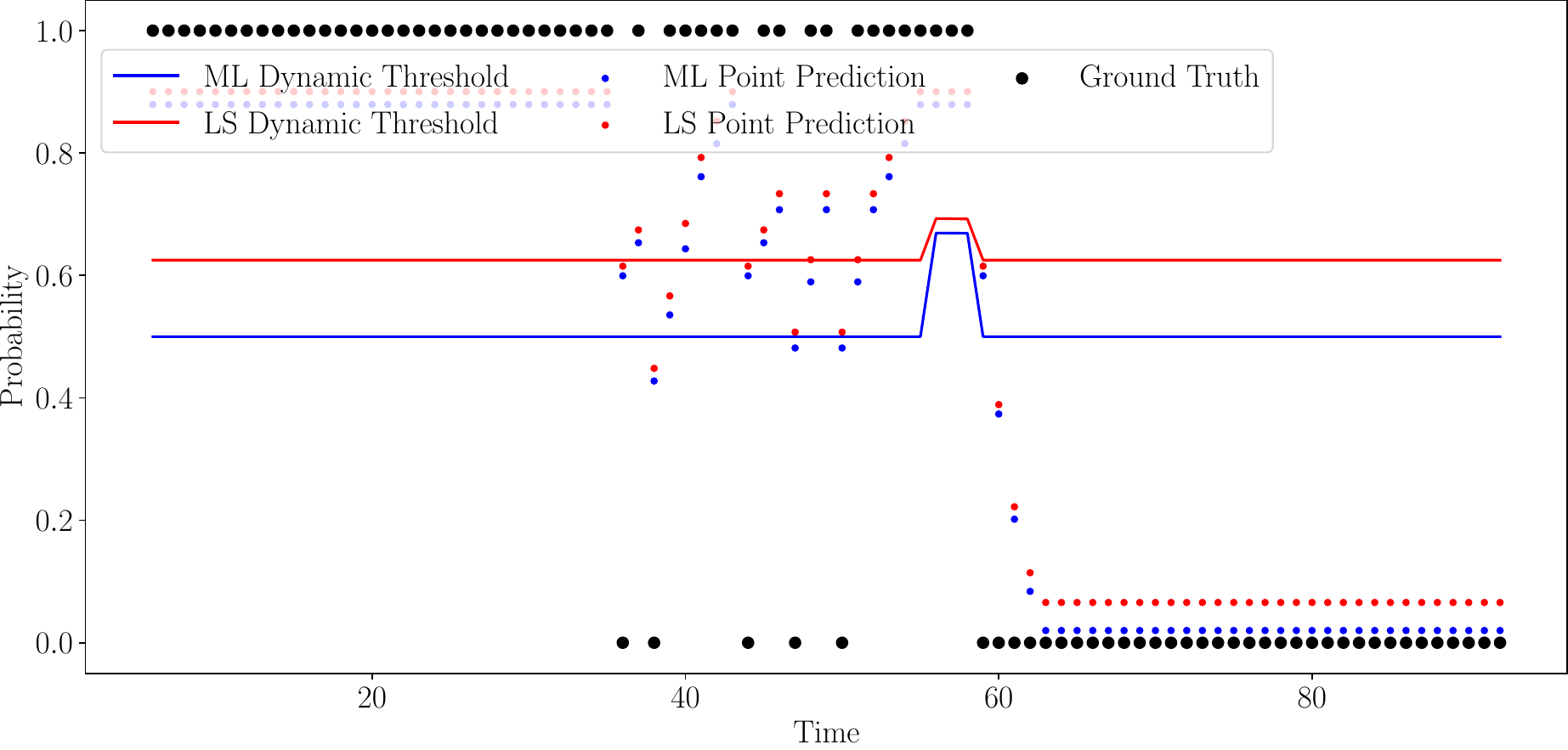}}}
    \subfigure[July-Sept (Static)]{\includegraphics[scale=0.21]{{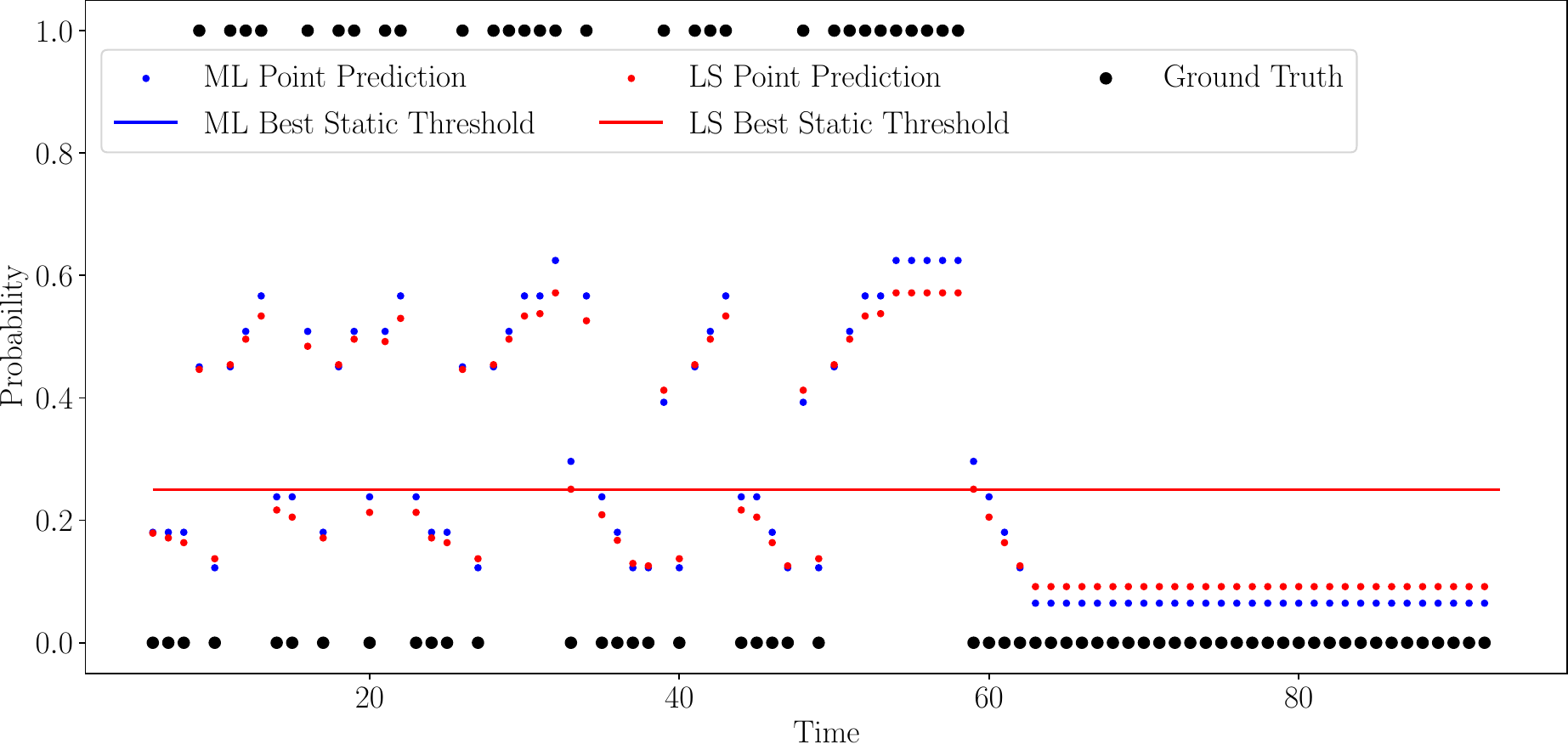}}}
    \subfigure[July-Sept (Dynamic)]{\includegraphics[scale=0.21]{{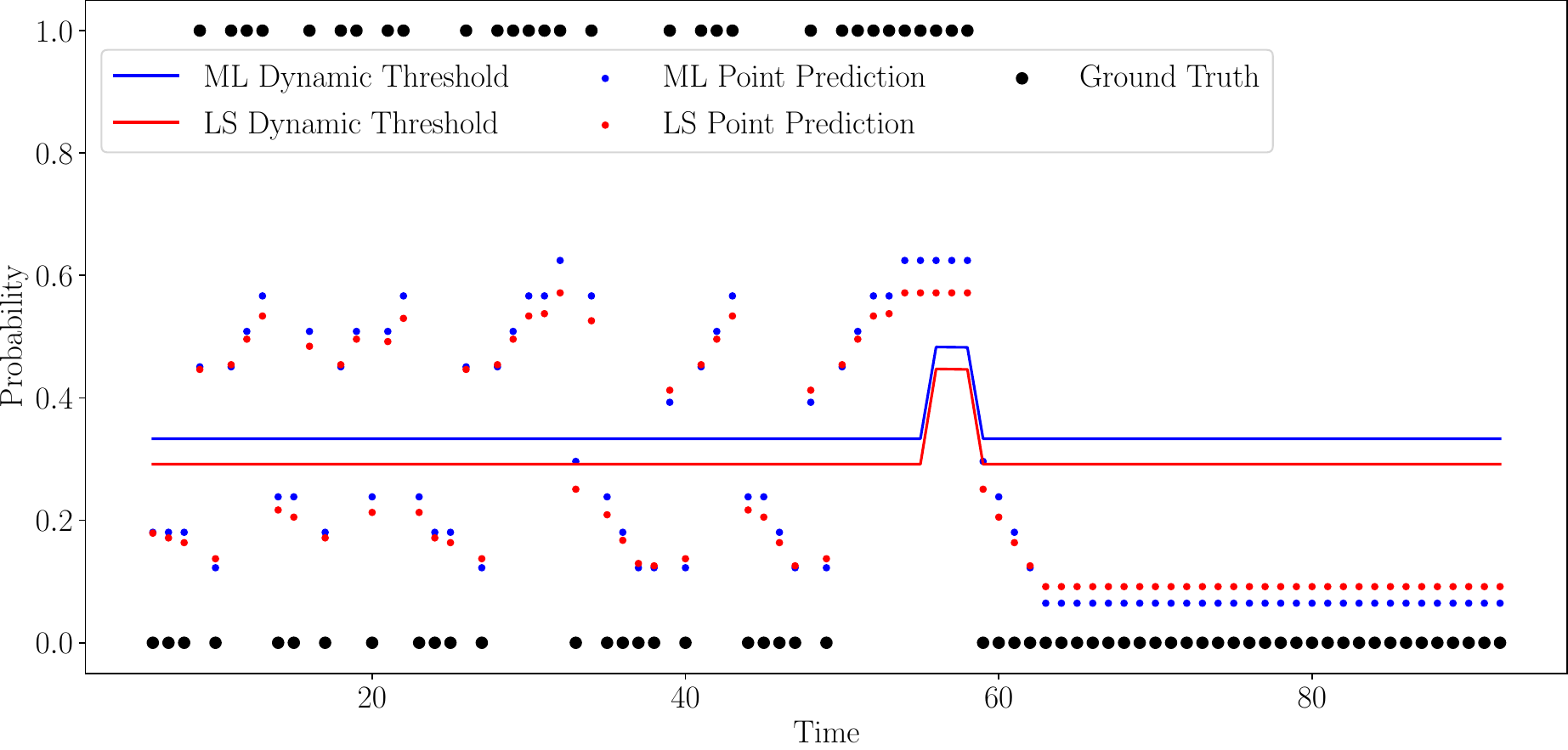}}}
    \vspace{-0.3cm}
    \caption{Online point prediction of probabilities for single-state ramping events using LS (red dots) and ML (blue dots), compared with true ramping events (black dots). The left column in each row is generated using a static threshold, and the right column in each row uses dynamic thresholds. }
    \label{fig:5_2_point_pred_rainfall}
\end{figure*}

\clearpage

\noindent \textit{2. Multi-state results}

\begin{figure*}[htbp]
  \centering
  \subfigure[$s=1$]{\includegraphics[scale=0.37]{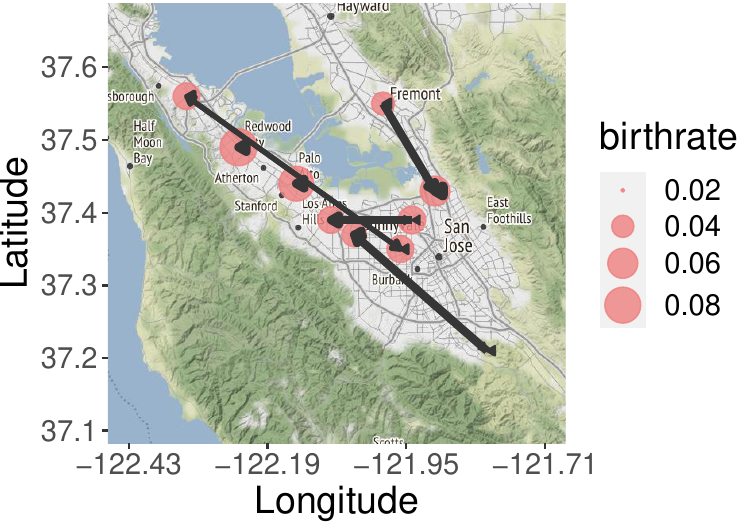}}
  \subfigure[$s=3$]{\includegraphics[scale=0.37]{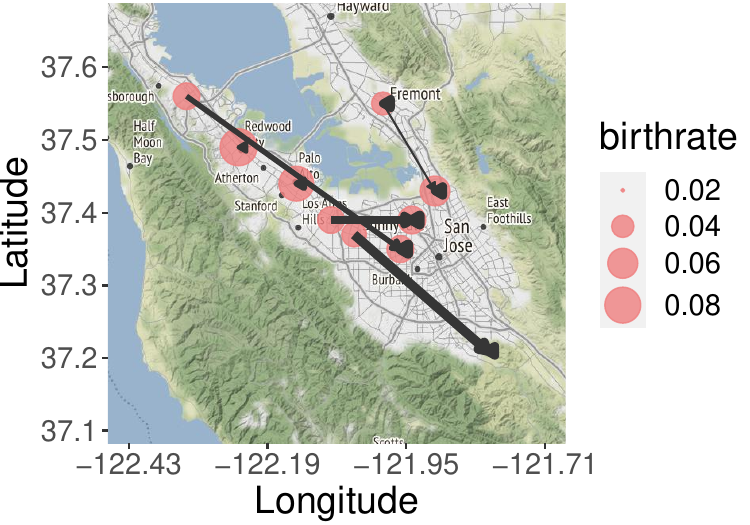}}
  \subfigure[$s=5$]{\includegraphics[scale=0.37]{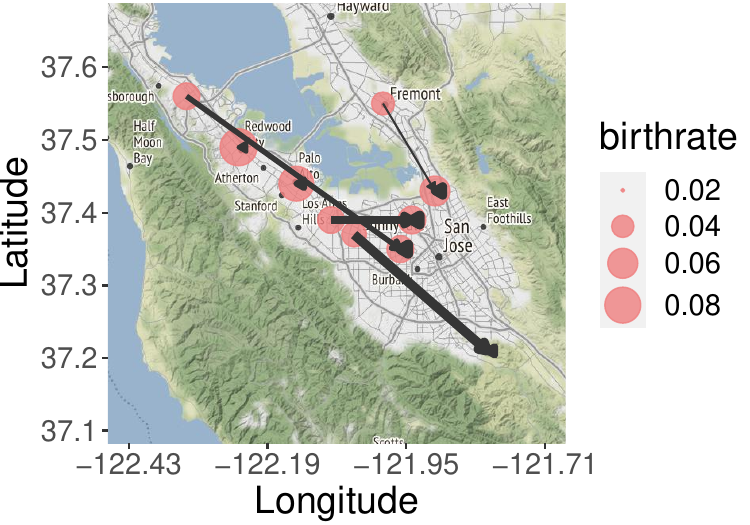}}
  \subfigure[$s=1$]{\includegraphics[scale=0.37]{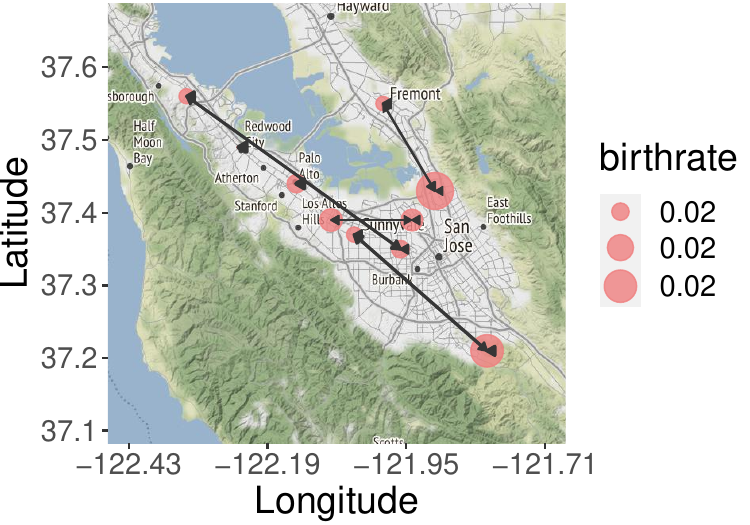}}
  \subfigure[$s=3$]{\includegraphics[scale=0.37]{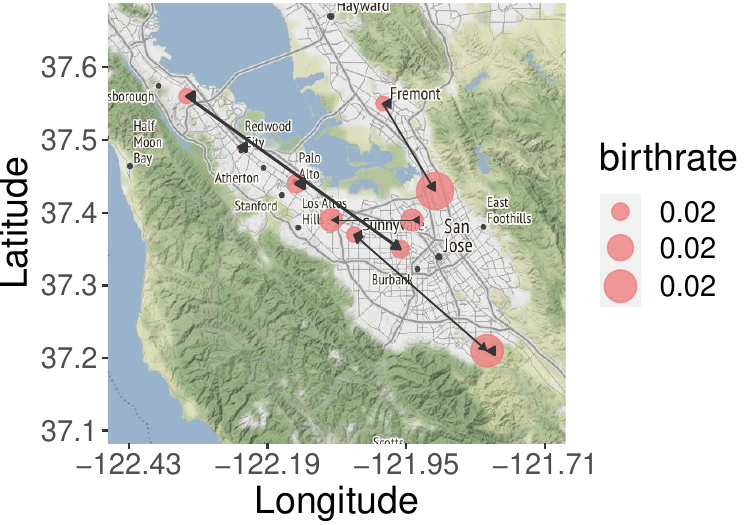}}
  \subfigure[$s=5$]{\includegraphics[scale=0.37]{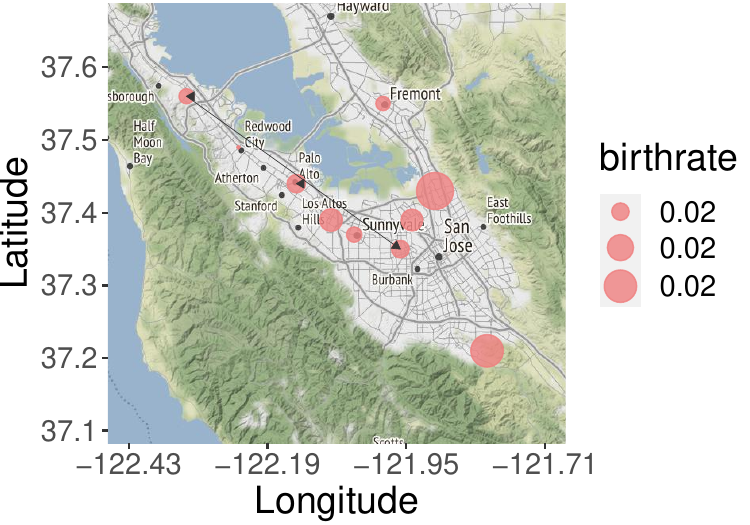}}
  \caption{Visualize multi-state, positive-to-positive spatio-temporal influences. The top row shows estimates from data from 01/01/2017-03/31-2017, and the bottom row uses data during 07/01/2017-09/30-2017. To make sure the edges are visible when $s>1$, we magnify the edge weight of North California (g-i) estimates five times.}
  \label{fig:5_2_grid_p2p_CA}
\end{figure*}

\begin{figure*}[htbp]
  \centering
  \subfigure[$s=1$]{\includegraphics[scale=0.37]{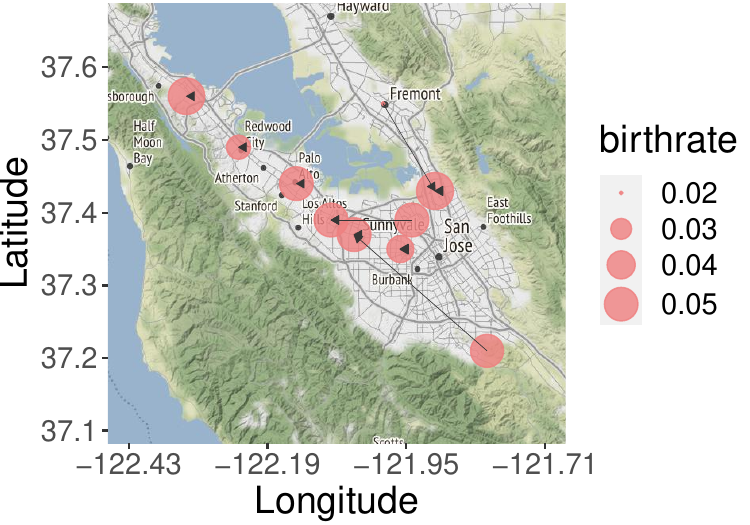}}
  \subfigure[$s=3$]{\includegraphics[scale=0.37]{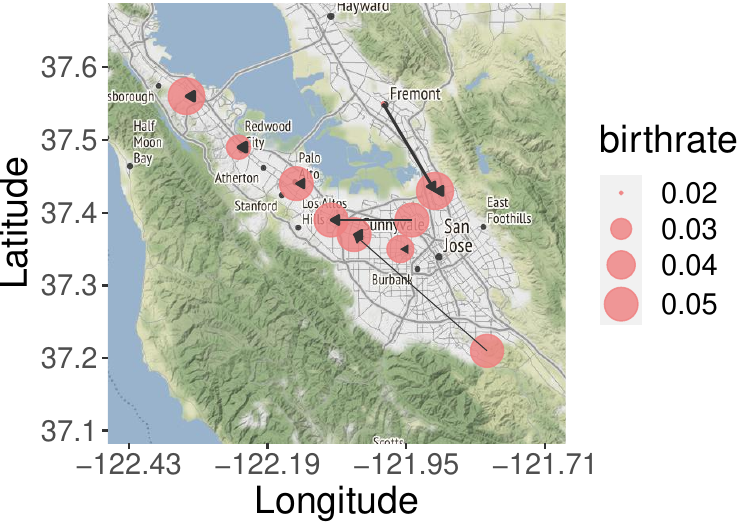}}
  \subfigure[$s=5$]{\includegraphics[scale=0.37]{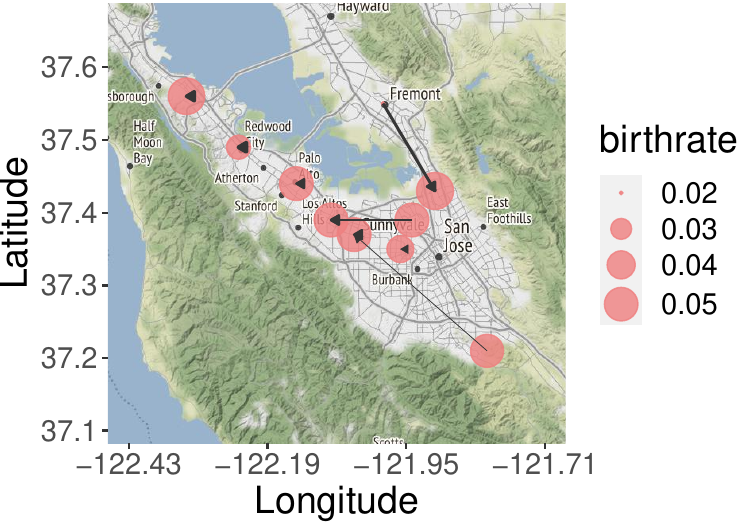}}
  \subfigure[$s=1$]{\includegraphics[scale=0.37]{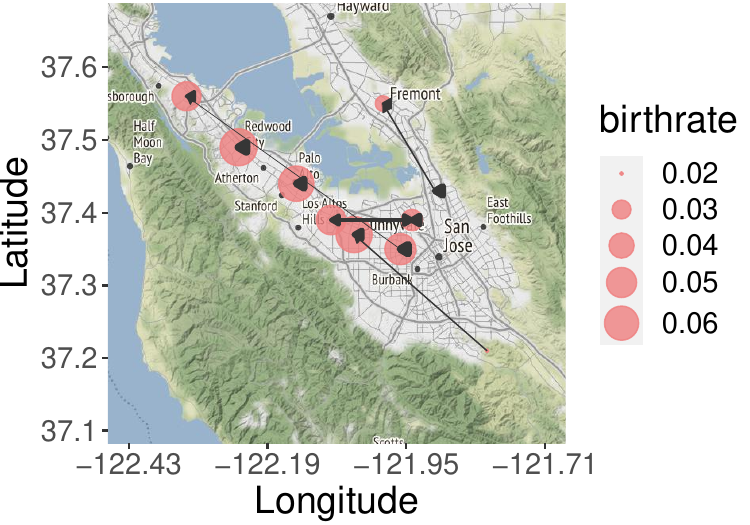}}
  \subfigure[$s=3$]{\includegraphics[scale=0.37]{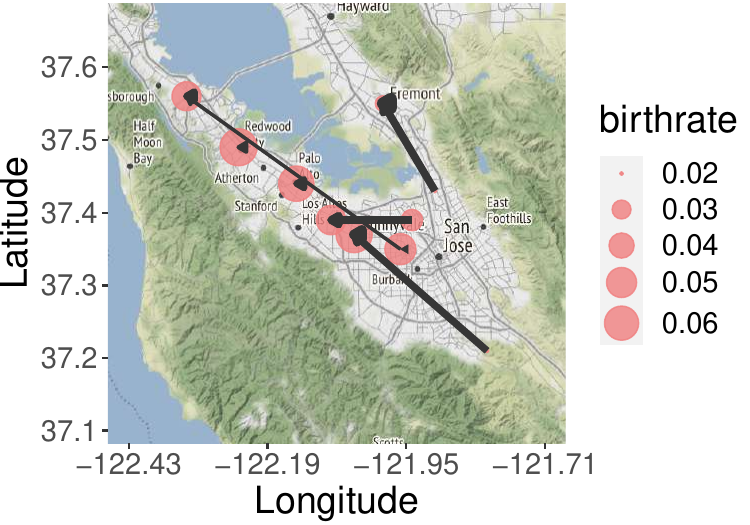}}
  \subfigure[$s=5$]{\includegraphics[scale=0.37]{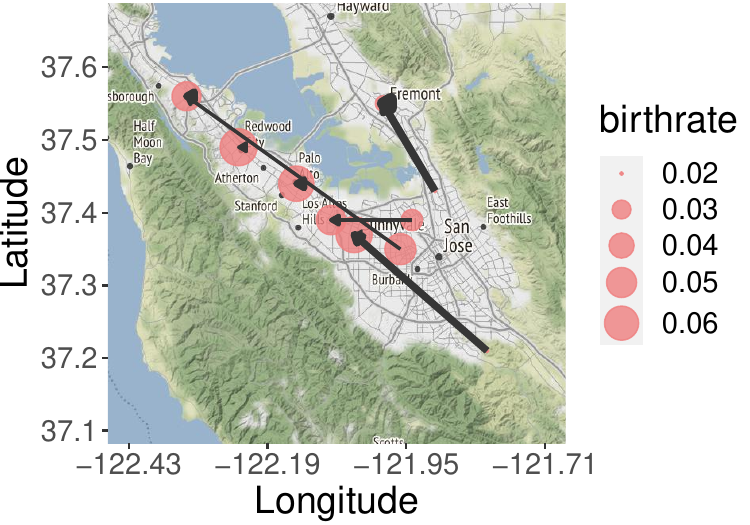}}
  \caption{Visualize multi-state, negative-to-negative spatio-temporal influences. The top row shows estimates from data from 01/01/2017-03/31-2017, and the bottom row uses data during 07/01/2017-09/30-2017. To ensure the edges are visible when $s>1$, we magnify the edge weight of North California (g-i) estimates five times.}
  \label{fig:5_2_grid_n2n_CA}
\end{figure*}

\begin{table}[htbp]
    \centering
    \caption{Sequential prediction performance for multi-state model by season: precision, recall, and $F_1$ score in Palo Alto under static vs. dynamic threshold after tuning. Some cells are marked as 0 because there is no ramping event of a given type during that season. The higher value between static vs. dynamic is in bold.}
    \label{tab:5_2_mod_accuracy_multi_rainfall}
    \resizebox{\linewidth}{!}{\begin{tabular}{p{2cm}p{2cm}p{2cm}p{2cm}p{2cm}p{2cm}p{2cm}p{2cm}}
     \hline
     &&\multicolumn{3}{c}{Least Square} &\multicolumn{3}{c}{Maximum Likelihood} \\
     Months & $\tau$ & Precision & Recall & $F_1$ & Precision & Recall & $F_1$\\
     \hline
\multirow{2}{*}{PP: Jan-Mar} & Static & 1.00 & 1.00 & 1.00 & 1.00 & \textbf{1.00} & \textbf{1.00} \\
 & Dynamic & 1.00 & 0.90 & 0.95 & \textbf{1.00} & 0.80 & 0.89 \\
\multirow{2}{*}{PP: July-Sept} & Static & 0.00 & 0.00 & 0.00 & 0.00 & 0.00 & 0.00 \\
 & Dynamic & 0.00 & 0.00 & 0.00 & 0.00 & 0.00 & 0.00 \\
\multirow{2}{*}{NN: Jan-Mar} & Static & 0.00 & 0.00 & 0.00 & 0.00 & 0.00 & 0.00 \\
 & Dynamic & 0.00 & 0.00 & 0.00 & 0.00 & 0.00 & 0.00 \\
\multirow{2}{*}{NN: July-Sept} & Static & 0.91 & 1.00 & 0.95 & 0.91 & \textbf{1.00} & 0.95 \\
 & Dynamic & 1.00 & 0.90 & \textbf{0.95} & \textbf{1.00} & 0.85 & 0.92 \\

     \hline
    \end{tabular}}
\end{table}

\begin{figure*}[htbp]
    \centering
    \subfigure[Jan-March (Static)]{\includegraphics[scale=0.21]{{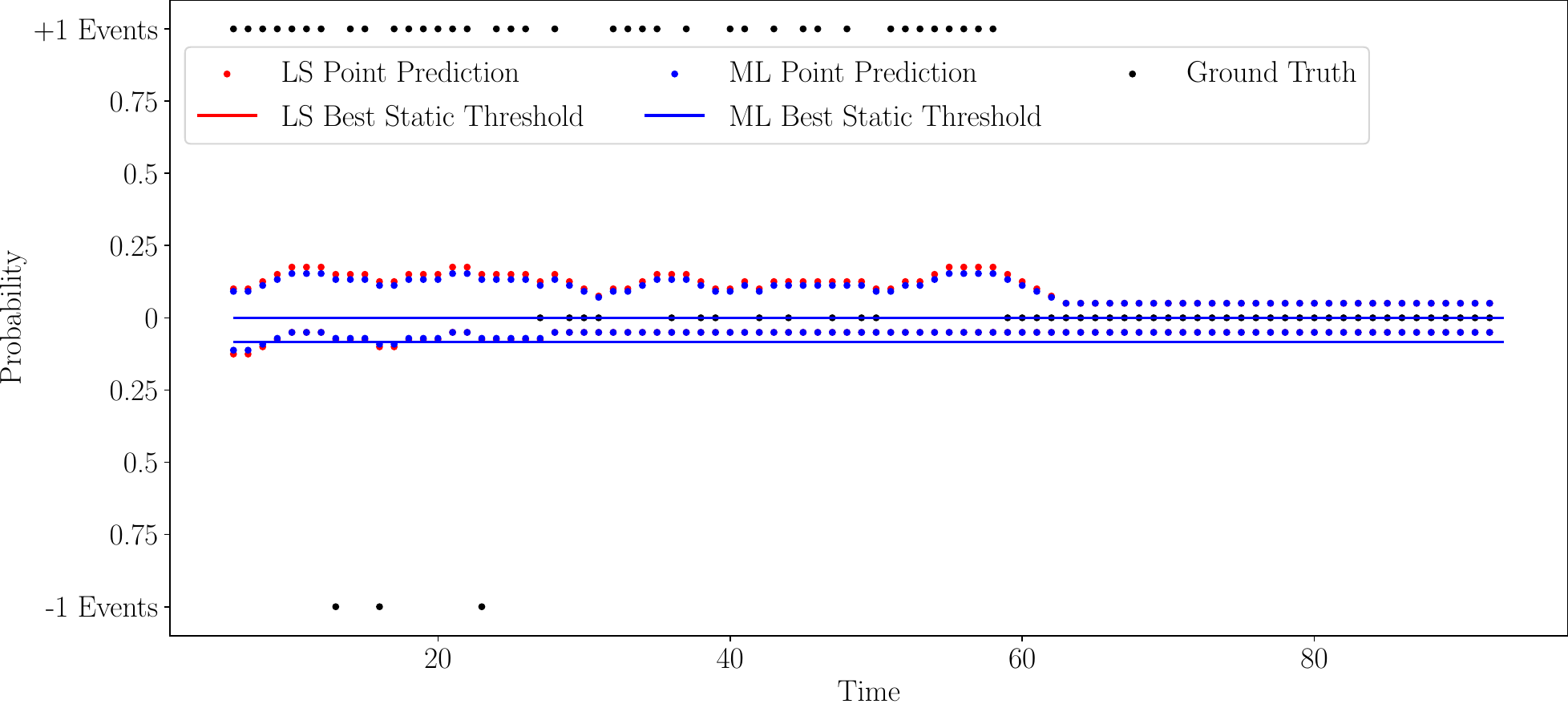}}}
    \subfigure[Jan-March (Dynamic)]{\includegraphics[scale=0.21]{{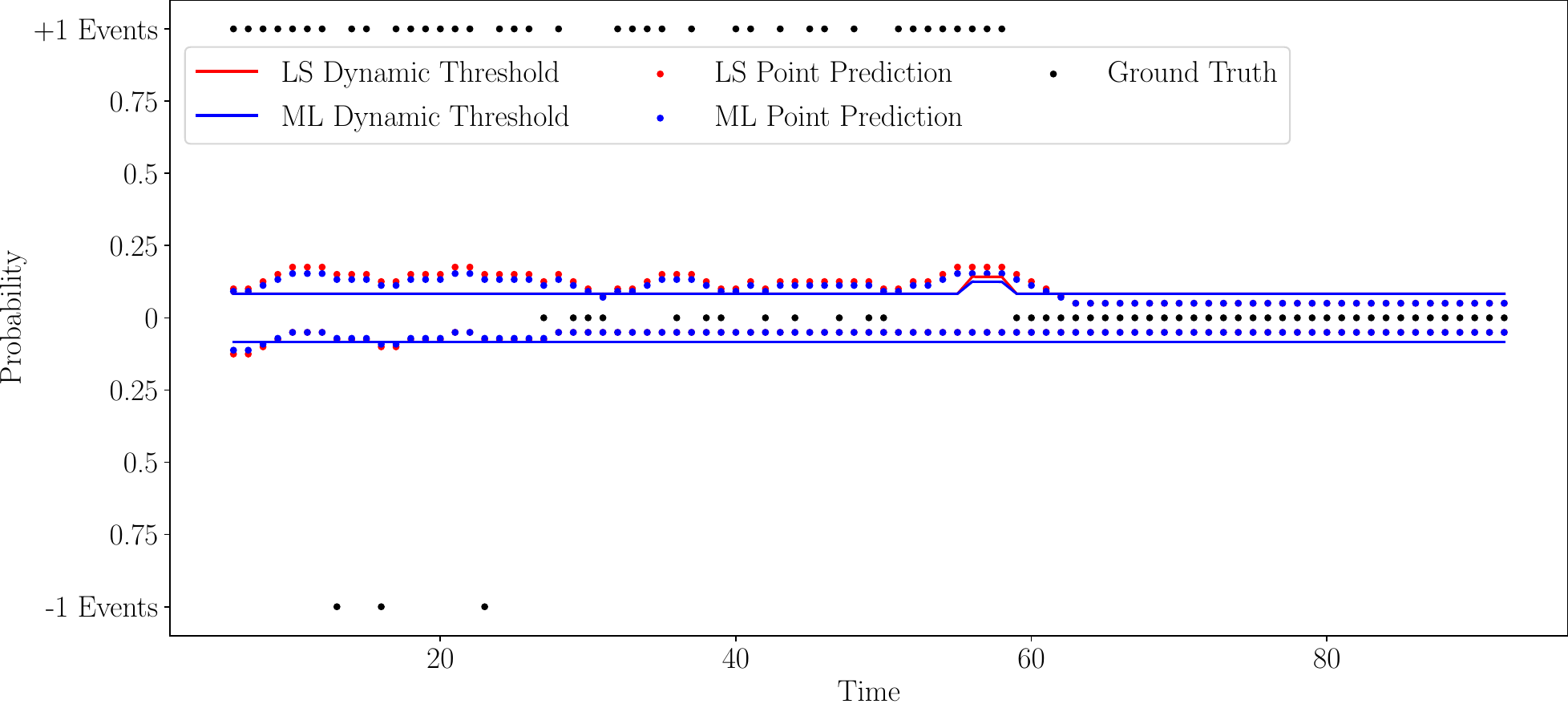}}}
    \subfigure[July-Sept (Static)]{\includegraphics[scale=0.21]{{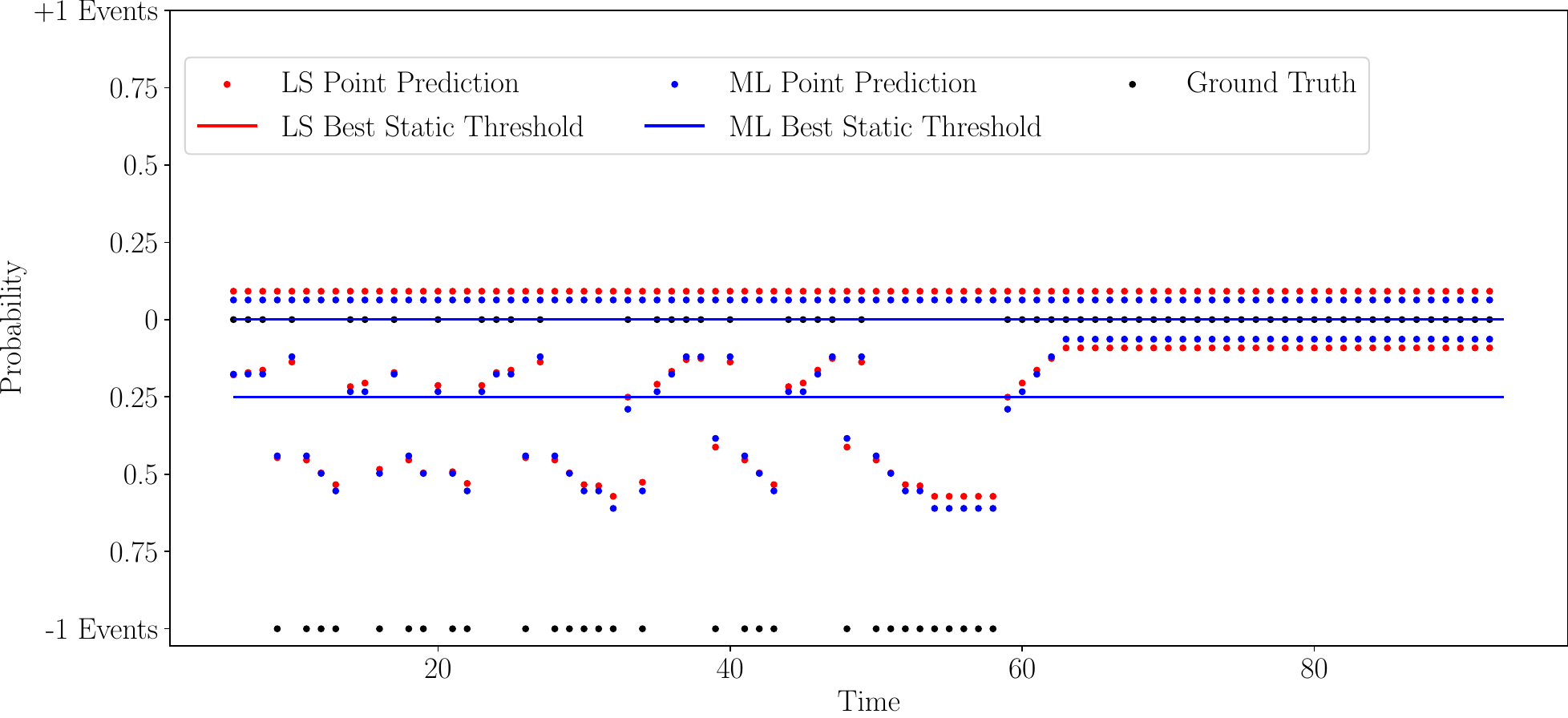}}}
    \subfigure[July-Sept (Dynamic)]{\includegraphics[scale=0.21]{{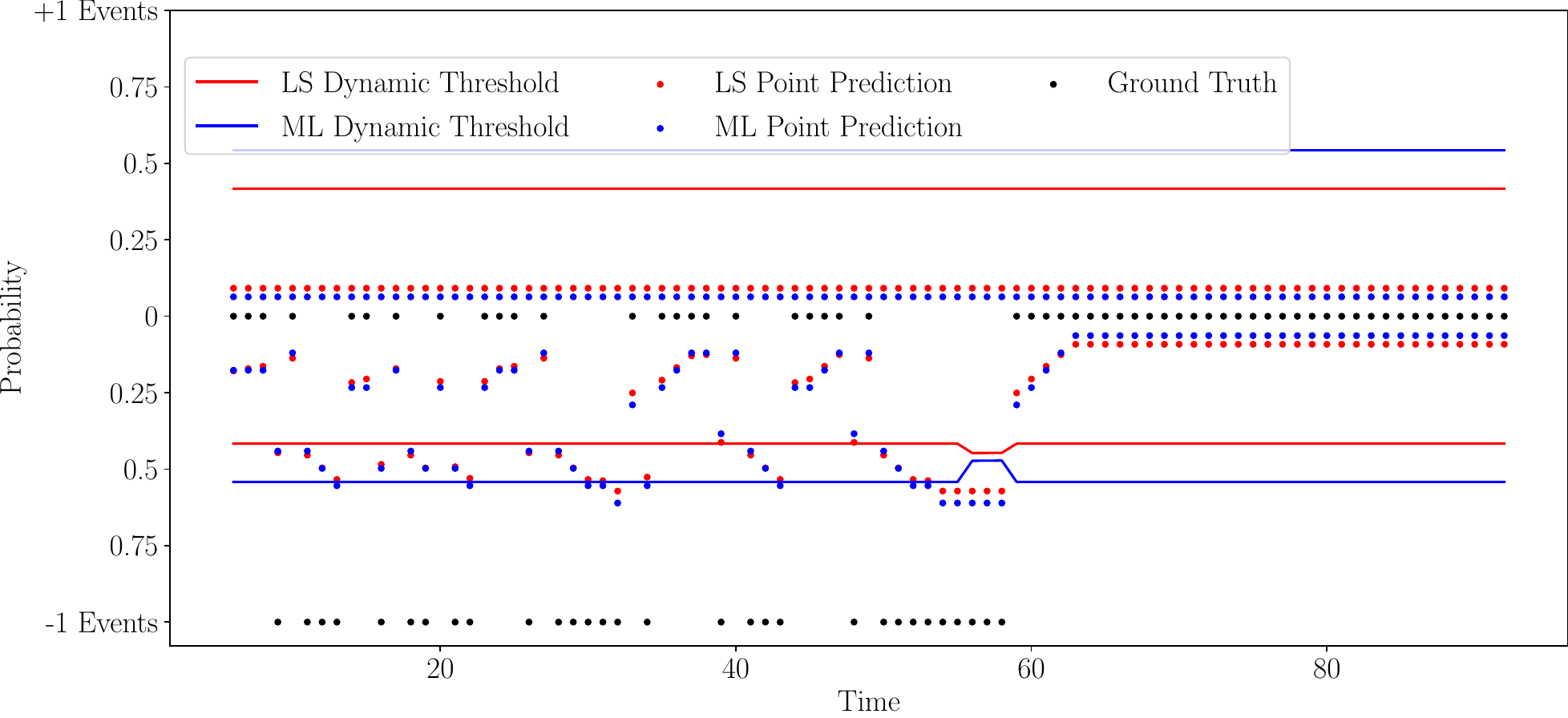}}}
    \vspace{-0.3cm}
    \caption{North California: Online point predictions of probabilities for multi-state ramping events using LS (red dots) and ML (blue dots), compared with true ramping events (black dots). The left column in each row is generated using a static threshold, and the right column uses dynamic thresholds.}
    \label{fig:5_2_point_pred_multi_rainfall}
\end{figure*}

\subsection{Visualizing interactions via LS methods on terrain map} \label{app:LS_estimates} As we mentioned earlier, these parameters by LS are visually similar to those by ML. We only plot the spatio-temporal interactions on the map.

\begin{figure*}[htbp]
  \centering
  \subfigure[$s=1$]{\includegraphics[scale=0.41]{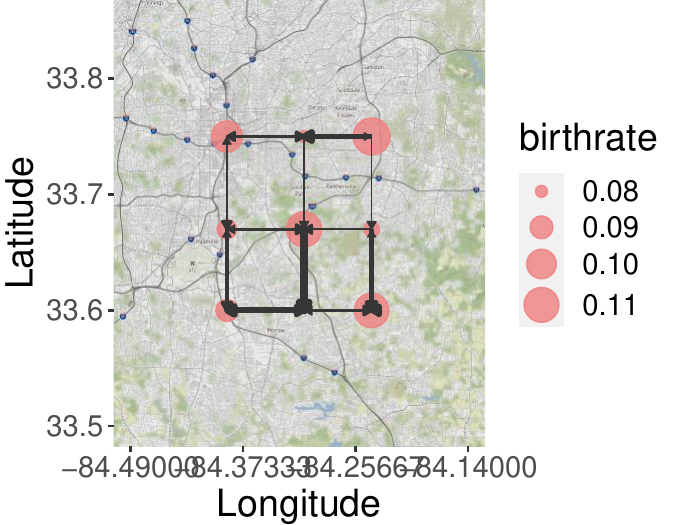}}
  \subfigure[$s=5$]{\includegraphics[scale=0.41]{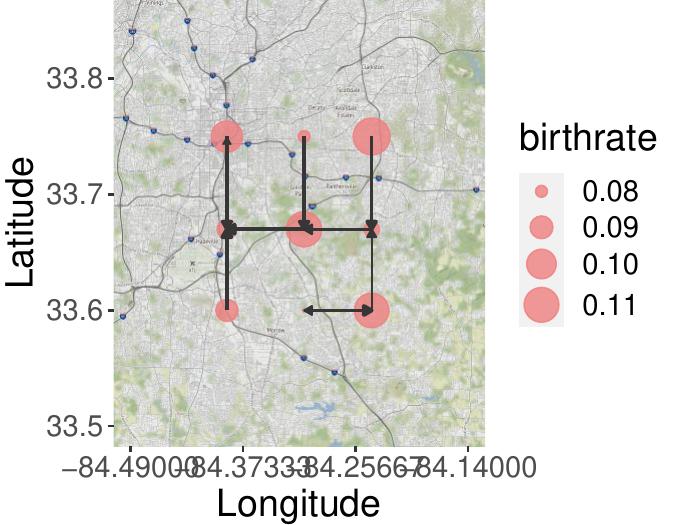}}
  \subfigure[$s=10$]{\includegraphics[scale=0.41]{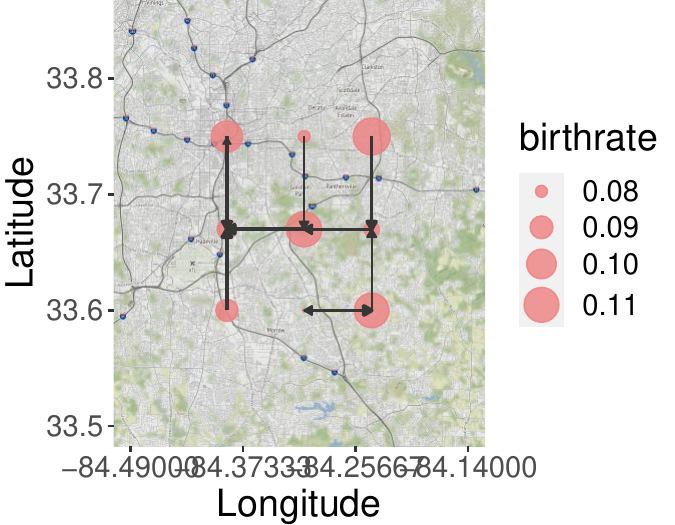}}
  \subfigure[$s=1$]{\includegraphics[scale=0.41]{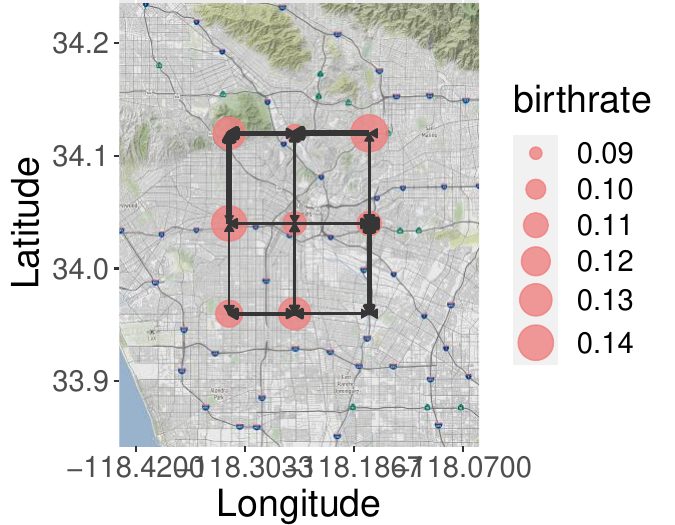}}
  \subfigure[$s=3$]{\includegraphics[scale=0.41]{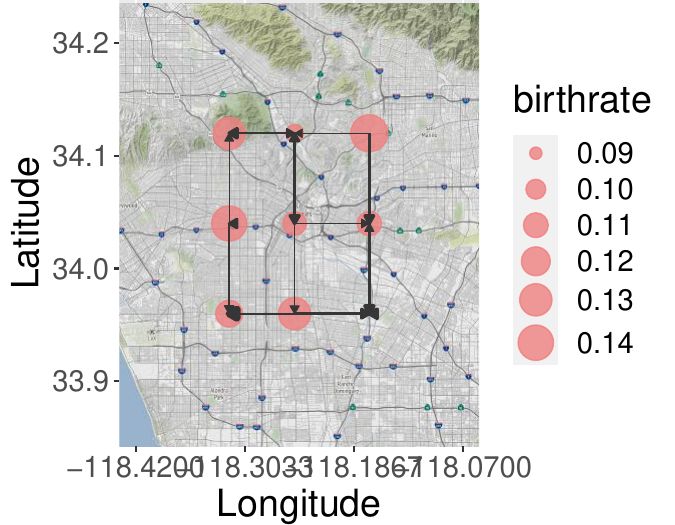}}
  \subfigure[$s=5$]{\includegraphics[scale=0.41]{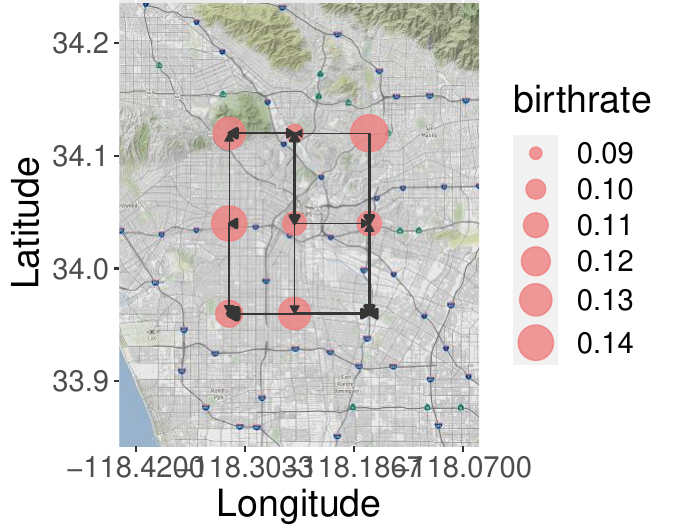}}
  \subfigure[$s=1$]{\includegraphics[scale=0.37]{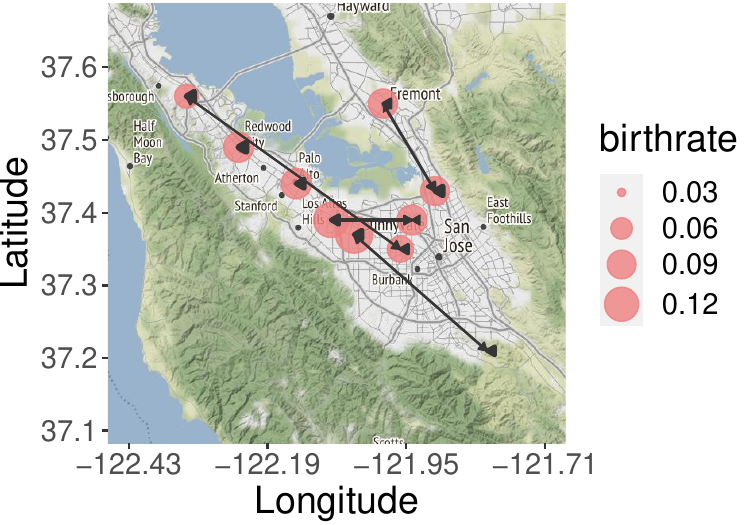}}
  \subfigure[$s=5$]{\includegraphics[scale=0.37]{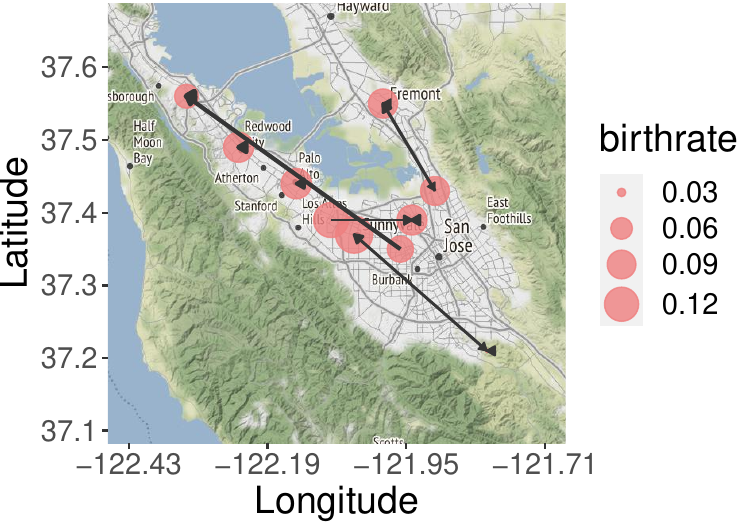}}
  \subfigure[$s=10$]{\includegraphics[scale=0.37]{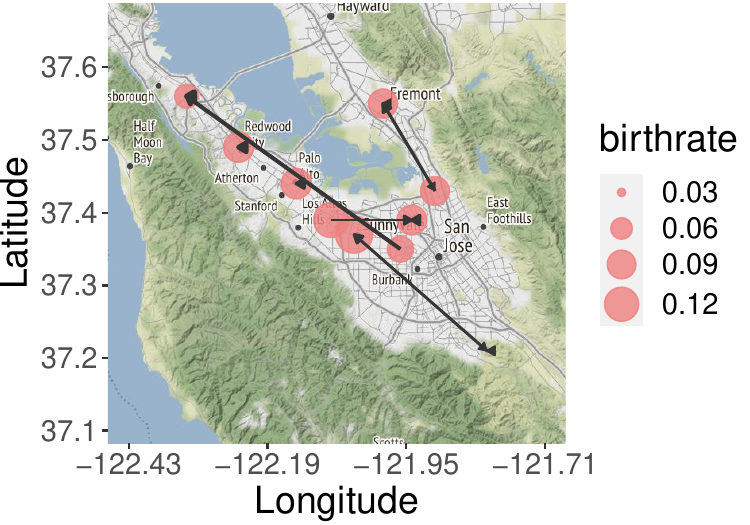}}
  \caption{Visualize single-state spatio-temporal interactions using the LS method. To make sure the edges are visible when $s>1$, we magnify the edge weight of Atlanta (a-c) estimates five times, of Los Angeles (d-f) estimates three times, and of North California (g-i) estimates five times.}
  \label{fig:5_2_grid_single_LS}
\end{figure*}

\begin{figure*}[htbp]
  \centering
  \subfigure[$s=1$]{\includegraphics[scale=0.41]{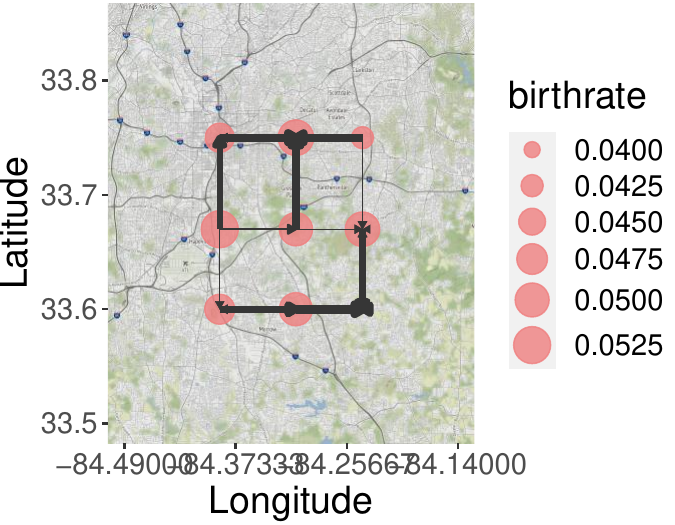}}
  \subfigure[$s=5$]{\includegraphics[scale=0.41]{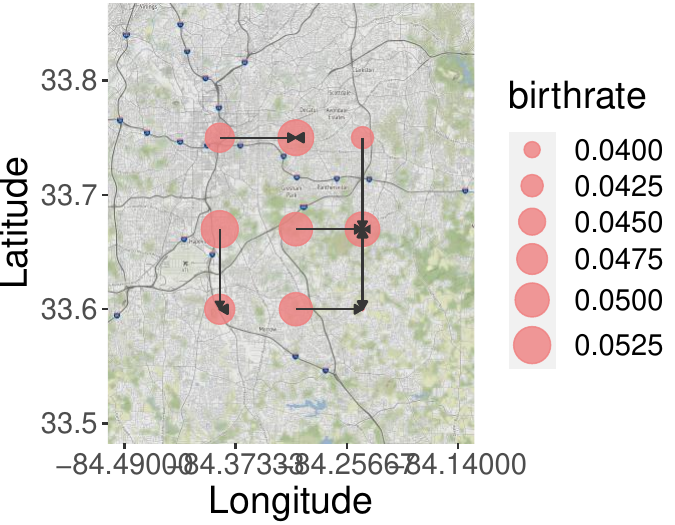}}
  \subfigure[$s=10$]{\includegraphics[scale=0.41]{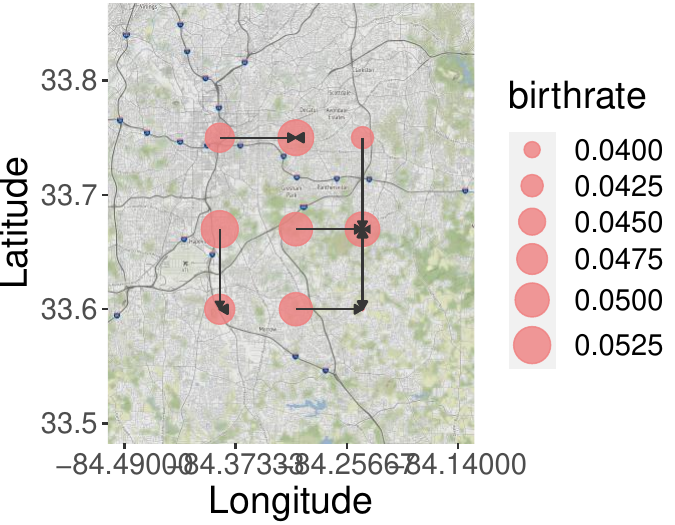}}

  \subfigure[$s=1$]{\includegraphics[scale=0.41]{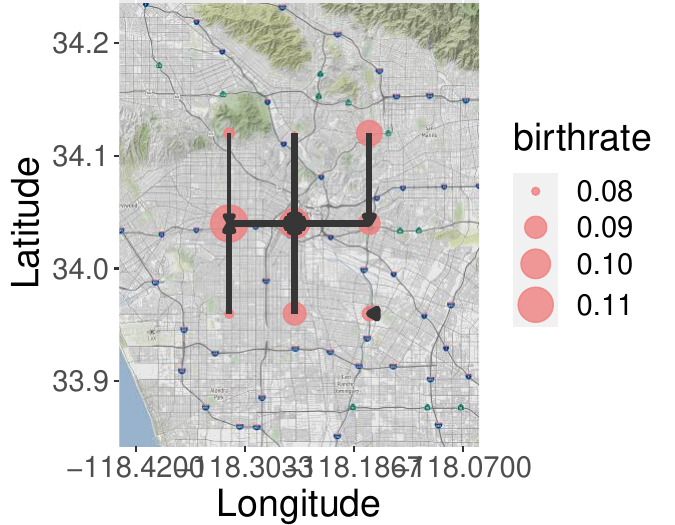}}
  \subfigure[$s=3$]{\includegraphics[scale=0.41]{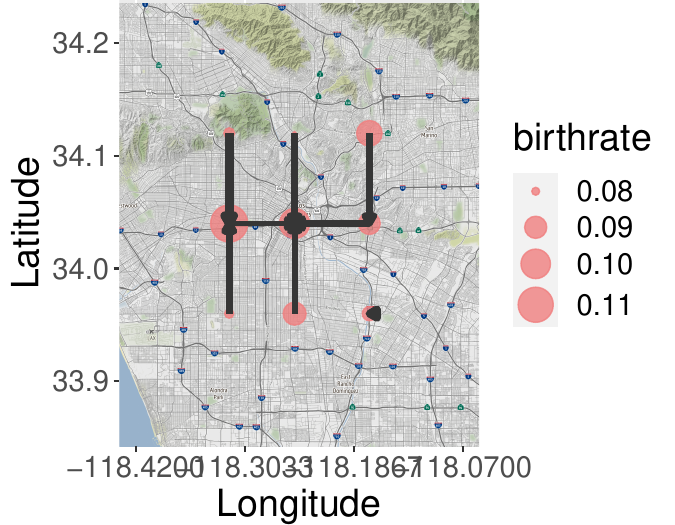}}
  \subfigure[$s=5$]{\includegraphics[scale=0.41]{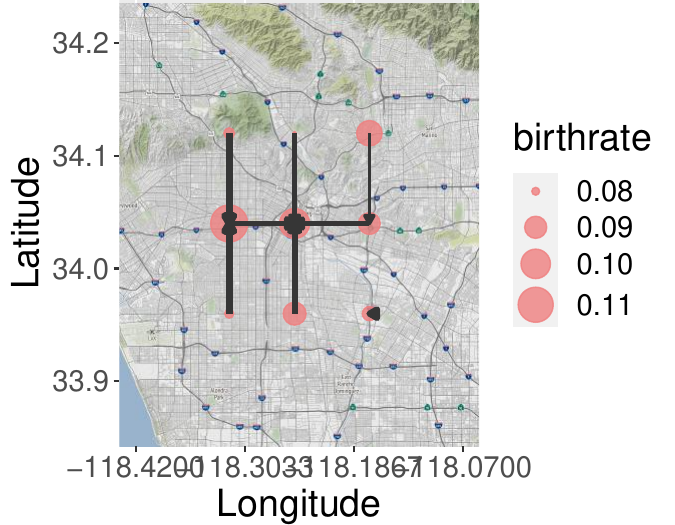}}

  \subfigure[$s=1$]{\includegraphics[scale=0.37]{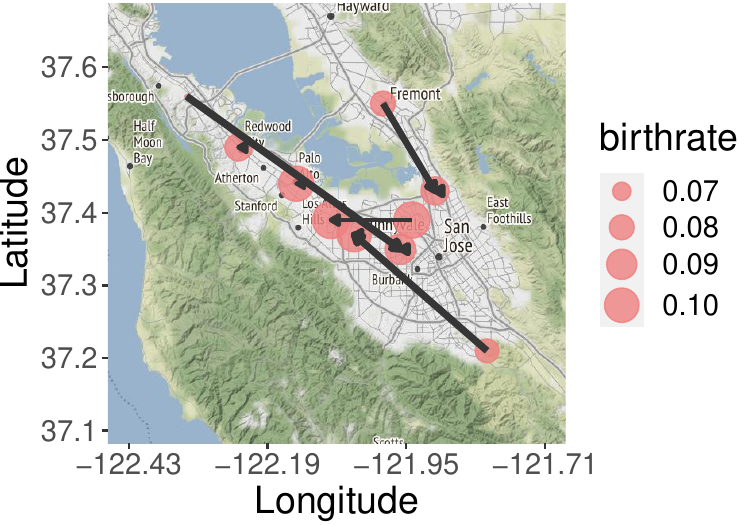}}
  \subfigure[$s=5$]{\includegraphics[scale=0.37]{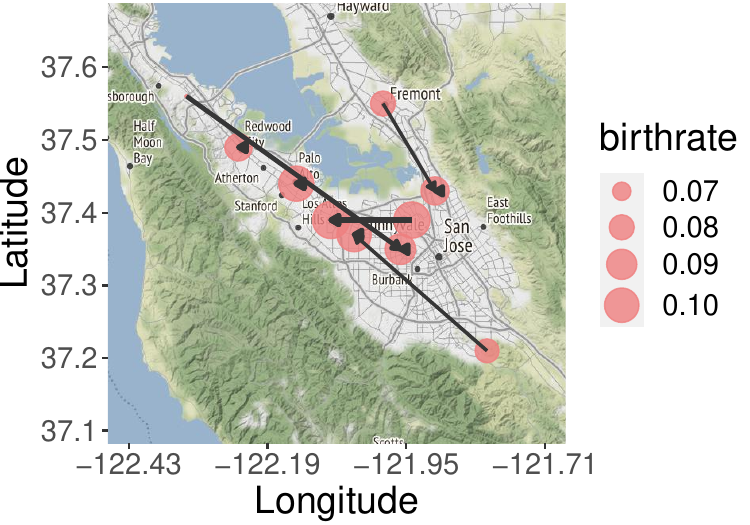}}
  \subfigure[$s=10$]{\includegraphics[scale=0.37]{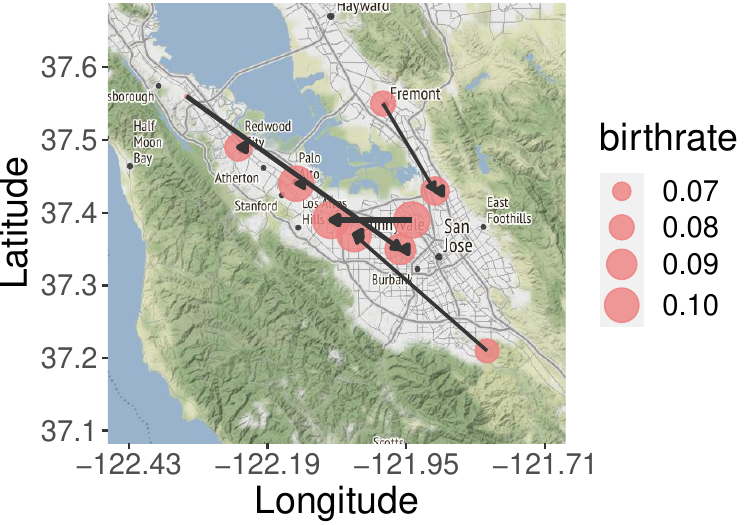}}
    \caption{Visualize multi-state, positive-to-positive spatio-temporal interactions. To make sure the edges are visible when $s>1$, we magnify the edge weight of Atlanta (a-c) and Los Angeles (d-f) estimates three times and North California (g-i) estimates five times.}
\end{figure*}

\begin{figure*}[htbp]
  \centering
  \label{fig:5_2_grid_multi_neg_LS}
  
  \subfigure[$s=1$]{\includegraphics[scale=0.41]{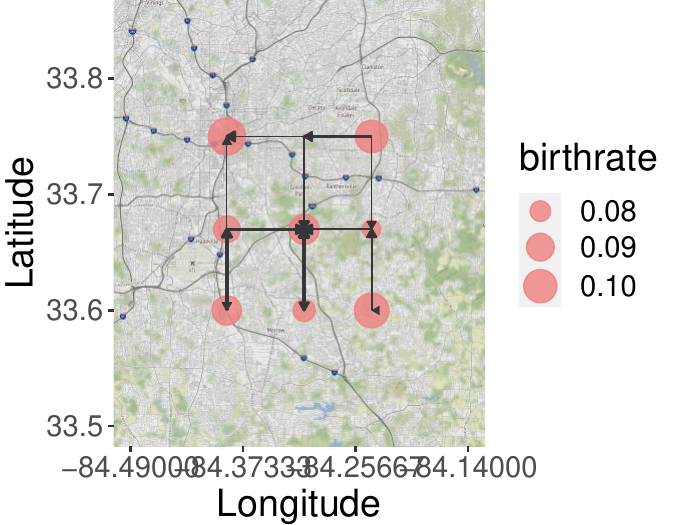}}
  \subfigure[$s=5$]{\includegraphics[scale=0.41]{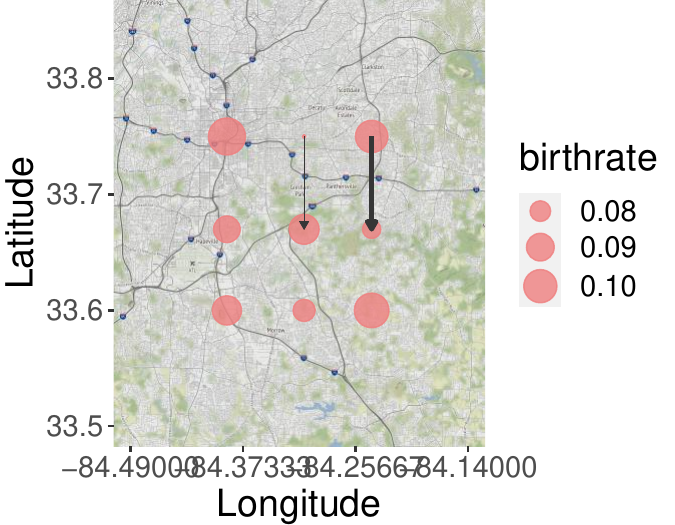}}
  \subfigure[$s=10$]{\includegraphics[scale=0.41]{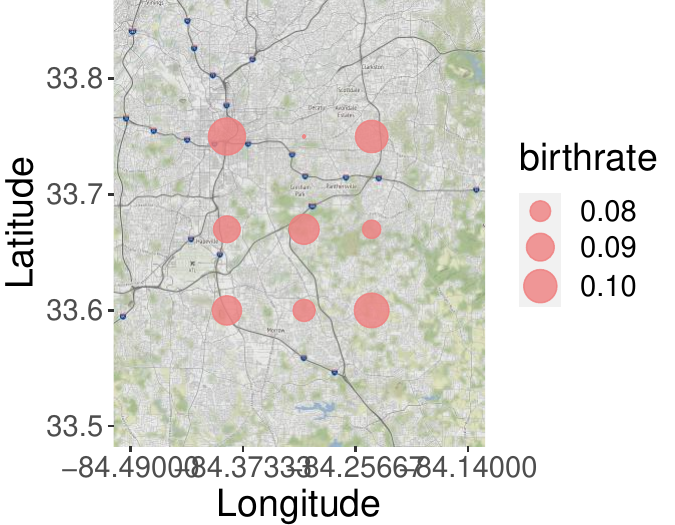}}

  \subfigure[$s=1$]{\includegraphics[scale=0.41]{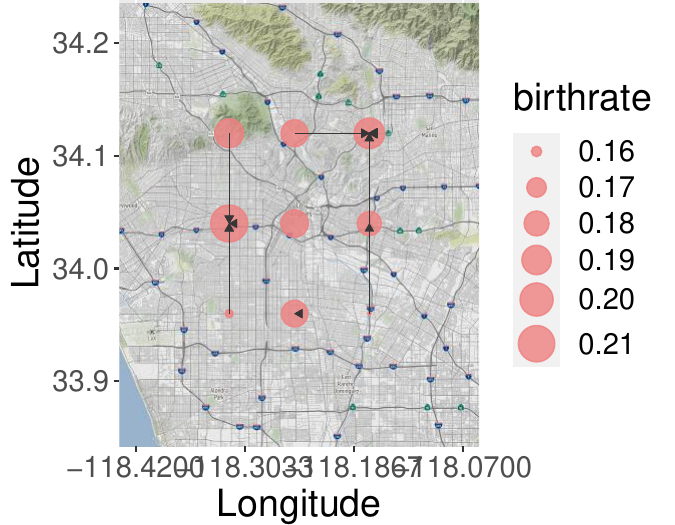}}
  \subfigure[$s=3$]{\includegraphics[scale=0.41]{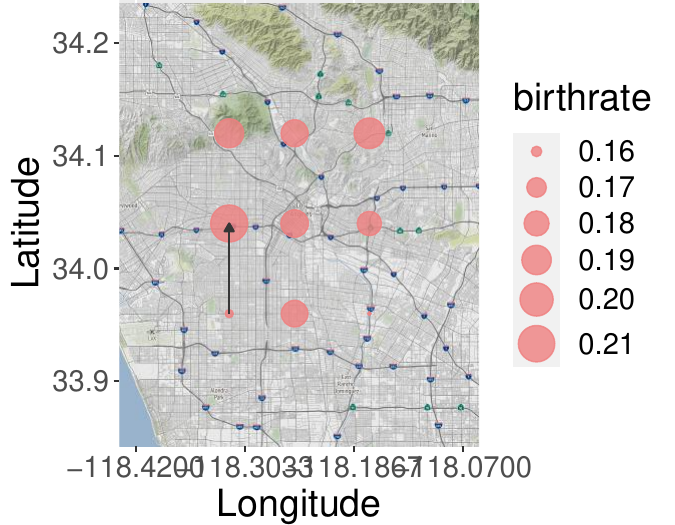}}
  \subfigure[$s=5$]{\includegraphics[scale=0.41]{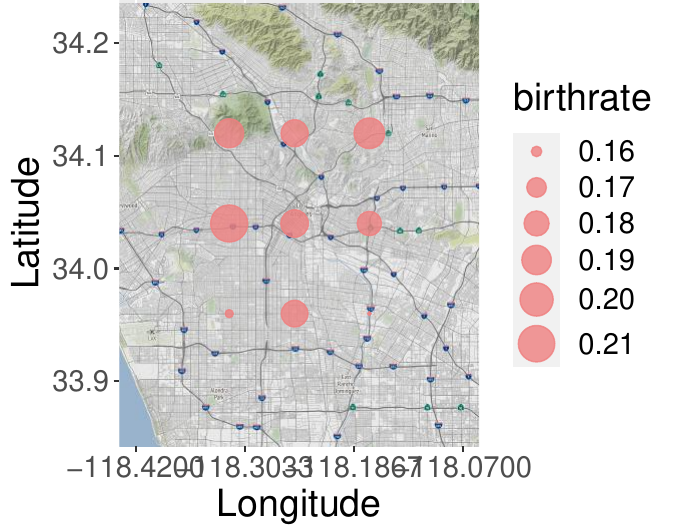}}
  
  \subfigure[$s=1$]{\includegraphics[scale=0.37]{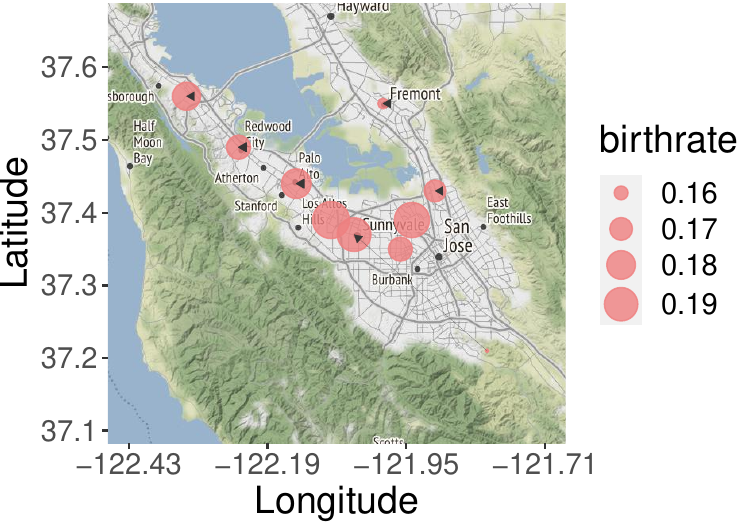}}
  \subfigure[$s=5$]{\includegraphics[scale=0.37]{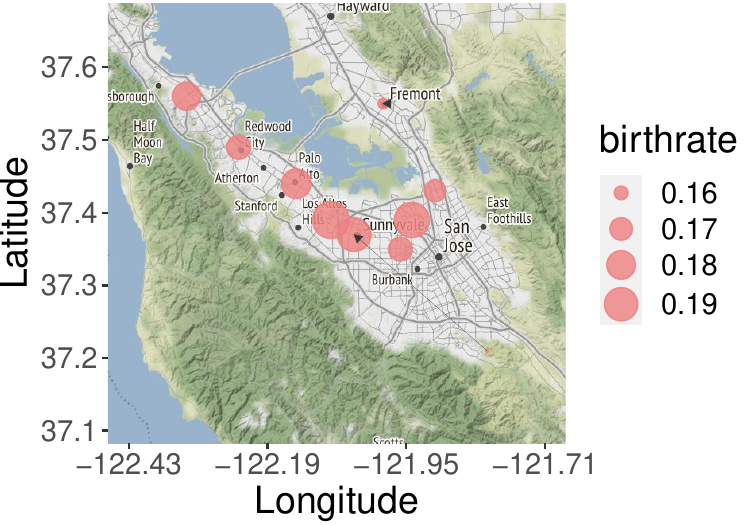}}
  \subfigure[$s=10$]{\includegraphics[scale=0.37]{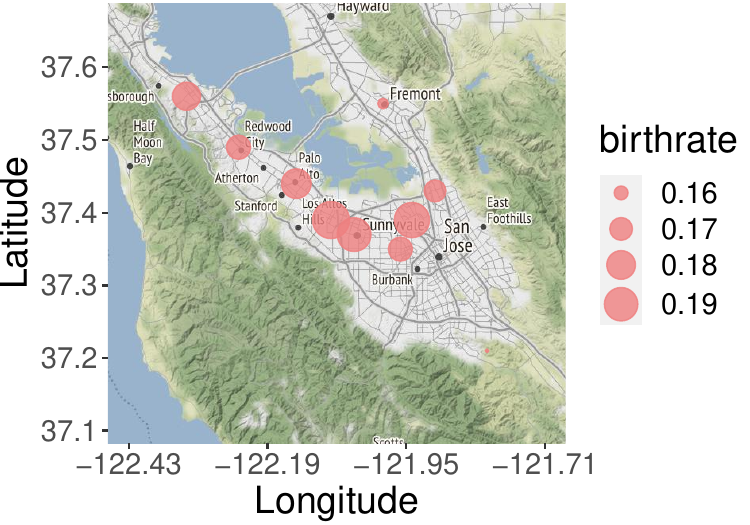}}
    \caption{Visualize multi-state, negative-to-negative spatio-temporal interactions. To make sure the edges are visible when $s>1$, we magnify the edge weight of Atlanta (a-c) estimates 25 times, North California (g-i) estimates five times, and Los Angeles (d-f) estimates six times.}
\end{figure*}


\end{appendix}


\end{document}